\documentclass[11pt]{article}
\usepackage[top=60pt,bottom=80pt,left=80pt,right=80pt]{geometry}

\usepackage{authblk}

\usepackage{natbib}

\usepackage[utf8]{inputenc} 
\usepackage[T1]{fontenc}    
\usepackage{hyperref}       
\usepackage{url}            
\usepackage{booktabs}       
\usepackage{amsfonts}       
\usepackage{nicefrac}       
\usepackage{microtype}      
\usepackage{xcolor}         
\usepackage[pdf]{graphviz} 

\usepackage{amsmath}
\usepackage{cryptocode}

\usepackage{amsthm}

\usepackage{thmtools,thm-restate}
\usepackage{comment}

\usepackage{thmtools}
\usepackage{mathtools}
\usepackage{enumerate}

\usepackage{amsfonts}

\usepackage{amssymb}
\usepackage{amsmath}

\usepackage{algorithm, algorithmic}

\usepackage{dirtytalk}

\usepackage{subfig}
\usepackage{booktabs}
\usepackage{siunitx}

\usepackage{datetime}

\newdateformat{monthdayyeardate}{%
  \monthname[\THEMONTH]~\THEDAY, \THEYEAR}

\usepackage{bbm}
\usepackage{ dsfont }
\usepackage{wrapfig}

\usepackage{tikz}

\usepackage{thm-restate}

\usepackage[full]{complexity}
\newcommand{\ENT}{\mathsf{ENT}}

\newcommand{\NEL}{\mathsf{NeL}}

\newcommand{\SampBQP}{\mathsf{SampBQP}}

\newcommand{\ALL}{\mathsf{ALL}}

\usetikzlibrary{quantikz2}
 
\newcommand{\N}{\mathbb{N}}
\newcommand{\Z}{\mathbb{Z}}
\newcommand{\Real}{\mathbbm{R}}
\newcommand{\Com}{\mathbbm{C}}

\newcommand{\cF}{\mathcal{F}}
\newcommand{\cG}{\mathcal{G}}
\newcommand{\cH}{\mathcal{H}}

\newcommand{\CHK}{\mathsf{CHK}}
\newtheorem{theorem}{Theorem}
\newtheorem{lemma}{Lemma}

\newtheorem{conjecture}{Conjecture}

\newtheorem{fact}{Fact}
\newtheorem{remark}{Remark}
\newtheorem{note}{Note}

\newtheorem{definition}{Definition}

\newcommand{\e}{\epsilon}
\newcommand{\nbits}{\{0,1\}^n}

\newcommand{\ra}[1]{\renewcommand{\arraystretch}{#1}}

\newcommand{\dist}{\mathcal{D}}

\newcommand{\id}{\mathbb{I}}
\newcommand{\adv}{\mathbf{A}}

\newcommand{\Prob}{\mathbb{P}}
\newcommand{\ver}{\mathbf{V}}
\newcommand{\prov}{\mathbf{P}}

\newcommand{\trace}{\textnormal{Tr}}

\newcommand{\negl}{\textnormal{negl}}
\newcommand{\Cdot}{C^{<\cdot,\cdot>}}
\newcommand{\re}{\text{Re}}
\newcommand{\im}{\text{Im}}
\newcommand{\hilb}{\mathcal{H}}
\newcommand{\qpt}{\text{QPT}}

\newcommand{\postqpt}{\PostBQP}
\newcommand{\XYm}{O_-}
\newcommand{\XYp}{O_+}
\newcommand{\lan}{\hat{\lambda}_n}

\newcommand{\alice}{\mathbf{A}}
\newcommand{\bob}{\mathbf{B}}

\newcommand{\zero}{\ket{0}}
\newcommand{\one}{\ket{1}}
\newcommand{\plus}{\ket{+}}
\newcommand{\minus}{\ket{-}}
\newcommand{\ii}{\ket{i}}
\newcommand{\minusi}{\ket{-i}}

\newcommand{\marco}[1]{\textcolor{magenta}{marco: #1}}

\newlength\myindent
\usepackage{todonotes}

\newcounter{mycomment}
\newcommand{\comm}[2]{
\refstepcounter{mycomment}
{%
    \todo[author = \textbf{#1~\#~\themycomment}, color={red!100!green!35}, fancyline, size = \footnotesize]{%
        #2}%
    }
}

\newcommand{\rom}[1]
    {\MakeUppercase{\romannumeral #1}}

\def\hrulefilltwo{\leavevmode\leaders\hrule height 0.7pt\hfill\kern}

\title{Nonlocality under Computational Assumptions}

\author[1]{Khashayar Barooti}
\affil[1]{AZTEC LABS}

\author[2]{Alexandru Gheorghiu}
\affil[2]{Chalmers University of Technology}

\author[3]{Grzegorz Głuch\footnote{Correspondence to \texttt{grzegorz.gluch@epfl.ch}}} 
\affil[3]{EPFL}

\author[4,5]{Marc-Olivier Renou}
\affil[4]{Inria Paris-Saclay, Bâtiment Alan Turing, 1 rue Honoré d’Estienne d’Orves – 91120 Palaiseau}
\affil[5]{CPHT, Ecole polytechnique, Institut Polytechnique de Paris, Route de Saclay – 91128 Palaiseau}

\date{}

\begin{document}

\maketitle

\begin{abstract}
Nonlocality and its connections to entanglement are fundamental features of quantum mechanics that have found numerous applications in quantum information science. 
A set of correlations is said to be nonlocal if it cannot be reproduced by spacelike-separated parties sharing randomness and performing local operations. 
An important practical consideration is that the runtime of the parties has to be shorter than the time it takes light to travel between them.
One way to model this restriction is to assume that the parties are computationally bounded.
We therefore initiate the study of nonlocality under computational assumptions and derive the following results:

\begin{itemize}
    \item[(a)] We define the set $\NEL$ (not-efficiently-local) as consisting of all bipartite states whose correlations arising from local measurements cannot be reproduced with shared randomness and \emph{polynomial-time} local operations.
    \item[(b)] Under the assumption that the Quantum Learning With Errors problem cannot be solved in \emph{quantum} polynomial-time, we show that $\NEL=\ENT$, where $\ENT$ is the set of \emph{all} bipartite entangled states (both pure and mixed). 
    This is in contrast to the standard notion of nonlocality where it is known that some entangled states, e.g. Werner states, are local.
    In essence, we show that there exist (efficient) local measurements of these states producing correlations that cannot be reproduced through shared randomness and quantum polynomial-time computation.
    \item[(c)] We prove that if $\NEL=\ENT$ unconditionally, then $\BQP\neq\PP$. In other words, the ability to certify all bipartite entangled states against computationally bounded adversaries leads to a non-trivial separation of complexity classes.
    \item[(d)] With the result from (c), we show that a certain natural class of 1-round delegated quantum computation protocols that are sound against $\PP$ provers cannot exist. 
\end{itemize}
\end{abstract}

\newpage

\tableofcontents

\sloppy

\newpage

\section{Introduction}

Quantum advantage refers to a situation in which a quantum machine outperforms classical counterparts. In the context of computation, it is when a quantum computer outperforms classical computers at solving certain problems, such as factoring integers or simulating quantum mechanical systems.
\emph{Nonlocality} is a different notion of quantum advantage that dates back to the inception of quantum mechanics itself.
This advantage arises in Bell games, where e.g. two non-communicating parties, Alice and Bob, are given inputs $x$ and $y$ and produce outcomes $a$ and $b$ distributed according to a probability distribution $p(a, b | x, y)$. Such distribution is said to be nonlocal if it cannot be expressed as a mixture of deterministic distributions, that is cannot be obtained by local operations on shared randomness.
As proven by John Bell~\cite{bell}, to answer questions raised in~\cite{einsteinepr}, there exist entangled quantum states that can produce nonlocal correlations through local quantum measurements.
The so-called CHSH game~\cite{CHSH69} fully formalizes this result. 
There, Alice and Bob are given inputs $x, y \in \{0, 1\}$ and asked to produce outputs $a, b \in \{0, 1\}$ such that $a \oplus b = x \cdot y.$ If Alice and Bob's correlations are local, it can be shown that, under uniformly random inputs $x, y$, their outputs will satisfy $a \oplus b = x \cdot y$ at most $75\%$ of the time. In contrast, if the two share an EPR pair $\ket{\phi^+} = \frac{1}{\sqrt{2}}(\ket{00} + \ket{11})$ on which they perform local measurements, as a function of their respective inputs, they can succeed with probability $\cos^2(\pi/8) \approx 85\%$.
This demonstrates the nonlocal nature of quantum mechanics.  

A natural question to ask is whether all entangled states can produce nonlocal correlations. In other words, for a given entangled state, does there exist a \emph{nonlocal game}, like the CHSH game, in which Alice and Bob can achieve a higher success rate by sharing the entangled state and performing local measurements than they could through \emph{any} strategy involving shared randomness and local operations?
If we restrict to pure entangled states, then this is indeed the case \cite{gisinpurestatesarenonlocal}. 
However, in 1989, Werner demonstrated the existence of mixed entangled states, for which a local model can always be constructed~\cite{wernerprojective}. 
That result holds only for projective measurements but it was later generalized to all POVMs in \cite{barrettmodel}.
Hence, there exists entangled states which correlations (obtained through any local quantum measurements) can be reproduced in a setting where the two parties hold only a separable state, which up to technical details is equivalent to sharing public randomness. 
It therefore seems that, at least in the standard nonlocal games framework, entanglement and nonlocality are fundamentally distinct notions.

In the hopes of finding a purely operational task that can exactly distinguish between entangled and separable states, other frameworks were considered. Among them is the so-called \emph{semi-quantum games} framework~\cite{Buscemi_SemiQuantumGame}. 
This differs from a standard nonlocal game by allowing Alice and Bob to receive trusted \emph{quantum inputs}, while still outputting classical bits. 
All entangled state (both pure and mixed) can produce nonclassical correlations in a semi-quantum game. More precisely, for every entangled state, there exists a semi-quantum game in which Alice and Bob can succeed with higher probability by sharing that state than what they could obtain from any strategy based on any separable state.

While semi-quantum games can distinguish between entangled and separable states, they have the rather unsatisfying feature of requiring quantum inputs. 
Effectively, the referee of the game must have a trusted device for preparing these quantum input states and then send them to Alice and Bob via quantum channels.
But to properly test quantum mechanics, one would aim for a more \emph{device-independent} characterization in which both inputs and outputs are classical and no quantum device is trusted. 
Subsequent work showed that this is possible at the expense of introducing two additional parties~\cite{fourparties}, while the bipartite setting remains open. 

A common feature of both nonlocal and semi-quantum games is that the ``no communication'' restriction on participating parties is often assumed to be enforced through \emph{spacelike separation}. 
In other words, Alice and Bob should send their responses to the referee in less time than it would take light to travel between them. 
This effectively puts a time limit on how long each party has to compute and send their response. 
However, do they always have the \emph{time} to compute this response? As we will observe later, \cite{barrettmodel} simulation algorithm seems computationally costly.
Combining this observation with the extended quantum Church-Turing thesis begs the question of whether accounting for the computational efficiency of the parties would change the set of states that can be certified with a nonlocal game.

\subsection{Main results} 

We introduce a model of two-party nonlocal games with \emph{computationally bounded} parties, where Alice and Bob's inputs and outputs are classical bit-strings and the only distinction from standard nonlocal games is the assumption of computational efficiency. 
In other words, we consider games in which Alice and Bob can achieve a higher success rate by sharing an entangled state and performing \emph{quantum polynomial-time} local operations than they could through \emph{any} quantum polynomial-time strategy involving shared randomness and local operations. 

Surprisingly, we show that this brings us closer to the goal of operationally characterizing all entangled states in the bipartite setting.

Our approach follows a recent trend of combining ideas from quantum information theory and computational complexity, which has yielded many new insights such as using cryptographic machinery to solve information-theoretic tasks like generating true randomness (\cite{brakerskicertifiedrandomness}), finding new results in quantum cryptography (\cite{cryptoeprint:2018/544}, \cite{cryptofromPRS}, \cite{morimaecommitmentsnoOWF}, \cite{Kretschmer2021QuantumPA}), using computational considerations to address paradoxes in quantum gravity (\cite{Aaronson2016TheCO}, \cite{bouland_wormholes}) and others. 
More recent works have also explored entanglement under computational assumptions, showing that any nonlocal game can lead to a test of quantum computational advantage~\cite{advantagefromNL} and that determining how entangled a state is can be computationally intractable (a notion referred to as \emph{pseudoentanglement})~\cite{pseudoentanglement, gheorghiu2020estimating}, later generalized in \cite{arnon-friedmancompent}.

Our contributions are the following:

\paragraph{New model of nonlocality.} 
We define a new notion of nonlocality that incorporates computational efficiency. 
We say that a state $\rho_{AB}$ is \emph{not-efficiently-local}, and denote the set of all such states as $\NEL$, if there exists a probability distribution arising from local measurements of $\rho_{AB}$ such that no \textit{efficient}, non-communicating parties (sharing a separable state) can reproduce this distribution. 
``Efficient'' in this context is defined as implementable in quantum polynomial time (QPT).

\paragraph{Cryptography implies entanglement certification.} 
We show that in our newly defined model, one can design a distinguishing experiment for \emph{all} entangled states, including mixed states. 

More concretely, under \emph{the Quantum Learning With Errors (QLWE) assumption}, we show that for every entangled state, $\rho_{AB}$, there exists a $3$-round (6-message) nonlocal game between a referee (which we will refer to as the \emph{verifier}) and non-communicating parties Alice and Bob such that,
\begin{itemize}
\item[(i)] if Alice and Bob share the state $\rho_{AB}$, there exist efficient quantum operations that they can perform locally so that they win the game with high probability and,
\item[(ii)] if Alice and Bob share any separable state, there do not exist any efficient quantum operations that they can perform locally in order to win the game with high probability.
\end{itemize}

This demonstrates that our notion of not-efficiently-local states, exactly characterizes the set of all bipartite entangled states. 
Summarizing, under the QLWE assumption, we have that $\ENT = \NEL$, where $\ENT$ denotes the set of all bipartite entangled states.

As mentioned the protocol relies on the QLWE assumption. 
The Learning With Errors (LWE) problem introduced in \cite{Regev05} is widely used in cryptography to create secure encryption algorithms. 
It is based on the idea of representing secret information as a set of equations with errors (\cite{LWEassumption}). 
The QLWE assumption is standard in cryptography (\cite{regevLWE}) and assumes that LWE is intractable for polynomial-time quantum algorithms.
The security of a scheme under the QLWE assumption is shown via a reduction to this problem. 
In our case, we show that if the parties win the game when they share some separable state, then this implies that the parties solved the LWE problem.
This contradicts the assumptions of our model as the parties were required to be quantum polynomial time.

\paragraph{Entanglement certification implies complexity class separation.}
Having shown that $\NEL=\ENT$ under the QLWE assumption, we then focus on what it would take to prove the equality unconditionally. In other words, how hard is it for Alice and Bob to fake entanglement with separable states? To address this question, we consider a particular way of \say{faking entanglement}, namely the Hirsch local model (\cite{hirshgenuinenonloc}) for $1$-round protocols. We show that this model can be implemented in $\PP$. 
Concretely, for some entangled state $\rho_{AB}$, we show how to simulate any QPT strategy of $\text{Alice}_1$ and $\text{Bob}_1$ having access to $\rho_{AB}$, by a $\PP$ strategy of $\text{Alice}_2$ and $\text{Bob}_2$ with access to a \emph{separable state} $\sigma_{AB}$.
This means that local simulation in the Bell scenario is at most as hard as $\PP$. 
Summarizing, $\ENT = \NEL \Rightarrow \BQP \neq \PP$.

This shows that separating $\BQP$ and $\PP$ is a necessary condition for being able to certify entanglement against computationally bounded parties, while the previous result shows that the QLWE assumption is a sufficient condition.\footnote{Also note that the QLWE assumption directly implies $\BQP \neq \PP$, as LWE $\in \mathsf{NP} \subseteq \PP$.}


\paragraph{Delegation of quantum computation.} 
The result $\ENT = \NEL \Rightarrow \BQP \neq \PP$ has interesting implications for protocols for delegating quantum computations (DQC). 
These are protocols in which a classical polynomial-time \emph{verifier} delegates a $\BQP$ computation to an untrusted quantum \emph{prover} and is able to certify the correctness of the obtained results (a property known as \emph{soundness}). 
A breakthrough result by Mahadev gave the first such DQC protocol with soundness against $\BQP$ provers~\cite{mahadev}. This was achieved under the QLWE assumption.
A major open problem in the field is whether DQC protocols that are sound against computationally unbounded provers exist.

We give evidence for the difficulty of resolving this question by showing that a version\footnote{Our version of DQC is related to the quantum fully-homomorphic encryption scheme considered in \cite{advantagefromNL}. See Definition~\ref{def:EDQC}.} of $1$-round DQC protocols sound against $\PP$ provers do not exist. 
We can also rephrase this by saying that if there exists a certain $1$-round DQC protocol sound against $\BQP$ provers, then $\BQP \neq \PP$. 
This is in the same spirit as the results of~\cite{aaronsonimposs}, showing that blind DQC protocols sound against all provers are unlikely to exist.
Our result can be viewed as an improvement over that work, as we show the non-existence of protocols sound against $\PP$ provers, which is a weaker requirement.


\section{Technical Overview}\label{sec:technicaloverview}

We start by introducing our new model of nonlocality, which we call the not-efficiently-local model. 
Our definition, in the spirit of the extended quantum Church-Turing thesis, assumes that any computation in the physical world can be modeled by a polynomial time quantum machine. 
This essentially translates to limiting the power of dishonest parties in a nonlocal game to quantum polynomial time (QPT). 
An informal version of the definition (see Definition~\ref{def:nonlocality} for the formal version) states

\begin{definition}[Not-efficiently-local - Informal]\label{def:nonlocalityinformal}
For a quantum state $\rho_{AB}$ we say that $\rho_{AB}$ is not-efficiently-local if there exists a game (or protocol) $\mathcal{G}(\rho_{AB})$ between a probabilistic polynomial-time (PPT) verifier, $\ver$, and two non-communicating QPT parties $\alice$, $\bob$. Specifically, for every $\ell \in \N$, all the parties run in time $\poly(\ell)$ and the verifier exchanges $\poly(\ell)$ bits of communication with $\alice$ and $\bob$. The game satisfies the following properties:
\begin{enumerate}
    \item{(\textbf{Completeness})} If $\mathcal{G}(\rho_{AB})$ is run with $\alice, \bob$ sharing $\rho_{AB}$, $\ver$ accepts the interaction with probability at least $c(\ell)$,
    \item{(\textbf{Soundness})} For \textbf{every} QPT $\alice',\bob'$, if $\mathcal{G}(\rho_{AB})$ is run with $\alice',\bob'$ sharing a separable state, $\ver$ accepts the interaction with probability at most $s(\ell)$,
\end{enumerate}
with $c(\ell) - s(\ell) > \frac{1}{\poly(\ell)}.$
We denote the set of all such states, $\rho_{AB}$, as $\NEL$.
\end{definition}

\begin{remark}
We emphasize that the polynomial runtime of $\alice$ and $\bob$ is always with respect to the parameter $\ell$ which sets the desired gap between completeness and soundness. In particular, the dimension of the shared entangled state that is being tested need not depend on $\ell$ and can be constant. In a practical run of such a protocol, the value of $\ell$ would be determined based on the spatial separation between $\alice$ and $\bob$, which establishes the maximum amount of time for $\alice$ and $\bob$ to respond to $\ver$, and cryptographic considerations like whether $2^{\ell}$ operations can be performed in that time.
\end{remark}

\begin{remark}
With this definition there is a slight ambiguity in whether we're assuming not just that $\alice$ and $\bob$ can solve only problems in $\BQP$, but that their operations can be modelled using quantum mechanics. That is to say, whether we can represent $\alice$ and $\bob$ using quantum circuits of polynomial size acting on some input state, as opposed to making no assumptions about their inner workings. This becomes important when one proves computational reductions using the operations of $\alice$ and $\bob$. In cryptography, it is the distinction between whitebox and blackbox reductions. We will assume that the former is the case (we model $\alice$ and $\bob$ as quantum circuits and perform whitebox reductions). As will become clear, this fact is relevant for our second result. The full details can be found in Section~\ref{sec:notefflocal} and we make additional comments about this distinction in Section~\ref{sec:openproblems}.    
\end{remark}
Essentially by definition, all states that are not-efficiently-local must be entangled. Denoting the set of all entangled bipartite states as $\ENT$, this means that $\NEL \subseteq \ENT$. A natural question is whether it is also the case that $\ENT \subseteq \NEL$, or in other words, whether $\ENT = \NEL$.
We phrase this question as a conjecture to which we will refer throughout the paper.


\begin{conjecture}[$\ENT = \NEL$]\label{conj}
For every finite-dimensional $\mathcal{H}_A,\mathcal{H}_B$ and every entangled state $\rho_{AB}$ on $\mathcal{H}_A \otimes \mathcal{H}_B$, $\rho_{AB}$ is not-efficiently-local.
\end{conjecture}
\noindent

If Conjecture~\ref{conj} is true then there is an equivalence between the notion of entanglement and our notion of not-efficient-locality.
We will also say that if Conjecture~\ref{conj} is true then we can certify all entangled states. 

\subsection{Cryptography implies Entanglement Certification}\label{sec:techcryptoimpliesentcert}

Our second contribution shows that Conjecture~\ref{conj} is true (and $\ENT = \NEL$), under the assumption that LWE is intractable for QPT algorithms.\footnote{Strictly speaking, we are assuming that LWE is intractable for \emph{non-uniform} QPT algorithms, i.e. $\text{LWE} \not\in \mathsf{BQP/qpoly}$. This is a standard assumption in post-quantum cryptography (see for instance Definition 2.5 in~\cite{BCM}).}
The result can be informally summarized as follows:


\begin{theorem}[Informal]\label{thm:mainCert}
Assuming QLWE, for every finite dimensional Hilbert spaces $\hilb_A,\hilb_B$, every entangled state $\rho_{AB}$ on $\hilb_A \otimes \hilb_B$ and every $\ell \in \N$ there exists a $3$-round ($6$-message\footnote{We assume $1$ round is always equal to $2$ messages.}) interactive protocol where the PPT($\ell$) verifier, $\ver$, exchanges $\poly(\ell)$ bits of communication with QPT($\ell$) $\alice$ and $\bob$ such that:
\begin{itemize}
    \item If $\alice$ and $\bob$ share $\rho_{AB}$ and follow $\ver$'s instructions, then their interaction is accepted by $\ver$ with probability $1 - \negl(\ell)$.
    \item For every $\alice', \bob' \in \text{QPT}(\ell)$, if they share a separable state then the interaction is accepted by $\ver$ with probability $1/\poly(\ell)$.
\end{itemize}
\end{theorem}


\noindent \textbf{Proof techniques.} We first develop a \emph{remote state preparation (RSP) protocol} which allows the classical verifier to certify the preparation of certain states by Alice and Bob. Assuming QLWE, the protocol ensures that Alice and Bob will not know which states they have prepared (while the verifier will). The protocol is based on the one from~\cite{vidick}.
Next, the verifier will perform the Buscemi \emph{semi-quantum game (SQG)} with Alice and Bob~\cite{Buscemi_SemiQuantumGame}. As mentioned, this game, which requires quantum inputs for Alice and Bob, allows the verifier to certify any bipartite entangled state. Our protocol will then achieve this under computational assumptions.
The ideas are illustrated schematically in Figure~\ref{fig:ThreeProtocols}.
To give more details, let us briefly review both SQG and RSP.


\begin{figure}
    \centering
    \includegraphics[width=0.525\columnwidth]{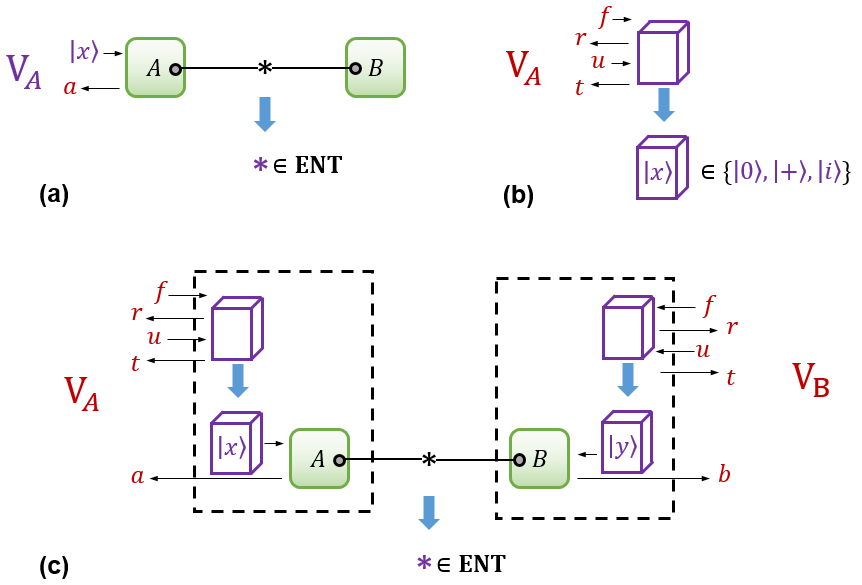}
    \caption{(a) The Semi-Quantum Game (SQG) protocol of \cite{Buscemi_SemiQuantumGame}. 
    In this picture, we split $\ver$ into two entities $\ver_A$ and $\ver_B$ for convenience.
    $\alice, \bob$ share an entangled state $\rho_{AB}$. 
    $\ver_A$ sends an input state $\ket{x}\in\{\zero,\one,\plus,\minus,\ii,\minusi \}$ to $\alice$, who performs a Bell state measurement on $\ket{x}$ and its share of $\rho_{AB}$. 
    $\ver_B, \bob$ do the same with input $\ket{y}$.
    The statistics $p(a,b \ | \ \ket{x},\ket{y})$ certify that $\rho_{AB}$ had to be entangled.\\
    (b) Our Remote State Preparation (RSP) protocol, see Fig.~\ref{fig:RSP_protocol} for details. $\ver_A$ and $\alice$ exchange three rounds of classical messages. 
    Under the assumption that LWE is quantumly hard, the protocol certifies that $\alice$ had prepared a quantum state $\ket{x}\in\{\zero,\one,\plus,\minus,\ii,\minusi \}$ which is not known, but can be operated on, by $\alice$, but is known by $\ver_A$.\\
    (c) Our Entanglement Certification protocol is obtained by combining two independent RSP protocols with $\alice$ and $\bob$, followed by the SQG protocol. 
    }
    \label{fig:ThreeProtocols}
\end{figure}

\emph{The SQG bipartite protocol.} 
We review the protocol from~\cite{Buscemi_SemiQuantumGame}, adopting the formulation of~\cite{MDIEW}. 
The protocol is based on \emph{entanglement witnesses}, which are operators $W$ such that $\trace[W\sigma_{AB}] \leq 0$ for all separable states $\sigma_{AB}$ but $\trace[W\rho_{AB}] >0$ for some entangled state $\rho_{AB}$. 
As the set of separable states is convex, any entangled state $\rho_{AB}$ can be associated with an entanglement witness $W$ with $\trace[W\rho_{AB}]>0$.
Any such $W$ admits a local tomographic decomposition $W=\sum_{s,t}\beta_{s,t} \  \tau_s\otimes\omega_t$ where $\tau_s, \omega_t$ are the density matrices corresponding to the states in $\mathcal{S} = \{\zero,\one,\plus,\minus,\ii,\minusi \}$.

Based on this decomposition,~\cite{MDIEW} proposes a two player game, which certifies that $\rho_{AB}$ is entangled (see Fig.~\ref{fig:ThreeProtocols}a).
Two players, sharing $\rho_{AB}$, independently receive random quantum input states $\tau_s$ and $\omega_t$, respectively, from two verifiers.\footnote{This can also be just one verifier, as before. We mention two verifiers to be consistent with the presentation from~\cite{MDIEW} and to be closer to how such an experiment would be realized in practice. Indeed, since our resulting protocol is 3-round, the only way to ensure spacelike separation would be to have a local verifier for Alice and one for Bob that are giving them their questions and recording their responses. The two verifiers would then communicate with each other to decide whether Alice and Bob win the game or not.} 
Each player is asked to perform a joint Bell state measurement on the received quantum input state and their share of $\rho_{AB}$ and output the result of the measurement.
One can introduce a Bell-like score based on the decomposition of $W$ such that if the players act honestly then the score is exactly $\trace[W\rho_{AB}]$, and if the players cheat by employing any strategy based on some separable shared state $\sigma_{AB}$, then the score is upper bounded by $\trace[W\sigma_{AB}] \leq 0$.
Hence, if the score is strictly positive, the protocol certifies that the players shared some entangled state.
The possible quantum inputs for Alice and Bob are the states in $\mathcal{S}$ and so, for our protocol, these are the states the verifier(s) would have to instruct them to prepare.

\emph{The RSP protocol.} 
Remote state preparation is a protocol between a classical verifier and a quantum prover, allowing the verifier to certify the preparation of a quantum state in the prover's system (see Fig.~\ref{fig:ThreeProtocols}b and~\ref{fig:RSP_protocol}). 
It will allow the verifier to check (under the QLWE assumption) that the prover has prepared a random state from the set $\mathcal{S}$.
We show that our protocol is (i) complete, meaning that an honest QPT prover succeeds with high probability in preparing a state from $\mathcal{S}$, known only to the verifier, and (ii) sound, meaning that any dishonest QPT prover attempting to deceive the verifier will fail with high probability.
The formal statement of soundness is delicate but informally speaking it guarantees that if the QPT prover is accepted, then its behavior is equivalent to it having received a random quantum state from $\mathcal{S}$.

Our protocol is simpler than the one from~\cite{vidick}. 
In that protocol the prover prepares eigenstates of all of the five bases $\left\{X,\frac{X-Y}{\sqrt{2}},Y,\frac{X+Y}{\sqrt{2}},Z \right\}$.
We require the prover to prepare eigenstates only of $\left\{X,Y,Z\right\}$.
It is because these eigenstates are exactly the members of $\mathcal{S}$ which we require to perform the SQG.
Our protocol performs fewer checks and therefore the soundness does not follow directly from that of~\cite{vidick}. 
As such, we require a careful analysis to derive the soundness of our protocol. 

\emph{The Entanglement Certification protocol.}
Our resulting entanglement certification protocol is a combination of RSP and SQG.
The verifier performs two independent runs of RSP, one with Alice and one with Bob. If either one fails, the verifier will count this as a loss. Otherwise, the verifier will then perform the Buscemi SQG with Alice and Bob.
Combining the two protocols involves several nontrivial steps stemming from the fact that we operate in the computationally-bounded model.
For instance, we prove that the number of repetitions of the protocol that are required to distinguish the two cases is small enough so that the verifier is efficient.

\subsection{Entanglement Certification implies Separation of Complexity Classes}\label{sec:techhentanglementiimpliescrypto}


The previous result gives a \emph{sufficient} condition for proving Conjecture~\ref{conj}, namely the QLWE assumption. Here, we derive a \emph{necessary} condition by showing that if the conjecture is true we obtain a non-trivial separation of complexity classes. Specifically:

\begin{restatable}{theorem}{thmmain}\label{thm:maininformal} 
$\ENT = \NEL \Rightarrow \BQP \neq \PP$.
\end{restatable}

We note that $\PP$ is a relatively large class, e.g. $\PH \subseteq \P^\PP$ by Toda's theorem (\cite{todastheorem}), but nevertheless $\PP \subseteq \PSPACE$.
A consequence of our result is therefore that proving Conjecture~\ref{conj} is at least as hard as proving that $\BQP \neq \PSPACE$. 


\paragraph{Proof techniques.}
Our starting point is a result of \cite{hirshgenuinenonloc} showing that for a certain entangled state, $\rho_{AB}$, there is a local model for \emph{all} POVM measurements. In other words, by sharing only a separable state (or random bits), Alice and Bob would be able to reproduce the statistics produced through local measurements of $\rho_{AB}$. In effect, this implies that the state $\rho_{AB}$ cannot be certified through a nonlocal game. The question then becomes: how hard is it for Alice and Bob to implement the strategy from~\cite{hirshgenuinenonloc}?
We show that their strategies, given as Algorithms~\ref{alg:alicesim} and~\ref{alg:bobsim}, can be implemented in $\PP$.
This then proves the contrapositive of Theorem~\ref{thm:maininformal}, i.e. that if $\BQP = \PP$ then $\ENT \neq \NEL$. 
For ease of explanation, we will describe Alice and Bob's strategies as implementable by QPT machines with postselection, making use of the seminal result of Aaronson that $\PP = \PostBQP$~\cite{PostBQPPP}. $\PostBQP$ is the set of problems solvable in quantum polynomial-time with \emph{postselection} (see Section~\ref{sec:prelimcomplexity} for the formal definition).

The input in Algorithm~\ref{alg:alicesim} (a similar situation happens for Algorithm~\ref{alg:bobsim}) is a POVM $ \{ \mathcal{A}_a\}_a$.
Up to \say{fine-graining} we can, without a loss of generality, assume that
$\mathcal{A}_a = \eta_a P_a$ with $\eta_a \in (0,1]$ and $P_a$ being a rank one projector (see Section~\ref{sec:prelimquantum}). 
In our computational setting, this POVM is implicitly defined by a question $x$, which $\alice$ received from $\ver$, and a circuit $C^\alice_{\poly(\ell)}$ that an honest $\alice$ would have applied to her share of $\rho_{AB}$ (the formal correspondence is given in \eqref{eq:howAacts}). 
The main challenge in implementing the simulation is to realize the operations in Algorithm~\ref{alg:alicesim} (and~\ref{alg:bobsim}), which are expressed in terms of $\{ \mathcal{A}_a\}_a$, efficiently while having access to the description of $C^\alice_{\poly(\ell)}$ only. 
Intuitively, the local model works as follows.
$\alice$ and $\bob$ will share many copies of Haar random single-qubit states, denoted $\ket{\lambda}$.\footnote{We note that the proof works even if the two share classical descriptions of these states that can represent them to a sufficiently high precision. See Section~\ref{sec:entanglementcertimpliescrypto} for details.}
Their answers to $\ver$ are then computed based on dot products of the form $\bra{\lambda} P_a \ket{\lambda}$.
The key step in the proof is to estimate these dot products up to enough precision in $\PostBQP$.
For instance, to implement $\mathbbm{1}_{\{\bra{\lambda} P_a \ket{\lambda} < \bra{\lambda} (\id - P_a) \ket{\lambda}\}}$ from Algorithm~\ref{alg:alicesim}, we show (Lemma~\ref{lem:circuittoPOVM}) that there exists a QPT algorithm with postselection that, given $a$, creates a 1-qubit state which is equal to the eigenvector of $P_a$. 
Next, we prove that given a description of $\ket{\lambda}$ it is possible to compute an approximation to $\bra{\lambda} P_a \ket{\lambda}$ while having access to polynomially many copies of the eigenvector state of $P_a$.
To do that we build on the ideas from \cite{PostBQPPP} to show that using the power of post-selection one can compute exponentially accurate dot products between 1-qubit states (see Lemma~\ref{lem:dotprodapxcorrect}) in polynomial time. 

The final step is to argue that if $\PostBQP =\BQP$ then we can implement our simulation in QPT. 
This is not immediate, as the assumption $\PostBQP =\BQP$ allows us to replace QPT machines with postselection with QPT machines without postselection \emph{only for decision problems}.
However, the QPT algorithm with postselection designed for our local simulation, samples an answer according to a distribution---it is thus an instance of a sampling problem.
As such, we need to also prove a sampling-to-decision reduction for this problem.
We do this through a careful binary-search-like procedure as shown in Lemma~\ref{lem:howtouseBQP=PostBQP}.

\begin{remark}
    It's important to clarify how postselection is used in our setting. Indeed, one thing that \textbf{we do not do} is to allow Alice and Bob to postselect on their shared randomness, as that would permit them to signal instantaneously. As mentioned, we use $\PostBQP$ merely for ease of explanation. We can equivalently view Alice and Bob as being efficient agents with access to a $\PP$ oracle which they use in order to obtain very precise estimates of inner products. From this perspective, postselection is used to provide a ``boost'' in computational power, not to communicate instantaneously.
\end{remark}



\subsection{Implications for Delegation of Quantum Computation (DQC)}

We consider a version of a $1$-round DQC protocol, which we refer to as Extended Delegation of Quantum Computation (EDQC), formally presented in Definition~\ref{def:EDQC}. We then show that EDQC protocols sound against $\PP$ provers do not exist. Equivalently, we prove that the existence of a $\BQP$-sound, EDQC implies $\BQP \neq \PP$.
This can be summarized in the following informal theorem:

\begin{restatable}[Informal]{theorem}{thmedqcinformal}\label{thm:edqcinformal} 
$\PP$-sound EDQC protocols do not exist. 
\end{restatable}

\noindent \textbf{Proof techniques.} 
The high-level idea of the proof is to construct a $\PP$ prover that is always able to cheat successfully in an EDQC protocol. Let us first informally describe what an EDQC protocol is.
For a circuit $\mathcal{C}$ that accepts as input a classical $s$ and an auxiliary 1-qubit\footnote{
We could have considered a more general definition, where the auxiliary register holds a state on many qubits. 
Our definition generalizes naturally. 
We chose this version for simplicity and because our result implies that the more general protocol is not possible either.}
state $\rho_Q$, EDQC guarantees the following. 
\begin{itemize}
    \item{\textbf{(Completeness)}} There exists a $\text{PPT}$ verifier and a $\text{QPT}$ prover such that for every state $\rho_{QE}$ on $\Com^2 \otimes \hilb_E$ chosen by the prover the interaction will be accepted and will satisfy the following. 
    For every $s$ the following two states are equal: (i) the output bit collected by the verifier and the contents of the $E$ register, (ii) the result of measuring the output qubit of $U_\mathcal{C} \otimes \id_E \ket{s} \rho_{QE}$ and the contents of the $E$ register (where $U_{\mathcal{C}}$ is the unitary corresponding to $\mathcal{C}$).
    In words, the protocol preserves the entanglement between the output of $\mathcal{C}$ and the $E$ register.
    \item{\textbf{(Soundness)}} For every $\text{QPT}$ prover accepted with high probability there exists a state $\rho_Q$ on $\Com^2$ such that the output bit collected by the verifier is, \textit{for every classical} $x$, close to the distribution of measuring the output qubit of $U_\mathcal{C} \ket{s} \rho_Q$.
\end{itemize}
Note that if the output qubit of $\mathcal{C}$ doesn't depend on $\rho_Q$ then our definition reduces to the standard verification definition for $1$-round DQC.
This means that the only additional requirement of EDQC is consistency with the auxiliary input.
A recent work (\cite{advantagefromNL}, Definition 2.3) considered a very similar requirement in the context of quantum fully-homomorphic-encryption (qFHE).
The authors demonstrate that the auxiliary input requirement is satisfied by preexisting qFHE protocols of \cite{mahadevFHE} and \cite{brakerskiQFHE}. Recently, this was used by~\cite{nonlocalitytoDQC} to show that the existence of qFHE with auxillary input implies the existence of DQC.
The main difference between our notion of EDQC and qFHE with auxillary input is that the soundness of qFHE is expressed as indistinguishability of the prover's views, whereas our definition requires a type of verifiability.
We note that although the EDQC definition is nonstandard, protocols satisfying its requirements most likely exist under the QLWE assumption in the quantum random-oracle model (QROM).
We refer the reader to Section~\ref{sec:whyEDQC} for a more in-depth discussion about the details of EDQC and a comparison to other variants of DQC.

\begin{table*}[h!]\centering
\ra{1.3}
\begin{tabular}{@{}rrrr@{}}\toprule
& $\BQP$-sound & $\PP$-sound & $\ALL$-sound \\
\midrule
$1$-round & $\textbf{YES} - \substack{\text{\cite{alagic}} \\ \text{QLWE, in QROM} } $ & $\textbf{NO} - \text{Theorem~\ref{thm:edqcinformal}}$ & $\textbf{NO}- \substack{\text{By def.} \\ \BQP \not\subseteq \AM }$ 
\\ 
\vspace{-4mm}
\\
$2$-round & $\textbf{YES}  - \substack{\text{\cite{mahadev}} \\ \text{QLWE}}$ &  ? & $\textbf{NO}- \substack{\text{\cite{babaiAMgames}} \\ \BQP \not\subseteq \AM }$ 
\\
\vspace{-4mm}
\\
$\poly$-round &  $\textbf{YES}  - \substack{\text{\cite{mahadev}} \\ \text{QLWE}}$  &  ? & $\textbf{NO}- \substack{\text{\cite{aaronsonimposs}} \\  \BQP \not\subseteq \sf{NP/poly} }$ \\
\bottomrule
\end{tabular}
\caption{Summary of known results about the possibility or impossibility of DQC protocols. The rows denote the number of rounds in a protocol and the columns denote the required soundness. We implicitly assume $\BQP$ completeness. Additionally, next to each entry we list the corresponding work and the assumptions needed. We note that our result doesn't require any assumptions.}\label{tab:DQC}
\end{table*}

To prove Theorem~\ref{thm:edqcinformal} we make use of Theorem~\ref{thm:maininformal}, i.e. $\ENT = \NEL \Rightarrow \BQP \neq \PP$. 
Specifically, we assume towards contradiction that a $\BQP$-sound, EDQC protocol exists and that $\BQP = \PP$. 
Next, we build a nonlocal game certifying $\ENT = \NEL$ by using two independent instances of EDQC to control the behavior of the two players. This is very much in the same spirit as how we used RSP and SQG to prove $\ENT=\NEL$ under the QLWE assumption in Theorem~\ref{thm:mainCert}. We then obtain a contradiction with Theorem~\ref{thm:maininformal} which concludes the proof.

Our approach yields a stronger result about the limits of DQC compared to previous works. We summarize those results, together with 
our contribution, in Table~\ref{tab:DQC}. Note that showing the impossibility of a $\PP$-sound protocol is a harder requirement than showing the impossibility of an $\ALL$-sound protocol. Finding a local hidden variable model implementable in a lower complexity class (i.e. a subset of $\PP$), or with more rounds of interaction, would imply even stronger lower bounds. We also expect our technique to be useful for showing impossibility results about other cryptographic primitives, such as (quantum) FHE.



\section{Related works}

\paragraph{Computational entanglement.} There are a number of recent works that study entanglement through the computational lens.
In \cite{pseudoentanglement} the authors give a construction of states computationally indistinguishable from maximally entangled states with entanglement entropy arbitrarily close to $\log n$ across every cut.
This gives an exponential separation between computational and information-theoretic quantum pseudorandomness (in the form of $t$-designs).
Extending upon \cite{pseudoentanglement} in \cite{arnon-friedmancompent} a rigorous study of computational entanglement is initiated. 
More concretely, they define the computational versions of one-shot entanglement cost and distillable entanglement.
Informally speaking, the operational measure of entanglement in both of these works relates to the number of maximally entangled states to which a state in question is equivalent.
In contrast, our operational measure is whether a state displays any nonlocality.

Recently, \cite{advantagefromNL} demonstrated how to compile any nonlocal game to a single-prover interactive game maintaining the same completeness and soundness guarantees, where the soundness holds against $\BPP$ adversaries.
The compilation uses a version of quantum fully-homomorphic-encryption (qFHE).
This means that under computational assumptions any non-local game leads to a single-prover quantum computational advantage protocol.
Building upon \cite{advantagefromNL} in \cite{nonlocalitytoDQC} a DQC protocol is constructed by compiling the CHSH game using a qFHE scheme (sound against $\BQP$ adversaries). 
These works are in some sense dual to Theorem~\ref{thm:mainCert}.
They show how to use nonlocal games to construct more efficient protocols for quantum advantage and DQC while Theorem~\ref{thm:mainCert} shows how to use RSP to construct nonlocal games.

\paragraph{Generalizations of nonlocality and local simulations.}
Compared to the Bell scenario, a larger set of entangled states can be certified in a scenario where many copies of the state are given~\cite{Superactivation_MultipleCopies}, or where more rounds of communication are allowed~\cite{hirshgenuinenonloc}, or in the broadcasting scenario~\cite{broadcasting}.
However, it is still an open question whether one of these scenarios (or a combination of them) can reconcile the concepts of entanglement and nonlocality.

The local simulation of Barrett (\cite{barrettmodel}) was generalized in \cite{hirshsigmalocal}, where it was shown that there exist entangled states with local models for all protocols in a special case of a sequential scenario called local filtering.
These are 3-message protocols, where the parties first send one bit each to the verifier, then receive a challenge, and finally reply to the challenge. 

\paragraph{Assumptions needed for $\BQP$-sound DQC.} 
It is believed that $\BQP$ is not in $\AM$, where $\AM$ is a class of languages recognizable by $1$-round protocols sound against $\ALL$ provers.
A classical result shows that for every constant $k >2$ we have $\AM[k] = \AM$, where $\AM[k]$ is the same as $\AM$ but with $k$-messages (each round consists of $2$ messages)~\cite{babaiAMgames}.
These two results imply that $\ALL$-sound, constant-round DQC protocols don't exist under the likely assumption that $\BQP \not\subseteq \AM$. 
In \cite{aaronsonimposs} a generalization to more rounds was considered. 
It was observed that $\ALL$-sound blind DQC with polynomially-many rounds of interaction does not exist provided $\BQP \not\subseteq \sf{NP/poly}$. It was also shown that blind DQC protocols exchanging $O(n^d)$ bits of communication, which is $\ALL$-sound imply that $\BQP \subseteq \MA/O(n^d)$---containment which does not hold relative to an oracle.

Note, that \cite{alagic} showed that a sufficient assumption for $\BQP$-sound $1$-round DQC protocol in the QROM is the QLWE assumption.
Our result shows that $\BQP \neq \PP$ is a necessary condition.
Similarly, the assumption of $\BQP \neq \PP$ was recently shown necessary for the existence of pseudorandom states, which shows that it is a nontrivial requirement~\cite{Kretschmer}.
Pseudorandom states, introduced in \cite{prsfirst}, are efficiently computable quantum states that are computationally indistinguishable from Haar-random states. 

\section{Discussion and open problems}\label{sec:openproblems}

We have introduced the concept of nonlocality under computational assumptions by considering nonlocal games in which participating parties are computationally bounded. 
In this model, we showed that all entangled states can be distinguished from separable states, under the QLWE assumption. 
This is in contrast to the standard notion of nonlocality in which states like certain Werner states cannot be distinguished from separable states. 
This can be seen in analogy to the notions of \emph{proof systems} and \emph{argument systems}. 
The former are interactive protocols that should be sound against unbounded adversaries, whereas the latter have to be sound merely against computationally bounded adversaries. 
It is known that in many situations argument systems are more expressive than proof systems. 
For instance, in the case of zero-knowledge protocols, $\mathsf{SZK}$, the set of problems that admit statistical zero-knowledge proofs is believed to be strictly smaller than $\mathsf{CZK}$, the set of problems admitting computational zero-knowledge arguments. 
Similarly, we showed that the set of entangled states that can be distinguished from separable states is strictly larger when the parties are computationally bounded.
While the QLWE assumption is sufficient for our results, we also showed that $\BQP \neq \PP$ is a necessary requirement. 
From this, we derived the non-existence of certain $\PP$-sound protocols for delegating quantum computations.
Our work opens up several interesting directions for further exploration.

\paragraph{Operational characterization of $\ENT$ and blackbox reductions.}
As mentioned in the introduction, one of the main goals of research into nonlocality is to give an operational characterization of the set of all entangled states. 
Importantly, this characterization should not assume the correctness of quantum mechanics a priori. 
For this reason, while we showed that all entangled states can be distinguished from separable ones under the QLWE assumption, we also assumed the correctness of quantum mechanics because we described the inner workings of Alice and Bob using quantum circuits. 
In the cryptographic terminology, we proved a \emph{whitebox} reduction from LWE to the soundness of the $\NEL$ protocol. 
This is inherited from the soundness proof of the RSP protocol, which uses the prover's quantum operations explicitly in order to construct an efficient adversary that solves LWE. 
As such, an interesting open problem is whether there exists a \emph{blackbox} reduction. 
If so, this would indeed provide the desired operational characterization of $\ENT$ under the QLWE assumption. 

It would also be desirable to have a $\NEL$ protocol involving only one round of communication, as opposed to our protocol which uses 3 rounds.

\paragraph{Further applications of computational entanglement.}
As we've noted, several recent works have investigated computational notions of entanglement, primarily pseudoentanglement. 
Much like how pseudoentanglement provides a separation between computational and information-theoretic pseudorandomness, we expect computational nonlocality to provide further interesting separations with respect to standard, information-theoretic nonlocality. 
For instance, one could consider nonlocal games in which both parties receive the same question from the referee. 
In standard nonlocality, this would not allow for the certification of any state, as both Alice and Bob would know each other's input. 
However, in the semi-quantum games framework, while both parties receive the same quantum state, due to the uncertainty principle they would still not know which state they received. 
This could then be extended to the $\NEL$ setting, showing that nonlocal games with both parties having the same question can still certify entangled states, provided the parties are computationally bounded.
In essence, in this computational setting, one is trading the uncertainty principle for non-rewindability. A similar connection between these two distinct notions of non-classicality was also observed by~\cite{advantagefromNL}.

Another avenue would be to explore scaled-down versions of standard multi-prover complexity classes. 
For instance, one could consider instances of $\MIP$ or $\QMA(2)$ with $\BQP$ honest provers. 
It's known that $\MIP^*[\BQP] = \BQP$, where $\MIP^*[\BQP]$ denotes the set of problems that can be verified efficiently by interacting with two $\BQP$ provers sharing entanglement (and in which soundness is against unbounded provers). 
Less is known in the setting where the provers are only allowed to share a separable state (such as $\MIP$ or $\QMA(2)$) or an entangled state that is local.
Connections to local simulation, along the lines of the proof of Theorem~\ref{thm:edqcinformal}, might yield a separation between $\MIP[\BQP]$ and $\MIP^*[\BQP]$.

\paragraph{Cryptography and $\ENT = \NEL$.}
It is natural to conjecture that Theorems~\ref{thm:mainCert} and~\ref{thm:maininformal} could be strengthened so that an equivalence between the existence of a cryptographic primitive and $\ENT = \NEL$ is achieved.
One avenue for strengthening Theorem~\ref{thm:mainCert} could be to base the result on qFHE with auxiliary input~\cite{advantagefromNL}. Indeed, this was shown in~\cite{nonlocalitytoDQC} to allow for DQC. 
Likely, their results would also lead to an RSP protocol, in which case one would indeed derive $\ENT=\NEL$ from $\BQP$-sound qFHE, following our approach. This would also have the intriguing feature of compiling a standard nonlocal game (in the case of~\cite{nonlocalitytoDQC}, the CHSH game) into a nonlocal game with computationally bounded parties.

To improve Theorem~\ref{thm:maininformal} one could try to use the fact that $\ENT = \NEL$ gives us an \emph{average-case hardness} and not just the worst-case hardness that we used. Additionally, since the main element in our proof was the estimation of certain inner products to within inverse exponential precision, any approach for doing this in a class that's lower than $\PP$ (i.e. a subset of $\PP$), would directly improve our result. It would be worth seeing, for instance, whether Stockmeyer's approximate counting could help, in which case one would obtain that $\ENT=\NEL$ implies $\BPP \neq \NP$.

It is worth mentioning that connections of a similar flavor have recently been found in other contexts, e.g. in \cite{zvikaequiv} an equivalence between a phenomenon in high-energy physics (the hardness of black-hole radiation decoding) and the existence of standard cryptographic primitives (EFI pairs from \cite{brakerskiquantumcrypto}) was shown.


\paragraph{Local simulation and DQC.}
Further exploration of connections with local simulations can be fruitful.
As we mentioned, \cite{hirshsigmalocal} gives a local simulation for 3-message protocols.
If one can implement it in $\PP$ also then that would likely imply that $\PP$-sound 3-message protocols for DQC don't exist.
Moreover, there is a strong connection between many-round local simulations and the major open problem of whether information-theoretic sound DQC protocols exist.
For instance, if one could show that there exists a local simulation for some entangled state, for any number of rounds, then that should rule out information-theoretic sound EDQC protocols.

\section*{Acknowledgements}
We thank Thomas Vidick for helpful discussions. AG is supported by the Knut and Alice Wallenberg Foundation through the Wallenberg Centre for Quantum Technology (WACQT). At the time this research was conducted, KB was affiliated with EPFL.

\section{Preliminaries}

Throughout, for $n \in \N$,  $[n]$ denotes $\{0,1\dots,n-1\}$. 
A function $\mu : \N \rightarrow \Real$ is called negligible if for every polynomial $p : \N \rightarrow \Real $ we have $\lim_{n \rightarrow \infty} |\mu(n)| \cdot  p(n) = 0$. 
We use $\log$ to denote the logarithm with base $2$.

\subsection{Quantum mechanics}\label{sec:prelimquantum}

$\mathcal{H}$ always denotes a finite-dimensional Hilbert space, $L(\mathcal{H})$ denotes the set of linear operators in $\mathcal{H}$, $\id \in L(\mathcal{H})$ is the identity operator, and $\mathcal{P}(\mathcal{H})$ denotes the set of positive semi-definite linear operators from $\mathcal{H}$ to itself. 
We define the set of normalized quantum states 
\begin{align*}
\mathcal{S}(\mathcal{H}) := \Big\{\rho \in \mathcal{P}(\mathcal{H}) : \trace[\rho] = 1 \Big\}.
\end{align*}
We also define the set of pure states $\Omega(\hilb) := \{\ket{\lambda} \in \hilb \ | \braket{\lambda}{\lambda}  = 1\}$. 
We denote by $\omega(\hilb)$ the unique distribution over $\Omega(\hilb)$ invariant under unitaries.
For pure states, $\ket{\lambda} \in \Omega(\hilb)$, we often interchangeably use $\ket{\lambda}$ and $\ket{\lambda}\bra{\lambda}$.
We write $\hilb_{AB} = \hilb_A \otimes \hilb_B$ for a bipartite system and $\rho_{AB} \in \mathcal{S}(\hilb_{AB})$ for a bipartite state. $\rho_A = \trace_B[\rho_{AB}], \rho_B = \trace_A[\rho_{AB}]$ denote the corresponding reduced densities.  

A collection of $k$ operators $\mathcal{A} = \{\mathcal{A}_i\}_{i=0}^{k-1}$, where $\mathcal{A}_i \in \mathcal{P}(\hilb)$ is called a \emph{$k$-outcome} POVM if
$\sum_{i=0}^{k-1} \mathcal{A}_i = \id$.
The probability of obtaining outcome $i \in [k]$ when $\mathcal{A}$ is measured on a state $\rho \in \mathcal{S}(\hilb)$ is equal to $p(i |\mathcal{A}) = \trace[\mathcal{A}_i \rho]$.
For our applications we can, without loss of generality, assume that all $\mathcal{A}_i$'s are rank one. This is because any measurement with operators of rank larger than one can always be realized as a measurement with rank-one operators. 
More concretely for every $i \in [k]$ we can write $\mathcal{A}_i = \sum_j \eta_i^{(j)} P_i^{(j)}$ with $\eta_i^{(j)} \in (0,1]$ and $P_i^{(j)}$ being the rank-one projectors of $\mathcal{A}_i$. 
Now $\mathcal{A}_{i,j} = \eta_i^{(j)} P_i^{(j)}$ is a POVM, which is a \say{finer grained} version of $\{\mathcal{A}_i\}$. 
To reproduce the statistics of the original POVM one simply applies the finer POVM and forgets the result $j$. 
Thus for every $i \in [k]$ we can write $\mathcal{A}_i = \eta_i P_i$, where $\eta_i \in (0,1]$ and $P_i$ is a rank-one projection. 

For $X,Z \in L(\hilb)$ we write $\{X,Z\} = XZ + ZX$ to denote the anticommutator, $\sigma_X,\sigma_Y,\sigma_Z \in L(\Com^2)$ are the single-qubit Pauli matrices. 
For $\theta \in \left\{0,\frac{\pi}{2}, \pi, \frac{3\pi}{2} \right\}$ we write $\ket{+_\theta} = \frac{1}{\sqrt{2}} \left(\ket{0} + e^{i\theta}\ket{1} \right)$.


\subsection{Entanglement and Nonlocality}\label{sec:history}

In this section, we introduce the definitions of entanglement and nonlocality. 
We recommend \cite{entnonsurvey} for a detailed survey about the relationship between the two notions. 

We imagine that there are two spatially separated parties $\alice$ and $\bob$ that share a state $\rho_{AB} \in \mathcal{S}(\hilb_A \otimes \hilb_B)$ or access to public randomness.
They interact with a verifier $\ver$ who sends and collects classical messages.
Next, we define a notion of entanglement.

\begin{definition}[Entanglement] For $\rho_{AB} \in \mathcal{S}(\mathcal{H}_A \otimes \mathcal{H}_B)$ we say it is \textbf{separable} if it can be expressed as $\rho_{AB} =\sum_i p_i \ \sigma_A^{(i)} \otimes \sigma_B^{(i)}$ for some $\sigma_A^{(i)} \in \mathcal{S}(\hilb_A), \sigma_B^{(i)} \in \mathcal{S}(\hilb_B), p_i > 0, \sum_i p_i = 1$. Otherwise, we call it \textbf{entangled}. 
\end{definition}

The most famous example of an entangled state is the Bell state, also known as an EPR pair, $\ket{\phi^+} = \frac{1}{\sqrt{2}}(\ket{00} + \ket{11})$. The notion of entanglement is a purely mathematical notion and a priori doesn't carry any operational meaning.

Next, we introduce the notion of nonlocality in a way that is often presented in the physics literature. This definition changes once we consider computational aspects (see Definition~\ref{def:nonlocalityinformal} and~\ref{def:nonlocality}).

\begin{definition} [Nonlocal] For a probability distribution $\Prob(a,b \ | \ x,y)$ we say that $\Prob$ is \textbf{local} if there exist probability distributions $\Prob_1(a \ | \ x,\lambda)$ and $\Prob_2(b \ | \ y,\lambda)$ and $\lambda$ such that 
\begin{equation}\label{eq:nonlocalphysicsdef}
    \Prob(a,b \ | \ x,y) = \int_{\lambda}  d\lambda \ \Prob_1(a \ | \ x,\lambda) \ \Prob_2(b \ | \ y,\lambda),
\end{equation}
where $\lambda$ is understood as a local hidden variable. Operationally this means that there exist $\alice$ and $\bob$, sharing randomness $\lambda$, such that their joint distribution replicates $\Prob$. 
If a probability distribution is not local, we call it \textbf{nonlocal}. 
\end{definition}

The Bell experiment (\cite{bell}), also known as the CHSH game, is one of the first examples that show the nonlocality of quantum correlations. 
In this game, $\alice$ and $\bob$ are given uniformly random single-bit inputs, $x$ and $y$ respectively, and are expected to reply with single-bit answers, $a$ and $b$, such that $x \land y \stackrel{(*)}{=} a\oplus b$. 
It can be shown that for any local strategy. i.e. satisfying \eqref{eq:nonlocalphysicsdef}, the probability of $(*)$ is upper-bounded by $75\%$. 
It can also be shown that if the parties share a $\ket{\phi^+}$, they can satisfy $(*)$ with $\approx 85\%$.
Hence, this is an example of a non-local probability distribution, i.e. $\ket{\phi^+}$ is non-local. 
This game certifies non-locality of $\ket{\phi^+}$.
The term \say{certified} usually means that we assume quantum theory but the Bell experiment proves more, i.e. that a theory governing the behavior of $\alice,\bob$ must contain some notion of entanglement. 
Later, the result from \cite{bell} was improved (\cite{gisinpurestatesarenonlocal}) to show that all pure entangled states are non-local.

As we mentioned in the introduction, for some time it was believed that entanglement equals nonlocality in the Bell scenario.
A surprising result was presented in \cite{wernerprojective}, where it was shown that for a class of entangled states, every distribution obtained by projective measurements can be simulated locally. 
Let us give more details. 

For $p \in [0,1]$ let $\rho(p)$ be a state in $\Com^2 \otimes \Com^2$ be defined as 
\begin{equation}\label{eq:wernerstate}
\rho(p) := p \ket{\psi_-}\bra{\psi_-} + \frac{1-p}{4} \id,
\end{equation}
where $\ket{\psi_-} = \frac{1}{\sqrt{2}}(\ket{01} - \ket{10})$.
These states are known as the Werner states. It was shown in \cite{wernerprojective} that $\rho_{AB}(p)$ is entangled if and only if $p > \frac13$.

\begin{figure}
    \centering
    \includegraphics{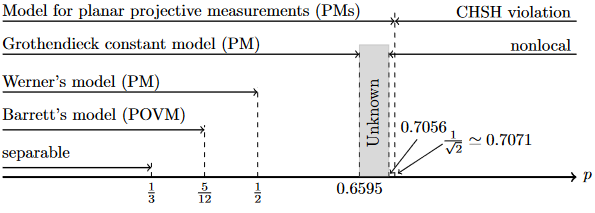}
    \caption{Properties of the two qubit Werner state $\rho(p) = p\ket{\psi_-}\bra{\psi_-} + (1-p)\frac{\id}{4}$, for $p \in [0,1]$. State $\rho(p)$ is entangled exactly for $p>\frac13$ however all POVMs can be simulated locally as long as $p<\frac{5}{12}$.}
    \label{fig:ranges}
\end{figure}

However, Werner showed that for any $p \leq \frac12$, all projective measurements on $\rho(p)$ can be simulated locally, i.e. expressed as in \eqref{eq:nonlocalphysicsdef}. This result was further generalized to all POVMs when $p<5/12$ (\cite{barrettmodel}). 
Figure~\ref{fig:ranges} summarizes the known results for $\rho(p)$ depending on the value of $p$. 

\subsection{Complexity theory}\label{sec:prelimcomplexity}

A language is a function $L : \{0,1\}^* \rightarrow \{0,1\}$.
The complexity classes are collections of languages recognizable by a particular model of computation. 
The classes of interest are $\BPP$, i.e. classical randomized polynomial time, $\BQP$, i.e. quantum polynomial time, and $\PP,\PostBQP$ that we define below.
We also consider $\SampBQP$, which is a sampling class that we also define below.

The class $\PP$ consists of problems solvable by an $\NP$ machine such that (i) if the answer is \say{yes} then at least $\frac12$ of computation paths accept, (ii) if the answer is \say{no} then less than $\frac12$ of computation paths accept. 

In the seminal result \cite{PostBQPPP} it was shown that $\PostBQP = \PP$. 
The class $\PostBQP$ is the class of languages recognizable by a uniform family of polynomially sized $\BQP$ circuits with the ability to post-select. 
This is the ability to post-select on a particular qubit being $\ket{1}$. 
More concretely, for a state $\sum_{x \in \{0,1\}^n} \alpha_x \ket{x}$, if the post-selection is applied to the first qubit then the resulting state is
$$\left(\sum_{x' \in \{0,1\}^{n-1}} |\alpha_{1x'}|^2 \right)^{-1} \sum_{x' \in \{0,1\}^{n-1}} \alpha_{1x'} \ket{1x'}.
$$
We can think of post-selection as a new 1-qubit gate that can be applied to a chosen qubit at any place in the circuit.

Sampling problems are problems, where given an input $x \in \nbits$, the goal is to sample (exactly or, more often, approximately) from some probability distribution $\dist_x$ over $\poly (n)$-bit strings.
$\SampBQP$ is the class of sampling problems that are solvable by polynomial-time quantum computers, to within $\e$ error in total variation distance, in time polynomial in $n$ and $1/\e$.

\subsection{Delegation of Quantum Computation (DQC)}

We give an overview of the history of DQC, but for a more in-depth review of DQC protocols, we refer the reader to a survey \cite{andrusurvey}. 

Arguably, the first time the question of delegating quantum computation was formalized was in \cite{scott25challenge}. 
It asks if a quantum computer can convince a classical observer that the computation it performed was correct. 
More formally we ask if there exists an interactive proof for $\BQP$, such that the honest prover is in $\BQP$, the verifier is in $\BPP$, and the protocol is sound against all adversaries. 
The importance of this question was later emphasized in \cite{vaziranifalsifiablequantum}, where it was argued that it has important philosophical implications for the testability of quantum theory. 

Before we delve deeper into details we bring the reader's attention to how $\BQP$ fits wrt to classical complexity classes. 
Firstly, it is believed that $\BQP \not\subseteq \NP$ (it is known that $\BQP$ is not in $\PH$ relative to an oracle \cite{talPH}), which means that, most likely, there does not exist an $\NP$-like witness for all problems in $\BQP$. 
This shows that the DQC problem is nontrivial.
Secondly, a simple argument (\cite{vaziraniBQP}) shows that $\BQP \subseteq \PSPACE$, which combined with the classical result $\IP = \PSPACE$ shows that if we didn't require the prover to be efficient then a DQC protocol had already existed. 

It turns out that one can verify $\BQP$ computations albeit in slightly modified settings. 
For example, if one considers the so-called multiprover systems, i.e. where the verifier interacts with two or more $\BQP$ provers, then protocols for $\BQP$ delegation are known to exist (\cite{MIPforBQP1,MIPforBQP2,MIPforBQP3}). 
On the other hand, protocols for delegation are also known in a setting with only one prover but where the verifier has access to a constant-size quantum computer (\cite{preparesend}, \cite{receivemeasure}).
This model falls under the category of $\QIP$. 
In a recent breakthrough, it was shown (\cite{mahadev}) that under the assumption that the Learning with Errors problem is quantumly hard, i.e., $\text{LWE} \not\in \BQP$, one can delegate $\BQP$ to a single prover.
In this result, the soundness guarantee holds against $\BQP$ adversaries.
The breakthrough of this protocol is in the fact that communication is purely classical.
However, the result was achieved at the cost of limiting the power of dishonest provers to $\BQP$.
Thus \cite{mahadev} didn't fully answer the initial question as we hoped for soundness against all dishonest provers.

Not long after, several new variants of the Mahadev protocol appeared in the literature. 
In \cite{vidick} it was shown how to build a blind version of this protocol by forcing the prover to prepare quantum states blindly. 
A \textit{blind} delegation is such that the prover can not distinguish the computation it was asked to perform from any other of the same size.
In \cite{alagic}, using a Fiat-Shamir-like argument and parallel repetition of the Mahadev protocol, the authors showed how to implement a protocol in $1$-round in the quantum-random-oracle-model (QROM). 
This shows that the quantum hardness of LWE + QROM gives sufficient assumptions required to achieve a $\BQP$-sound, $1$-round DQC.

\section{A new model - Not-efficiently-local}\label{sec:notefflocal}

We introduce a new definition of nonlocality that takes into account the computational power of $\alice$ and $\bob$. As our model is different from the standard setup we define the setup very carefully.

\begin{definition}[Distinguishing Game]\label{def:game}
For $k,\ell \in \N, \rho_{AB} \in \mathcal{S}(\mathcal{H}_A \otimes \mathcal{H}_B)$ we define $\mathcal{G}(\rho_{AB},k,\ell)$ to be a game between $\alice, \bob$ and $ \ver$. $\mathcal{G}$ is played in one of two modes. $\alice,\bob$ will have access to either (i) $\rho_{AB}$ or (ii) a separable (\say{classical}) state $\sigma_{AB}$. First, the hyperparameters $k,\ell$ are distributed to all parties and one of the modes is chosen. The mode is not known to $\ver$. $\mathcal{G}$ proceeds in $k$ rounds. 

In each round $\alice$ and $\bob$ are given their respective share of $\rho_{AB}$ in mode (i) or of $\sigma_{AB}$ in mode (ii) and are forbidden to communicate. Then
\begin{enumerate}
    \item $\ver$ sends a message $x$ to $\alice$ and $y$ to $\bob$, where $x,y \in \{0,1\}^{p(\ell)}$,\footnote{$p$ is a fixed polynomial.}
    \item $\alice, \bob$ compute their answers $a,b \in \{0,1\}^{q(\ell)}$.\footnote{$q$ is also a fixed polynomial.} 
    They can operate on their respective shares of $\rho_{AB}$ in mode (i) or access their share of $\sigma_{AB}$ in mode (ii). 
    The computational modeling of the behavior of $\alice,\bob$ is discussed below.
    \item Answers $a,b$ are sent to $\ver$, which stores them.
\end{enumerate}
Then, the next round starts. 
After the $k$-th round, $\ver$ outputs, based on $a,b$'s, either YES or NO.
\end{definition}


Next, we define a new notion of nonlocality that we call $\NEL$, standing for not-efficiently-local. We assume the extended quantum Church-Turing hypothesis that any computation in the physical world can be modeled by a polynomial time quantum machine. This essentially translates to limiting the power of dishonest parties to $\BQP$. More concretely we define  

\begin{definition}[Not-efficiently-local]\label{def:nonlocality}
For a state $\rho_{AB} \in \mathcal{S}(\mathcal{H}_A \otimes \mathcal{H}_B)$ we say that $\rho_{AB}$ is not-efficiently-local if for every sufficiently small $\delta$ there exists $k \in \N$, a game $\mathcal{G}(\rho_{AB},k,\cdot)$, $\alice(\cdot),\bob(\cdot) \in \text{QPT}(\cdot), \ver \in \text{PPT}(\cdot)$ and a polynomial $p$
such that for every $\ell \in \N$
\begin{enumerate}
    \item{(\textbf{Completeness})} If $\mathcal{G}(\rho_{AB},k,\ell)$ was run in mode (i) with $\alice(\ell), \bob(\ell)$ then $$\Prob[\ver \text{ accepts}] \geq 1 - \delta,$$
    \item{(\textbf{Soundness})} For \textbf{every} $\alice'(\cdot),\bob'(\cdot) \in \text{QPT}(\cdot)$ if $\mathcal{G}(\rho_{AB},k,\ell)$ was run in mode (ii) with $\alice',\bob'$ then $$\Prob[\ver \text{ accepts}] \leq 1 - \delta - \frac{1}{p(\ell)}.$$
\end{enumerate}
\end{definition}

\paragraph{Modeling.} 
Whenever we say that $\alice \in \text{QPT}$ we mean that for every $\ell$, $\alice$'s actions can be modeled by a polynomial-sized quantum circuit that is generated by a polynomial-time Turing machine run on $1^\ell$.
This essentially translates to $\alice,\bob \in \BQP$, as in both cases it implies that $\alice,\bob$ are solving problems in $\BQP$.
However, there is a nuance that in our setup $\alice$ and $\bob$ are \textit{sampling} answers from a distribution, thus writing $\alice,\bob \in \BQP$ could be misleading. 
Moreover, as we already discussed (Section~\ref{sec:techhentanglementiimpliescrypto}) once we start considering the separation of complexity classes (Section~\ref{sec:entanglementcertimpliescrypto}) the distinction between $\BQP$ and $\SampBQP$ starts becoming important.

More concretely, for every $\ell \in \N$, $\alice(\ell)$'s answer is the result of applying some polynomially sized circuit with constant sized qubit gates on $\rho_A \otimes \ket{x} \otimes \ket{0^{q(\ell) - p(\ell) - 1}}$ for some polynomial $q(\ell)$ and measuring all\footnote{Wlog we can assume that $\alice$ and $\bob$ measure all the qubits. It is because honest acting parties sending a superset of information would also yield a valid protocol as $\ver$ can just ignore the extra bits.} qubits in the computational basis. 
We may assume that the gates come from a universal gate set.  
It is because the Solovay-Kitaev theorem guarantees that we can approximate any circuit consisting of constant-qubit gates to within $\e$ error by incurring a multiplicative blow-up of $\text{polylog}(1/\e)$ in the number of gates. 
Finally, we use a result from \cite{ToffoliH} to argue that we can assume that the gate set used is equal to $\{\text{Toffoli},\text{Hadamard}\}$.
This choice incurs another polylogarithmic in $(\ell,1/\e)$ blowup in the number of gates. 
Note that the gate set $\{\text{Toffoli},H\}$ is not universal in the standard sense as both matrices contain only real entries. 
However, it is enough for our purposes as we are interested in computational universality only (see \cite{ToffoliH}).

To summarize, there exists a uniform family of circuits $\left\{C^{\alice}_\ell \right\}_{\ell \geq 1}$  acting on $q(\ell)$ qubits with $t(\ell)$ gates (for some polynomials $q,t$) 
such that for every $\ell \in \N$ we have that $\alice(\ell)$'s answer is the result of applying $C_\ell^\alice$ on 
\begin{equation}\label{eq:howAacts}
\rho_A \otimes \ket{x} \otimes \ket{0}^{q(\ell) - p(\ell) -1}
\end{equation}
and measuring all qubits in the computational basis. 
For $\bob$ the situation is analogous, i.e. there is a uniform family of circuits $\left\{C^\bob_\ell \right\}_{\ell \geq 1}$ such that the answer of $\bob(\ell)$ is the result of applying $C_\ell^\bob$ on $\rho_B \otimes \ket{y} \otimes \ket{0}^{q(\ell) - p(\ell) -1}$ and measuring all qubits in the computational basis.\footnote{Note that we can assume wlog that both circuits operate on $q(\ell)$ qubits.} 

We require $\ver$ to be efficient, i.e. $\ver \in \text{PPT}$.
It is also possible to consider a definition where $\ver$ is unbounded.
This choice follows the usual setup of interactive proof systems, where $\ver$ needs to be efficient.
We note that our second result (Section~\ref{sec:entanglementcertimpliescrypto}) holds also for the more challenging definition where $\ver$ can be unbounded. 

\paragraph{Access to $\sigma_{AB}$.} In mode (ii) of the game (Definition~\ref{def:game}) $\alice,\bob$ have access to a separable state $\sigma_{AB} \in \mathcal{S}(\hilb_A \otimes \hilb_B), \ \sigma_{AB} = \sum_i p_i \ \sigma_A^{(i)} \otimes \sigma_B^{(i)}$. We assume that $\alice$ and $\bob$ have special registers, where, at the beginning of each round, $i$ is sampled according to $p_i$ and a state $\sigma_A^{(i)}$ is placed in $\alice$'s register and $\sigma_B^{(i)}$ is placed in $\bob$'s register. Note that the assumption that $\alice,\bob \in \text{QPT}(\ell)$ implies that $\sigma_{AB}$ is a state on at most $\poly(\ell)$ qubits. 

Shared public randomness can be simulated by access to a separable state.
It is because any random string $r \in \nbits$ can be represented as $\ket{r}_A \otimes \ket{r}_B$.
Moreover, note that the number of qubits needed to represent it is equal to the number of bits of randomness.
However, it is not clear that any separable state can be replaced by public randomness.
When there are no restrictions on computational capabilities of $\alice,\bob$ then it is possible. 
This is the case because any distribution $\{p_i\}$ can be approximated to high precision with access to an infinite string of common randomness and every separable state on finite-dimensional space can be approximated by local operations.
Thus one can achieve an inverse polynomial in $\ell$ separation as the definition requires.
But once we limit the computational power of $\alice,\bob$ the equivalence is not clear. 

We prove $\text{LWE} \not\in \mathsf{BQP/qpoly} \Rightarrow \ENT = \NEL$ in a model where cheating parties share a separable state.
This, according to the discussion above, is a stronger result than considering public randomness only.
In the second result, i.e. $\ENT = \NEL \Rightarrow \BQP \neq \PostBQP$, assume the model where the honest parties share public randomness.
As discussed this is a stronger result.

\section{Cryptography implies Entanglement Certification}\label{sec:cryptoimpliesentcert}

In this section, we formally show how to design a protocol certifying the entanglement of all entangled states assuming the existence of post-quantum cryptography, specifically $\text{LWE} \not\in \mathsf{BQP/qpoly}$, known as the (non-uniform) QLWE assumption. 
As we discussed our protocol is an amalgamation of a Remote State Preparation (RSP) protocol and a Semi-Quantum Game (SQG) for certifying all entangled states.
We present these two ingredients in Sections~\ref{sec:RSP} and~\ref{sec:quantuminputsprotocol} and then combine them to arrive at the final protocol in Section~\ref{sec:final}.

\subsection{Remote State Preparation (RSP)}\label{sec:RSP}

\begin{figure}
    \rule{\textwidth}{0.6pt}
    Let $\ell \in \N$ be a security parameter.
    \begin{enumerate}[itemsep=5pt]

        \vspace{5pt}
        
        \item[R1.] The verifier selects $G \gets_{U} \{0,1\}$. If $G = 0$ they sample a key $(k,td) \gets \text{GEN}_\mathcal{F} (1^\ell)$. If $G = 1$ they sample $(k, td) \gets \text{GEN}_\mathcal{G}(1^\ell)$. The verifier sends the key $k$ to the prover and keeps the trapdoor information $td$ private.
        \item[R1.] The prover returns a $y$ to the verifier. If $G = 0$, for $b \in \{0, 1\}$, the verifier uses $td$ to compute $\hat{x}_b \gets f^{-1}_{pk} (y,b,td)$. If $G = 1$, the verifier computes $(\hat{b},\hat{x}_{\hat{b}}) \gets g^{-1}_{pk}(y,td)$.
        \item[R2.] The verifier samples a type of round uniformly from $\{2a,2b\}$ and performs the corresponding of the following
        \begin{enumerate}[itemsep=5pt]
            
            \vspace{5pt}
            
            \item[R2a.] (\textit{preimage test}) The verifier expects a preimage. The prover returns $(b, x)$. If $G = 0$ and $\hat{x}_b \neq x$, or if $G = 1$ and $(b, x) \neq (\hat{b}, \hat{x}_{\hat{b}})$, the verifier \textsc{Aborts}.
            \item[R2b.] (\textit{measurement test}) The verifier expects an equation $d \in \Z^{n/2}_4$ from the prover. If $G = 0$, the verifier computes $\widehat{W} \gets \widehat{W}(d) , \hat{v} \gets \hat{v}(d)$, i.e. $\widehat{W} + 2 \hat{v} = d  (x_1 - x_0) \text{ mod } 4$. 
            
            \vspace{2pt}
            
            The verifier samples a type of round uniformly from $\{3a,3b\}$ and performs the corresponding of the following
            \begin{enumerate}[itemsep=5pt]

                \vspace{5pt}
            
                \item[R3a.] (\textit{consistency check}) The verifier samples $c \gets_U \left\{X,\frac{X-Y}{\sqrt{2}},Y,\frac{X+Y}{\sqrt{2}},Z \right\}$, sends $c$ to the prover and expects $v \in \{0,1\}$ back. 

                \vspace{2pt}
                
                If $c = Z$ and $G = 1, v \neq \hat{b}$ the verifier \textsc{Aborts}. If $c \in \left\{X,\frac{X-Y}{\sqrt{2}},Y,\frac{X+Y}{\sqrt{2}}\right\}, G = 0$ then

                \begin{enumerate}
                    \item[A.] If $c = \widehat{W}, v \neq \hat{v}$ the verifier \textsc{Aborts}.
                    \item[B.] If $\widehat{W} \in \{X,Y\}$ and $c \in \left\{\frac{X-Y}{\sqrt{2}},\frac{X+Y}{\sqrt{2}}\right\}$, performs a QRAC test on $v$. 
                \end{enumerate} 
                \item[R3b.] (\textit{succesful state preparation}) The verifier knows that, if $G = 0$, the prover holds $|\hat{b} \rangle$, and if $G = 0$, the prover holds $\ket{+_{\frac{\pi}2(\widehat{W} + 2 \hat{v})}}$.
            \end{enumerate}
        \end{enumerate}
    \end{enumerate}
    \rule{\textwidth}{0.6pt}
    \caption{RSP protocol.} 
    \label{fig:qubit_preparation_protocol}
\end{figure}

This section describes the RSP protocol. 
More concretely we show how to delegate the preparation of a random state out of the set $\{\zero,\one,\plus,\minus,\ii,\minusi \}$. 
We formally define the protocol in Figure~\ref{fig:qubit_preparation_protocol}.
The RSP protocol takes place between two parties that we will call $\ver$ and $\prov$, as in prover.
Later on, two instances of this protocol will be used for communication between $\ver$ and $\alice$ and between $\bob$ and $\ver$. 
The RSP will be then naturally extended, via an amalgamation with the result from \cite{Buscemi_SemiQuantumGame}, to give rise to the final protocol certifying all entangled states. 

The protocol is based on a set of trapdoor claw-free (TCF) functions $\cF=\{f\}$ which are assumed to have three important properties.
First, they are 2-to-1: any image $y$ has exactly two preimages $(0,x_0), (1,x_1)$ such that $y=f(0,x_0)=f(1,x_1)$\footnote{Note that formally $f$ is 1-to-1 not 2-to-1 as it accepts an additional input bit apart from $x_0/x_1$. 
It is done to split the domain into two sets, i.e. the 0-preimages and 1-preimages. It is still best to think of $f$ as a 2-to-1 function.}.
Second, there is a trapdoor: anyone who \emph{selected} some $f\in\cF$ can easily invert it, i.e. find the two preimages corresponding to any image $y$.
Third, they are \textit{claw free}: anyone who is \emph{given} some $f\in\cF$ can provide an image $y$, a preimage $(0,x_0)$ (or $(1,x_1)$), but cannot find even a single bit of information about $x_1$ (or $x_0$).
Our protocol also uses a second set of functions $\cG=\{g\}$ which are  1-to-1 (any image $y$ has exactly one preimage $(b,x)$ such that $y=g(b,x)$) and invertible by anyone who \emph{selected} it.
At last, these two sets of functions are assumed to be indistinguishable: anyone who is \emph{given} a function $h\in\cH:=\cF\cup\cG$ in one set or the other cannot guess from which of the two sets $\cF$ or $\cG$ it comes from.
All these properties are not satisfied in absolute, but only when the prover is limited to be in $\BQP$.
Families $\cF,\cG$ can be constructed based on a cryptographic assumption that LWE is quantumly hard. 
Formal properties of such functions are given in Appendix~\ref{apx:clawfree}.

\begin{note}
When discussing the RSP protocol we will use $x,y$ to denote elements of the the domain and range of $f$ respectively and \textbf{not} questions to $\alice$ and $\bob$ as in Definition~\ref{def:game}.
We do that to be consistent with the notation of most of the protocols based on claw-free functions and also to be consistent with the standard notation for Bell games.
We think this conflict of notation is not problematic as it is only present when discussing the RSP protocol, which is used a modular way.
\end{note}

Let us now present the RSP protocol in more detail, which is a three-round protocol (see Fig.~\ref{fig:RSP_protocol}). This proves its \emph{completeness}: an honest prover with access to a $\BQP$ machine can succeed by following the procedure below.  


\begin{figure}
    \centering
    
    \includegraphics[scale=0.35,trim={0 2cm 0 0},clip]{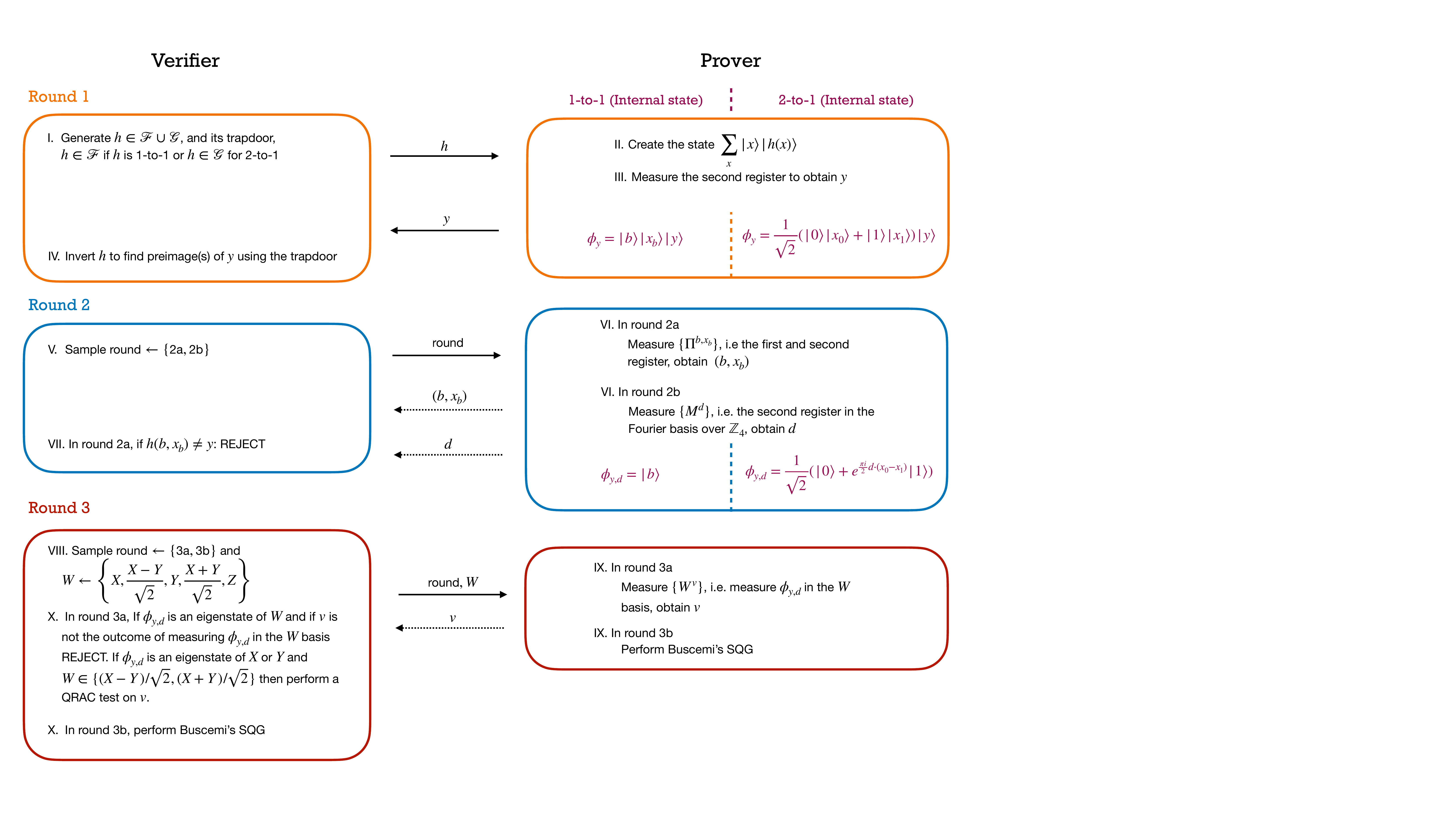}
    
    \caption{Description of the RSP protocol.}
    \label{fig:RSP_protocol}
\end{figure}

\textit{RSP, Round 1:} The verifier generates (\rom{1}) either a function $f\in\cF$ or $g\in\cG$ which from now on we call $h$.
He sends a description of this function to the prover.
The prover creates (\rom{2}) the state $\ket{\phi}=\sum_{b,x} \ket{b}\ket{x}\ket{h(b,x)}$. Next, she measures (\rom{3}) the last register, sending the output $y$ to the verifier.
Note that when $h=f$, the prover's first register is projected on the state $(\ket{0}\ket{x_0}+\ket{1}\ket{x_1})\ket{y}$ with $f(0,x_0)=f(1,x_1)=y$. 
When $h=g$, it is projected on the state $\ket{b}\ket{x} \ket{y}$ with $g(b,x)=y$.
Importantly, the verifier upon receiving $y$ can invert (\rom{4}) $h$ hence can deduce the prover's state, which the prover cannot do according to our cryptographic assumption.
In the following, the verifier exploits this additional knowledge to check that the prover is not cheating. 
For the remaining rounds, he randomly samples (\rom{5}) one of two types of rounds, 2a or 2b, followed by 3a or 3b. The "a" rounds are designed to check that the prover is not cheating and abort the protocol, while the "b" rounds lead to the preparation of a state.

\textit{RSP, Round 2a:} 
Half of the time, as a consistency check the verifier asks for a preimage of $y$. 
He requests the prover to measure (\rom{6}) the first and second registers in the computational basis, and checks that the provers replied $(b,x)$ are indeed satisfying $y=h(b,x)$. 

\textit{RSP, Round 2b:} 
The rest of the time, the verifier sends a measurement challenge.
He requests the prover to measure (\rom{6}) the second register in the $\Z_4-$Fourrier basis and to send back the outcome $d$.
At this stage, the verifier can compute what state is held by the prover (up to an irrelevant phase). 
When $h=f$, it is $\ket{0}+e^{i\theta}\ket{1} \in \{\plus,\minus,\ii,\minusi \}$ where $\theta=\frac{\pi}{2} \ d\cdot(x_1-x_0)$ (modulo $\Z_4$).
When $h=g$, it is $\ket{b} \in \{\zero,\one\}$.

\textit{RSP, Round 3a:} 
Half of the time, as a consistency check the verifier asks for the following. 
He selects (\rom{8}) at random a basis $X, \frac{X-Y}{\sqrt{2}}, Y, \frac{X+Y}{\sqrt{2}}$ or $Z$ and asks the prover to measure (\rom{9}) the first register in that basis and send back the outcome $v$.
The verifier checks (\rom{10}) that $v$ is \say{consistent}. This check has two modes. If the prover should hold an eigenstate of the chosen basis the verifier checks if $v$ has the right value. If the prover should hold an eigenstate of $X$ or $Y$ and the chosen basis is $\frac{X-Y}{\sqrt{2}}$ or $\frac{X+Y}{\sqrt{2}}$ then the verifier performs~\footnote{By collecting statistics over many runs of the protocol} a 2-to-1 quantum random access code test (QRAC), which can be thought of as a version of the maximal violation of the CHSH inequality.\footnote{The choice of basis $X,\frac{X-Y}{\sqrt{2}}, Y,\frac{X+Y}{\sqrt{2}}$ or $Z$ should be done at random: if not, the choice of basis itself would provide additional information to the prover, who could, for instance, guess that $h\in\cF$ when she is asked for measurement in basis $Z$}


\textit{RSP, Round 3b:} The rest of the time, the verifier knows what state is held by the prover, and can ask the prover to use it in any subsequent protocol, e.g. perform (\rom{9}) the Buscemi's SQG. At the end the verifier performs (\rom{10}) any required followup actions.

\paragraph{Comparison to \cite{vidick}}
The biggest difference between our protocol and that of \cite{vidick} is that in their protocol the prover prepares eigenstates of all of the five bases $\left\{X,\frac{X-Y}{\sqrt{2}},Y,\frac{X+Y}{\sqrt{2}},Z \right\}$.
In our protocol, on the other hand, she is only required to prepare eigenstates of $\left\{X,Y,Z\right\}$.
This stems from the fact that we only need the following, tomographically complete set of states to be prepared $\{\zero,\one,\plus,\minus,\ii,\minusi \}$.
Or to phrase it differently we don't need the eigenstates of $\left\{\frac{X-Y}{\sqrt{2}},\frac{X+Y}{\sqrt{2}} \right\}$.
This difference is reflected, for instance, in the measurement test (R2b. in Figure~\ref{fig:RSP_protocol}), where we expect an equation in $\Z_4^{n/2}$ instead of $\Z_8^{n/3}$.
This change makes our protocol simpler but the overall structure of the scheme and the proof strategy remains similar.
It is important to note that the checks in the protocol still require performing measurements in the following bases $\left\{\frac{X-Y}{\sqrt{2}},\frac{X+Y}{\sqrt{2}}\right\}$ (see R3aB. in Figure~\ref{fig:RSP_protocol}).

\subsubsection{QRAC test}

One of the building blocks of our protocol is a QRAC test.
Intuitively, this procedure self-tests that $4$ quantum states and two observables are such that, up to isometry, there are the $4$ eigenstates of $X,Y$ and the two observables are $\frac{X-Y}{\sqrt{2}},\frac{X+Y}{\sqrt{2}}$.
For more details about QRACs, we recommend a survey \cite{qracs}. 
We define the test formally as follows.

\begin{definition}
A quantum random access code (QRAC) is specified by four single-qubit density matrices $\{\phi_u\}_{u \in \{1,3,5,7\}}$ and two single-qubit observables $X,Y$. 
For $u \in  \{1, 3, 5, 7\}$ let $u_0, u_2 \in \{0, 1\}$ be such that $u_0 = 0$ if and only if $u \in \{1, 7\}$ and $u_2 = 0$ if and only if $u \in \{1, 3\}$.
The success probability of the QRAC is defined as    
$$
\frac14 \sum_{u \in \{1,3,5,7\} } \trace(X^u \phi_u) + \trace(Y^u \phi_u).
$$
\end{definition}

It's been shown (\cite{vidick}) that the optimal winning probability of a QRAC test is $OPT_Q := \frac12 + \frac1{2\sqrt{2}}$.
Moreover, if the QRAC test is passed with the maximal probability then $X,Y$ anticommute.
This is a crucial property we will use later in the proof.
More precisely, we will need a robust version of this statement that is stated in the Appendix (Lemma~\ref{lem:qrac}).

\subsubsection{RSP Completeness} 

Now we are ready to explain the behavior of an honest $\prov$ interacting in the protocol from Figure~\ref{fig:qubit_preparation_protocol}. 
We summarize the result with the following theorem.

\begin{theorem}[Completeness]\label{thm:completeness_for_qubit_prep}
For every $\ell \in \N$ there exists $\prov \in \text{QPT}(\ell)$ that wins in the protocol in Figure~\ref{fig:qubit_preparation_protocol} with probability 1.\footnote{We assume here that the R3a B check is not performed here. 
It is analyzed in detail in Theorem~\ref{thm:main_formal}.} 
Moreover
\begin{itemize}
    \item  When $G = 1$ the prover's post-measurement state after returning $d$, is
    $\ket{b} \in \{\zero, \one\}$,
    where $b$ is negligibly in $\ell$ close to uniform.
    \item When $G = 0$ the prover's post-measurement state after returning $d$, is
    $$
    \frac{1}{\sqrt{2}}(\ket{0} + e^{i \theta}\ket{1}) \propto \ket{+_\theta} \in \{\plus, \minus, \ii, \minusi\},
    $$
    where $\theta = \frac{\pi}{2} d (x_1 - x_0)$ and $\theta$ is negligibly in $\ell$ close to uniform.
\end{itemize}

\end{theorem}

\begin{proof}
We proceed by analyzing the protocol round by round.

\textit{Round 1:} The verifier generates (\rom{1}) either a function $h \in \cF$ or $h \in \cG$, i.e. if $G = 1$ then $h \in \cF$ and if $G = 0$ then $h \in \cG$. He sends a description, in the form of a public key $pk$, of $h_{pk}$ to $\prov$.
Upon receiving the public key $pk$, $\prov$ prepares a uniform superposition over the codomain and evaluates the function $h_{pk}$ on this superposition (\rom{2})
\begin{align*}
    &\sum_{b,x_b }\ket{b}_B\ket{x_b}_{\text{dom}}\ket{0}_{\text{range}} \rightarrow \\ 
    &\rightarrow \phi = \sum_{b,x_b}\ket{b}_B\ket{x_b}_{\text{dom}}\ket{h_{pk}(b,x_b)}_{\text{range}},
\end{align*}
where we split the input to $h_{pk}$ into two parts $b \in \{0,1\}$ and $x_b \in \{0,1\}^{n}$.
$\prov$ then measures (\rom{3}) the ${\text{range}}$ register in the computational basis, and sends the obtained image $y$ to $\ver$. 
The state $\phi_y$ after this action is, depending on whether the function was 2-to-1 ($h_{pk} \in \cF$) or injective (1-to-1, $h_{pk} \in \cG$), one of the following. 

\begin{equation}\label{eq:statefor1to1}
  \phi_{y} =
    \begin{cases}
      \frac{1}{\sqrt{2}} \Big(\ket{0}_B\ket{x_0}_{\text{dom}}+ \ket{1}_B\ket{x_1}_{\text{dom}} \Big)\ket{y}_{\text{range}}, & \text{if $h_{pk} \in \cF$, } \text{where~} h_{pk}(x_0)=h_{pk}(x_1)=y,\\
      \\
      \ket{b}_B\ket{x_b}_{\text{dom}}\ket{y}_{\text{range}}, & \text{if $h_{pk} \in \cG$, } \text{where~} h_{pk}(x_b)=y.
    \end{cases}       
\end{equation}

Importantly, $\ver$, upon receiving $y$, can invert (\rom{4}) $h_{pk}$ hence can deduce $\prov$'s state.

In the following, $\ver$ exploits this additional knowledge to check that $\prov$ is not cheating. 
For the remaining rounds, $\ver$ randomly samples (\rom{5}) one of two types of rounds, 2a or 2b, followed by 3a or 3b. The "a" rounds are designed to check that $\prov$ is not cheating and abort the protocol, while the "b" rounds lead to the preparation of a state.

\textit{Round 2a (round = 2a):}
In the preimage test $\ver$ requests $\prov$ to measure (\rom{6}) the $B$ and the ${\text{dom}}$ registers in computational basis, and then check that $\prov$ replied $(b,x)$ is indeed satisfying $y=h_{pk}(b,x)$. Because of \eqref{eq:statefor1to1} $\prov$ always succeeds. 

\textit{Round 2b (round = 2b):} 
In the measurement test $\ver$ requests $\prov$ to measure (\rom{6}) the ${\text{dom}}$ register in the Fourier basis in $\Z_4^{n/2}$, and returns the output $d \in \Z_4^{n/2}$. At this stage, $\ver$ can compute what state ($\phi_{y,d}$) is held by $\prov$ (up to an irrelevant phase), 
\begin{equation}
  \phi_{y,d} =
    \begin{cases}
      \frac{1}{\sqrt{2}} \Big(e^{\frac{\pi i}{2} d \cdot x_0}\ket{0}_B+e^{\frac{\pi i}{2} d \cdot x_1}\ket{1}_B \Big)\ket{d}_{\text{dom}}, & \text{if $h_{pk} \in \cF$, } \text{where } \theta = \frac{\pi}{2} d \cdot (x_1 - x_0), \\
      \\
      e^{\frac{\pi i}{2} d \cdot x_b}\ket{b}_B\ket{d}_{\text{dom}} \propto \ket{b}_B\ket{d}_{\text{dom}}, & \text{if $h_{pk} \in \cG$,}
    \end{cases}       
\end{equation}
where the inner products are taken modulo $4$. Simplifying
\begin{equation}\label{eq:state_after_fourier}
  \phi_{y,d} =
    \begin{cases}
      \frac{1}{\sqrt{2}} \Big(\ket{0}+e^{i\theta}\ket{1} \Big), & \text{if $h_{pk} \in \cF$,}\\
      \\
      \ket{b}, & \text{if $h_{pk} \in \cG$.}
    \end{cases}       
\end{equation}

\textit{Round 3a (round = 3a):}
Next, half of the time, $\ver$ asks for a consistency check. 
He selects (\rom{8}) at random a basis $c \in \left\{X,\frac{X-Y}{\sqrt{2}}, Y, \frac{X+Y}{\sqrt{2}}, Z \right\}$ and asks $\prov$ to measure (\rom{9}) the B register in that basis and send back the outcome $v$.
The verifier performs two types of checks on $v$ (\rom{10}). 
In the first type, if $\phi_{y,d}$ is an eigenstate of $c \in \{X,Y,Z\}$ and the chosen basis is the corresponding one then $\ver$ checks that $v$ corresponds to measuring $\phi_{y,d}$ in the $c$ basis. 
In the second type, if $c \in \left\{\frac{X-Y}{\sqrt{2}},\frac{X+Y}{\sqrt{2}}\right\}$ and $\phi_{y,d}$ is an eigenstate of $\{X,Y\}$, then $\ver$ perform a QRAC test on $v$ \footnote{As explained before this is done by collecting statistics over many runs of the protocol. 
The probability of the estimate being far from the actual expectation can be bounded by Chernoff bounds and it is analyzed in Theorem~\ref{thm:main_formal}.}. Because of \eqref{eq:state_after_fourier} $\prov$ passes the tests. 
Note that for the R3a B check the $\prov$'s reply is consistent with the maximizing configuration of the QRAC test, i.e. the states are eigenstates of $X,Y$ and the observables are $\frac{X-Y}{\sqrt{2}},\frac{X+Y}{\sqrt{2}}$. 

\textit{Round 3b (round = 3b):} The rest of the time, $\ver$ knows what $\phi_{y,d}$ is and can ask to use it in any subsequent protocol, e.g. perform (\rom{9}) the Buscemi's SQG. At the end, $\ver$ performs (\rom{10}) any required follow-up actions.
\end{proof}

\paragraph{Note.} We often think of $d \cdot (x_1 - x_0) \in \Z_4$ as an angle of a state that should have been prepared, i.e. $\ket{+_\theta} \in \{\plus, \minus, \ii, \minusi\}$, after a natural identification of $\Z_4 \cong \frac{\pi}{2} \cdot \Z_4$. When $\ver$ performs his checks, as in Figure~\ref{fig:qubit_preparation_protocol}, it first represents $d \cdot(x_1 - x_0) = \widehat{W} + 2 \hat{v}$ for $\widehat{W},\hat{v} \in \{0,1\}$. We think of $\widehat{W}$ as one of the angles $\{0,\frac{\pi}{2}\}$ and also identify it with one of $\{X,Y\}$ in a natural way. Furthermore $\hat{v}$ can be then identified with one of $\ket{+},\ket{-}$ in the case of $\widehat{W} = 0$ and one of $\ket{i},\ket{-i}$ in the case of $\widehat{W} = 1$. 
Such a treatment might be confusing at first as $\theta \in \{0,1,2,3\}$ but $\widehat{W} \in \{0,1\} \equiv \{X,Y\}$. We do that in order to be compatible with later definitions. Look for instance at Definition~\ref{def:adaptivebit}, where the representation in the form $\widehat{W} + 2\hat{v}$ becomes very handy.

\subsubsection{RSP Soundness}

The following is a formal statement of soundness.

\begin{theorem}\label{thm:mainRSP}
    Let $\epsilon \in \Real_+$ be a constant and $\ell \in \N$ large enough such that $\epsilon = \omega(\text{negl}(\ell))$. Let 
    $D$
    be a quantum polynomial time (in $\ell$) device that succeeds in the protocol in Figure~\ref{fig:qubit_preparation_protocol} with probability $1-\epsilon$. There exists a universal constant $c > 0$ and $\delta = O(\epsilon^c)$ and an efficiently computable isometry $V: \mathcal{H}_B \to \Com^2 \otimes \mathcal{H}_{B'}$ and a state $\ket{\text{AUX}} \in \mathcal{H}_{B'} \otimes \mathcal{H}_{B''}$ such that under the isometry $V$ the following holds.
\begin{itemize}
    \item When $G = 1$ the joint state of the challenge bit $b$ and the prover's post measurement state after returning $d$, is $\delta$ computationally indistinguishable from a state $$\sum_{b\in\{0,1\}} \ket{b}\bra{b}\otimes \ket{b}\bra{b} \otimes\ket{\text{AUX}}\bra{\text{AUX}}$$ 
    \item When $G = 0$ the joint state of $\widehat{W},\hat{v}$. and the prover's post measurement state after returning $d$ is $\delta$ computationally indistinguishable from a state
    \begin{align*}
    &\sum_{v \in \{0,1\}} \ket{X}\bra{X} \otimes \ket{v}\bra{v} \otimes \ket{+_{0 \cdot \frac{\pi}{2} + v \pi}}\bra{+_{0 \cdot \frac{\pi}{2} + v \pi}} \otimes \ket{\text{AUX}}\bra{\text{AUX}} + \\
    &\sum_{v \in \{0,1\}} \ket{Y}\bra{Y} \otimes \ket{v}\bra{v} \otimes \ket{+_{1 \cdot \frac{\pi}{2} + v \pi}}\bra{+_{1 \cdot \frac{\pi}{2} + v \pi}} \otimes \ket{\text{AUX}}\bra{\text{AUX}}.
    \end{align*}
\end{itemize}
\end{theorem}

We defer the proof to Appendix~\ref{apx:nonlocality}.
The proof strategy is similar to the one in \cite{vidick} but, as always, one needs to carefully verify that the simplified version of the protocol is still sound.

\subsection{Semi-quantum Game (SQG)}\label{sec:quantuminputsprotocol}

In this section, we describe a protocol from \cite{Buscemi_SemiQuantumGame} for certifying entanglement in a model, where trusted quantum inputs for $\alice$ and $\bob$ are allowed.
Such a setup is often called a semi-quantum game. 
This is the protocol that we later (Section~\ref{sec:final}) combine with the RSP protocol to obtain the final Entanglement Certification.

\begin{figure}
    \rule{\textwidth}{0.6pt}
    Let $\ell \in \N$ be a security parameter.
    \begin{enumerate}[itemsep=5pt]

        \vspace{5pt}
        
        \item[R1.] The verifier selects $G \gets_{U} \{0,1\}$. If $G = 0$ they sample a key $(k,td) \gets \text{GEN}_\mathcal{F} (1^\ell)$. If $G = 1$ they sample $(k, td) \gets \text{GEN}_\mathcal{G}(1^\ell)$. The verifier sends the key $k$ to the prover and keeps the trapdoor information $td$ private.
        \item[R1.] The prover returns a $y$ to the verifier. If $G = 0$, for $b \in \{0, 1\}$, the verifier uses $td$ to compute $\hat{x}_b \gets f^{-1}_{pk} (y,b,td)$. If $G = 1$, the verifier computes $(\hat{b},\hat{x}_{\hat{b}}) \gets g^{-1}_{pk}(y,td)$.
        \item[R2.] The verifier samples a type of round uniformly from $\{2a,2b\}$ and performs the corresponding of the following
        \begin{enumerate}[itemsep=5pt]
            
            \vspace{5pt}
            
            \item[R2a.] (\textit{preimage test}) The verifier expects a preimage. The prover returns $(b, x)$. If $G = 0$ and $\hat{x}_b \neq x$, or if $G = 1$ and $(b, x) \neq (\hat{b}, \hat{x}_{\hat{b}})$, the verifier \textsc{Aborts}.
            \item[R2b.] (\textit{measurement test}) The verifier expects an equation $d \in \Z^{n/2}_4$ from the prover. If $G = 0$, the verifier computes $\widehat{W} \gets \widehat{W}(d) , \hat{v} \gets \hat{v}(d)$, i.e. $\widehat{W} + 2 \hat{v} = d  (x_1 - x_0) \text{ mod } 4$. 
            
            \vspace{2pt}
            
            The verifier samples a type of round uniformly from $\{3a,3b\}$ and performs the corresponding of the following
            \begin{enumerate}[itemsep=5pt]

                \vspace{5pt}
            
                \item[R3a.] (\textit{consistency check}) The verifier samples $c \gets_U \{X,Y,Z\}$, sends $c$ to the prover and expects $v \in \{0,1\}$ back. 

                \vspace{2pt}
                
                If $c = Z$ and $G = 1, v \neq \hat{b}$ the verifier \textsc{Aborts}. If $c \in \{X,Y\}, G = 0$ then

                \begin{enumerate}
                    \item[A.] If $c = \widehat{W}, v \neq \hat{v}$ the verifier \textsc{Aborts}.
                    \item[B.] If  $\widehat{W} \in \{X,Y\}$ and $c \in \left\{\frac{X-Y}{\sqrt{2}},\frac{X+Y}{\sqrt{2}}\right\}$, the verifier checks that $v$ follows the right distribution.
                \end{enumerate} 
                \item[\textbf{R3b.}] (\textit{answer collection}) The verifier expects $a \in \{0,1\}$ from the prover. If $G = 0$, the verfier returns $(|\hat{b} \rangle, a)$, and if $G = 0$, the verfier returns $\left(\ket{+_{\frac{\pi}2(\widehat{W} + 2 \hat{v})}},a \right)$.
            \end{enumerate}
        \end{enumerate}
    \end{enumerate}
    \rule{\textwidth}{0.6pt}
    \caption{Protocol between $\ver$ and $\alice$ ($\bob$).} 
    \label{fig:final_protocol}
\end{figure}

We take some time to more formally define semi-quantum games. 
The following is a slight modification of Definition~\ref{def:game}. 
The main difference is that the communication from $\ver$ to $\alice,\bob$ is quantum. 
The difference is small but we include the definition for completeness nonetheless.

\begin{definition}[Semi-quantum game]\label{def:semiquantumgame}
For $k \in \N, \rho_{AB} \in \mathcal{S}(\mathcal{H}_A \otimes \mathcal{H}_B)$ we define $\mathcal{G}(\rho_{AB},k)$ to be a game between $\alice, \bob$ and $ \ver$. $\mathcal{G}$ is played in one of two modes. $\alice,\bob$ will have access to either (i) $\rho_{AB}$ or (ii) a separable (\say{classical}) state $\sigma_{AB}$. First, the hyperparameter $k$ is distributed to all parties and one of the modes is chosen. The mode is not known to $\ver$. $\mathcal{G}$ proceeds in $k$ rounds. 

In each round, $\alice$ and $\bob$ are given their respective share of $\rho_{AB}$ in mode (i) or of $\sigma_{AB}$ in mode (ii) and are forbidden to communicate. Then
\begin{enumerate}
    \item $\ver$ sends a quantum state $\tau \in \{\zero,\one,\plus,\minus,\ii,\minusi \}$ to $\alice$ and $\ver$ sends a quantum state $\omega \in \{\zero,\one,\plus,\minus,\ii,\minusi \}$ to $\bob$,
    \item $\alice, \bob$ compute their answers $a,b \in \{0,1\}$.
    \item Answers $a,b$ are sent to $\ver$, which stores them.
\end{enumerate}
Then, the next round starts.
After the $k$-th round, $\ver$ outputs, based on $a,b$'s, either YES or NO.
\end{definition}

The goal is, of course, to design a protocol where $\ver$ can distinguish mode (i) and mode (ii). An ideal functionality for certifying entanglement of $\rho_{AB}$ would satisfy (Completeness) if the game is played in mode (i) then the interaction is accepted, (Soundness) if the game is played in mode (ii) then the interaction is rejected.

\paragraph{Score functions.} However, one usually cannot design protocols that satisfy these strong requirements. It is because there often is an inherent randomness in the game that needs to be considered. In the literature, it is often addressed by utilizing the so-called \textit{score functions}.

We first introduce some notation, for every $\alpha,\beta \in \{0,1\}$, every $\tau,\omega \in \{\zero,\one,\plus,\minus,\ii,\minusi \}$
we define $\Prob^\rho_{\alice,\bob} \left[a = \alpha, b = \beta \ | \ \tau, \omega \right]$ as the probability of $\alice$ returning $\alpha$ when given $\tau$, $\bob$ returning $\beta$ when given $\omega$ while they play in mode (i), i.e. they share $\rho_{AB}$. 
Similarly we define $\Prob^\lambda_{\alice,\bob} \left[a = \alpha, b = \beta \ | \ \tau, \omega \right]$ for when they play in mode (ii). 
Then, we define a score function $I$ as
$$
I := \sum_{s,t \in \{0,\dots,5\}} \beta_{s,t} \ \Prob \left[a = 1, b = 1 \ | \ \tau_s, \omega_t \right],
$$
where $s,t$ index states in $\{\zero,\one,\plus,\minus,\ii,\minusi \}$ and $\beta_{s,t} \in \Real$. 
What we mean formally is that we compute two quantities: $I_\rho:= \sum_{s,t \in \{0,\dots,5\}} \beta_{s,t} \ \Prob^\rho_{\alice,\bob} \left[a = 1, b = 1 \ | \ \tau_s, \omega_t \right]$ and $I_\lambda:= \sum_{s,t \in \{0,\dots,5\}} \beta_{s,t} \ \Prob^\lambda_{\alice,\bob} \left[a = 1, b = 1 \ | \ \tau_s, \omega_t \right]$.\footnote{We denote it by $\lambda$ as in shared randomness (local hidden variable).} 
That is we compute the score function for the two modes (i) and (ii). 
Then we say that the score function distinguishes the two modes, i.e. certifies entanglement of $\rho_{AB}$, if $|I_\rho - I_\lambda| \gg 0$.

As we mentioned, even though $|I_\rho - I_\lambda| \gg 0$ inherent randomness usually implies that we can't distinguish the two modes in one run of the protocol. 
The expression for $I$ includes probabilities, which means that formally we would need to repeat the game many times to obtain a reasonably good approximation to $\hat{I}$ so that there is still a separation between $\hat{I}_\rho$ and $\hat{I}_\lambda$. 
This important detail is often omitted in the literature. 
Our result deals with the notions of complexity so the number of repetitions might play a role.
We will try to be more explicit with the details without being overwhelming. 
See for instance Theorem~\ref{thm:main_formal} where the number of runs in the protocol takes into account the repetitions needed to estimate the score functions accurately. 

\subsubsection{SQG protocol for entanglement certification}\label{sec:buscemiSQG}

In this section, we give a definition and analysis of the semi-quantum game certifying entanglement of all entangled states due to \cite{Buscemi_SemiQuantumGame}. 

To describe the protocol we first define what an entanglement witness is. Entanglement witnesses were considered in \cite{horodeckiEW} as one of the possible criteria to distinguish entangled and separable states. 

\begin{definition}[\textbf{Entanglement witness}]
For an entangled state $\rho_{AB} \in \mathcal{S}(\hilb_A\otimes \hilb_B)$, an entanglement witness $W_{\rho_{AB}}$ with a parameter\footnote{As the set of separable states is convex we can require a $\eta > 0$ separation instead of the usual non-negative/positive separation.} $\eta > 0$ is a Hermitian operator acting on $\hilb_A\otimes \hilb_B$ such that,
\begin{align}
&\trace[W_{\rho_{AB}} \rho_{AB}] < -\eta, \nonumber \\
&\trace[W_{\rho_{AB}} \sigma_{AB}] > \eta \text{ for all separable states } \sigma_{AB} \in \mathcal{S}(\hilb_A \otimes \hilb_B). \label{eq:EQprop}
\end{align}
\end{definition}

Let $\rho_{AB} \in \mathcal{S}(\Com^2 \otimes \Com^2)$ be an entangled state. The set of separable states is convex, so all entangled states have entanglement witnesses. Any entanglement witness $W_{\rho_{AB}}$ of $\rho_{AB}$ can be rewritten as
\begin{equation}
    W_{\rho_{AB}} = \sum_{s,t} \beta_{s,t} \ \tau_s^\top \otimes \omega_t^\top,
\end{equation}
where $\tau_s,\omega_t$ are the projectors onto the corresponding state in $\{\ket{0},\ket{1},\ket{+},\ket{-},\ket{i},\ket{-i}\}$ and $\beta_{s,t} \in \Real$.


The following lemma defines a score function for distinguishing between $\rho_{AB}$ and all separable states. 
The proof is directly adapted from \cite{Buscemi_SemiQuantumGame}.
We include it here for completeness. 

\begin{lemma}\label{lem:entanglementwitness}
For every $\rho_{AB} \in \mathcal{S}(\Com^2 \otimes \Com^2)$ and it's entanglement witness $W_{\rho_{AB}} = \sum_{s,t} \beta_{s,t} \ \tau_s^\top  \otimes \omega_t^\top$ satisfying \eqref{eq:EQprop} the score function $I = \sum_{s,t} \beta_{s,t} \ \Prob \left[a = 1, b= 1 \ | \ \tau_s , \omega_t \right]$ has the following properties
\begin{itemize}
    \item{(\textbf{Completeness})} If $\alice,\bob$ shared $\rho_{AB}$, and their replies are defined as performing the joint projection onto the maximally entangled state and returning $1$ if the projection was successful, then $I_\rho < \eta/4$.
    \item{(\textbf{Soundness})} If $\alice,\bob$ shared a separable state only then $I_\lambda > \eta/4$.
\end{itemize}
\end{lemma}

\paragraph{Note.} Crucially $\alice,\bob$ in the soundness are not limited computationally in this lemma.

\begin{proof}
We start with completeness and then move on to soundness.
\paragraph{Completeness.} 
Recall that the strategy of honest $\alice$ is to project $\tau_s \otimes \rho_A$ onto an a maximally entangled state $\ket{\phi_+} = \frac1{\sqrt{2}}(\ket{00} + \ket{11})$. 
For $\bob$ it is to project $\rho_b \otimes \omega_t$ onto $\ket{\phi_+}$. 
This strategy yields the following score for $I$:
\begin{align}
I
&= \sum_{s,t} \beta_{s,t} \ \Prob_{\alice,\bob}^\rho[a = 1, b = 1 \ | \ \tau_s, \omega_t] \nonumber \\
&= \sum_{s,t} \beta_{s,t} \ \trace\left( \left(\ket{\phi^+}_A\bra{\phi^+}_A \otimes \ket{\phi^+}_B\bra{\phi^+}_B \right)(\tau_s \otimes \rho_{AB} \otimes \omega_t) \right) \nonumber \\
&= \sum_{s,t} \beta_{s,t} \ \trace((\tau_s^\top \otimes \omega_t^\top) \rho_{AB})/4 \nonumber \\
&= \trace(W_{\rho_{AB}} \rho_{AB})/4 \nonumber \\
&\leq - \eta/4 && \text{By \eqref{eq:EQprop}}. \nonumber 
\end{align}
This establishes completeness.

\paragraph{Soundness.} 
Assume $\alice, \bob$ share some separable state $\sigma_{AB} = \sum_k p_k \ \sigma_A^k \otimes \sigma_B^k$. 
Then we can model their actions by some general POVM $A,B$. 
Let $A_1,B_1$ be the elements corresponding to outcome $1$. 
We can then write
\begin{align}
I 
&= \sum_{s,t} \beta_{s,t} \ \Prob_{\alice,\bob}^\lambda \left[a = 1, b= 1 \ | \ \tau_s , \omega_t \right] \nonumber \\
&= \sum_{s,t} \beta_{s,t} \ \trace \left( \left(A_1 \otimes B_1 \right) \left(\tau_s \otimes \sigma_{AB} \otimes \omega_t \right) \right) \nonumber \\
&=\sum_{s,t} \beta_{s,t} \ \sum_k p_k \trace \left( \left(A_1^k \otimes B_1^k \right)(\tau_s \otimes \omega_t) \right) \nonumber \\
&= \sum_k p_k \ \trace \left( \left(A_1^k \otimes B_1^k \right) W_{\rho_{AB}}^\top \right) \nonumber \\
&= \sum_k p_k \ \trace \left(W_{\rho_{AB}} \left[(A_1^k)^\top \otimes (B_1^k)^\top\right]\right), \nonumber
\end{align}
where we denoted $A_1^k := \trace_A(A_1(\id \otimes \sigma_A^k)), B_a^k := \trace_B(B_1(\sigma_B^k \otimes \id))$ to be the efffective POVM elements acting on $\tau_s, \omega_t$ (subscript under $\trace$ denotes the partial trace). 
Because of \eqref{eq:EQprop} and the fact that $(A_1^k)^\top, (B_1^k)^\top$ are positive, Hermitian, we have that $I > \eta$. 
This concludes the proof.
\end{proof}

\subsection{Entanglement Certification - RSP + SQG}\label{sec:final}

We are ready to define the final protocol that proves Theorem~\ref{thm:mainCert}. 

In Figure~\ref{fig:final_protocol} we give a formal description of a one-round interaction between $\ver$ and $\alice$ that is a building block of the final protocol described later and defined formally in Figure~\ref{fig:really_final}. 
It is a slight modification of the RSP (Figure~\ref{fig:qubit_preparation_protocol}). 
The only difference is that in round 3b the verifier expects an additional output from the prover, i.e in round 3b $\ver$ collects an answer $a \in \{0,1\}$ from $\alice$ and \say{returns} a pair $(\tau, a)$, where $\tau \in \{\zero,\one,\plus,\minus,\ii,\minusi \}$. 
By \say{returns} we mean the following. 
The final protocol consists of running the protocol from Figure~\ref{fig:final_protocol} between $\ver$ and $\alice$ and an identical protocol between $\ver$ and $\bob$ (which collects $(\omega,b)$) in parallel, over many repetitions. 
If, in a repetition,  both $\alice$ and $\bob$ reach round 3b then a tuple $((\tau,a),(\omega,b) )$ is collected. After some number of repetitions, the following score function is computed
$$
\hat{I} = \sum_{s,t} \beta_{s,t} \ \widehat{\Prob} \left[a = 1, b= 1 \ | \ \tau_s , \omega_t \right],
$$
where $\widehat{\Prob}$ denotes the empirical probability over the collected samples. In the end, an interaction is accepted if it was not aborted in any of the rounds by neither of $\alice$ nor $\bob$ and if $I < 0$. A formal description is given in Figure~\ref{fig:really_final}.

\begin{figure}
    \vspace{3pt}
    Let $\rho_{AB} \in \mathcal{S}(\Com^2 \otimes \Com^2)$ be an entagled state, $W = \sum_{s,t} \beta_{s,t} \ \tau_s^\top  \otimes \omega_t^\top$ it's entanglement witness satisfying \eqref{eq:EQprop}, $\lambda \in \N$ a security parameter and $\delta > 0$ a confidence parameter.
    
    Repeat the following $N := O\left( \frac{\log(1/\delta) }{\eta \min_{s,t}\{ |\beta_{s,t}| \}}\right)$:
    \begin{enumerate}
        \item Run the protocol from Figure~\ref{fig:final_protocol}, with parameter $\lambda$, between $\ver_A$ and $\alice$ and between $\ver_B$ and $\bob$.
        \item If any of the runs returned $\textsc{abort}$ return $\textsc{not-entangled}$.
        \item If $\ver \leftrightarrow \alice$ returned $(\tau,a)$ and $\ver \leftrightarrow \bob$ returned $(\omega,b)$ collect the answers.
    \end{enumerate}
    If $\hat{I} = \sum_{s,t} \beta_{s,t} \ \widehat{\Prob}\left[a = 1, b= 1 \ | \ \tau_s , \omega_t \right] < 0$ return $\textsc{entangled}$, otherwise return $\textsc{not-entangled}$.
    \\
    \vspace{1pt}
    \hspace{-7pt}
   
    \caption{Final protocol.} 
    \label{fig:really_final}
\end{figure}

Now we are ready to state the final theorem of this subsection, which is a formal version of Theorem~\ref{thm:mainCert}.

\begin{theorem}\label{thm:main_formal}
There exists a universal constant $c > 0$ such that for every entangled state $\rho_{AB} \in \mathcal{S}(\Com^2 \otimes \Com^2)$, its entanglement witness $W = \sum_{s,t} \beta_{s,t} \ \tau_s^\top  \otimes \omega_t^\top$ satisfying \eqref{eq:EQprop} with $\eta$, $0 < \delta < \omega \left( \left(\frac{\eta}{\max_{s,t}|\beta_{s,t}|}\right)^{1/c} \right)$, every $n$ large enough so that $\delta = \omega(\text{negl}(n))$ the interactive protocol from Figure~\ref{fig:final_protocol} between $\ver,\alice,\bob$ is a $O\left(\frac{\log(1/\delta) \max^2_{s,t}|\beta_{s,t}|}{\eta^2 } \right)$-message protocol that exchanges \newline $O \left(\frac{\poly(n)\log(1/\delta) \max^2_{s,t}|\beta_{s,t}|}{\eta^2 } \right)$ bits of communication and satisfies the following properties:
\begin{itemize}
    \item{(\textbf{Completeness})} There exist $\alice, \bob \in \text{QPT}(n)$ such that if $\alice$ and $\bob$ shared $\rho_{AB}$ in every round of the protocol then their interaction is accepted by $\ver$ with probability $1 - \delta$.
    \item{(\textbf{Soundness})} If $\alice$ and $\bob$ share only a separable state, $\text{LWE} \not\in \mathsf{BQP/qpoly}$ and $\alice, \bob \in \text{QPT}(n)$ then the the interaction is accepted by $\ver$ with probability at most $\delta$.
\end{itemize}
\end{theorem}

The proof is deferred to Appendix~\ref{apx:nonlocality}.

\section{Entanglement Certification implies Separation of Complexity Classes}\label{sec:entanglementcertimpliescrypto}

In this section, we show that if $\ENT = \NEL$ (Conjecture~\ref{conj} is true), which means that one can certify entanglement of all entangled states, then $\BQP \neq \PP$. 
This is summarized in the following theorem

\thmmain*

\paragraph{Note.} As we discussed before our proof works also if we don't require $\ver \in \BPP$ in Definition~\ref{def:nonlocality}.
It is because the simulation we implement gives rise to a distribution that is close to the target distribution in the TV distance.
This means that it is impossible to distinguish distributions even statistically.
 

\paragraph{Proof of Theorem~\ref{thm:maininformal}.}

Now we are ready to start proving Theorem~\ref{thm:maininformal}. First, we provide some helpful definitions.

\begin{definition}\label{def:correspondingPOVM}
Let $C$ be a quantum circuit acting on $\ell$ qubits, $U$ its corresponding unitary, and $a \in \{0,1\}^\ell$. If we define
$
\ket{\phi_a}\ket{0}^{\ell-1} :=  \left(\id \otimes \ket{0}^{\ell - 1} \bra{0}^{\ell - 1} \right) U^\dagger \ket{a},
$
where $\ket{\phi_a}$ is understood as an element of $\mathcal{S}(\Com^2)$ ($\ket{\phi_a}$ might be unnormalized) then we call a pair 
$$
\left(\braket{\phi_a}{\phi_a}, \frac{\ket{\phi_a}\ket{0}^{\ell-1}\bra{0}^{\ell-1}\bra{\phi_a}}{\braket{\phi_a}{\phi_a}} \right) \in [0,1] \times \mathcal{P}(\Com^2)
$$
the POVM element of $C$ corresponding to $a$.
\end{definition}

Next lemma justifies Definition~\ref{def:correspondingPOVM}.

\begin{lemma}\label{lem:circuittoPOVM}
Let $C$ be a quantum circuit on $\ell$ qubits and $a \in \{0,1\}^\ell$. 
If we define a POVM $\left\{\mathcal{A}_a \right\}_a$ on $\Com^2$ that acts on $\ket{\lambda} \in \Omega(\Com^2)$ by running $C$ on $\ket{\lambda} \otimes \ket{0}^{\ell-1}$ and measuring all the qubits in the $Z$ basis then for every $a \in \{0,1\}^\ell$ we have that
$$
\mathcal{A}_a = \eta_a P_a,
$$
where $\left(\eta_a, P_a\right)$ is the POVM element of $C$ corresponding to $a$ (as defined according to Definition~\ref{def:correspondingPOVM}). 
\end{lemma}

\begin{proof}
We have (with a slight abuse of notation) that for every $\ket{\lambda} \in \Omega(\Com^2), a \in \{0,1\}^\ell$
\begin{align}
    \Prob \left[C \text{ outputs } a \text{ on }\ket{\lambda} \otimes \ket{0}^{\ell-1} \right] &= \bra{\lambda} \mathcal{A}_a \ket{\lambda} \nonumber \\
    &= \bra{\lambda}(\id \otimes \ket{0}^{\ell-1})^{\dagger}(U)^\dagger \ket{a}\bra{a} U (\id \otimes \ket{0}^{\ell-1})\ket{\lambda}. \nonumber
\end{align}
The POVM element $\mathcal{A}_a$ can be thus seen as,
\begin{equation}
    \mathcal{A}_a = (\id \otimes \ket{0}^{\ell-1})^{\dagger}U^\dagger \ket{a}\bra{a} U (\id \otimes \ket{0}^{\ell-1}). \label{eq:povmform}
\end{equation}
Let $\ket{\phi_a}$ be an un-normalized vector
\begin{equation}\label{eq:phix}
\ket{\phi_a}\ket{0}^{\ell-1} =  \left(\id \otimes \ket{0}^{\ell - 1} \bra{0}^{\ell - 1} \right) U^\dagger \ket{a},
\end{equation}
where we treat $\ket{\phi_a}$ as an element of $\mathcal{S}(\Com^2)$. Recall that by definition $\mathcal{A}_a = \eta_a P_a$. 
We can express $\eta_a$ and $P_a$ as
\begin{equation}\label{eq:relations}
\eta_a = \braket{\phi_a}{\phi_a}, \ \ P_a = \frac{\ket{\phi_a}\ket{0}^{\ell-1}\bra{0}^{\ell-1}\bra{\phi_a}}{\braket{\phi_a}{\phi_a}}.
\end{equation}
This is a direct consequence of the observation that: 
$$ \Big(\mathcal{A}_a \Big)^2 = \eta_a^2 P_a = \eta_a \mathcal{A}_a .$$
plugging in $\mathcal{A}_a$ from equation \eqref{eq:povmform} we arrive at \eqref{eq:relations}.    
\end{proof}

The next lemma shows that there exists a $\postqpt$ algorithm that when given access to $\poly$ copies of a pair of states can compute the square of the dot product of them up to exponential precision. The core of the proof is a combination of a binary-search-like procedure with a technique from \cite{PostBQPPP}. This technique allows, for a state $\gamma_1\ket{0} + \gamma_2\ket{1}$, where $\gamma_1,\gamma_2 > 0$, to design a $\postqpt$ algorithm that, when given access to polynomially many copies of that state, computes $\gamma_1,\gamma_2$ up to exponential precision.

\begin{lemma}{(\textbf{Dot product estimation})}\label{lem:dotproductNEW}
There exists a uniform family of $\PostBQP(\ell)$ circuits $\left\{\Cdot_\ell \right\}_\ell$ such that for all pairs of states $\ket{\psi} = \gamma_1 \ket{0} + \gamma_2 \ket{1}$, where $\gamma_1,\gamma_2 \in \Real$ and $\ket{\psi'} \in \mathcal{S}(\Com^2)$, where $\ket{\psi'}$ is not necessarily normalized the following holds. 
For every $\ell \in \N$, if $\Cdot_\ell$ is given access to $\poly(\ell)$ copies of $\ket{\psi}$, $\ket{\psi'}$ is given as a $\poly(\ell)$ bit string, and $|\gamma_1|,|\gamma_2| \in \{0\} \cup (2^{-\ell}, 2^\ell)$.
then $\Cdot_\ell$ returns $v \in [0,1]$, which with probability $1 - 2^{-\Omega(\ell)}$, satisfies
$$
\left| v - \left|\braket{\psi}{\psi'}\right|^2 \right|\leq 2^{-\Omega(\ell)}.
$$
\end{lemma}

\begin{proof}
The full algorithm will consist of a series of subroutines. 


\paragraph{Subroutine 1.} For a state $ \eta_1 \ket{0} + \eta_2 \ket{1}$ take its two copies, apply the CNOT gate and postselect on the second qubit being $\ket{0}$. The result is
$$
\frac{\eta_1^2 \ket{0} + \eta_2^2 \ket{1}}{\sqrt{\eta_1^4 + \eta_2^4}}.
$$

\paragraph{Subroutine 2.} For $\beta \in \Real$ and a state $\eta_1 \ket{0} + \eta_2 \ket{1}$ prepare $\frac{1}{\sqrt{1+\beta^2}}(\ket{0} + \beta \ket{1}) (\eta_1 \ket{0} + \eta_2 \ket{1})$, apply the CNOT gate on it and postselect on the second qubit being $\ket{0}$. The result is
$$
\frac{\eta_1 \ket{0} + \beta \eta_2 \ket{1}}{\sqrt{\eta_1^2 + \beta^2 \eta_2^2}}.
$$

\paragraph{Subroutine 3. (inspired by \cite{PostBQPPP}).} For a state $\zeta_1 \ket{0} + \zeta_2\ket{1}$, where $\zeta_1,\zeta_2 \in \{0\} \cup  \left(2^{-O(m)},2^{O(m)}\right)$ \footnote{Note that we used $m$ instead of $\ell$ here as the parameter. We do it because the subroutine will be used later with $m = \poly(\ell)$.}
, and real, positive numbers $\alpha,\beta$ prepare a state $\alpha \ket{0} (\zeta_1 \ket{0} + \zeta_2 \ket{1}) + \beta \ket{1} H (\zeta_1 \ket{0} + \zeta_2 \ket{1})$ and postselect on the second qubit being $\ket{1}$. 
The result is
\begin{equation}
\ket{\phi_{\beta/\alpha}} :=\frac{\alpha \zeta_2 \ket{0} + \frac{\beta}{\sqrt{2}}(\zeta_1 - \zeta_2)\ket{1}}{\sqrt{\alpha^2 \zeta_2^2 + \frac{\beta^2}{2}(\zeta_1 - \zeta_2)^2}}.
\end{equation}
Assume $\zeta_1 > \zeta_2 + 2^{-O(m)}$. 
We claim that there exists $i \in [-O(m),O(m)]$ such that if we set $\beta/\alpha = 2^i$ then $\ket{\phi_{2^i}}$ is close to the state $\ket{+}$, i.e.
\begin{equation}\label{eq:mincorr}
|\braket{+}{\phi_{2^i}}| \geq \frac{1+\sqrt{2}}{\sqrt{6}} > 0.985.
\end{equation}
From the assumption that $\zeta_1 > \zeta_2 + 2^{-O(m)}$ we have that there exists $i \in [-O(m),O(m)]$ such that $\ket{\phi_{2^i}}$ and $\ket{\phi_{2^{i+1}}}$ fall on opposite sides of $\ket{+}$ (in the first quadrant). 
The worst case is $\braket{+}{\phi_{2^i}} = \braket{+}{\phi_{2^{i+1}}}$, which happens for $\ket{\phi_{2^i}} = \sqrt{\frac23}\ket{0} + \sqrt{\frac13}\ket{1}, \ket{\phi_{2^{i+1}}} = \sqrt{\frac13}\ket{0} + \sqrt{\frac23}\ket{1}$. 
These give exactly the bound from \eqref{eq:mincorr}. 

Second case is when $\zeta_1 < \zeta_2$ then $\ket{\phi_{2^i}}$ never lies in the first or third quadrants and thus $|\braket{+}{\phi_{2^i}}| < \frac{1}{\sqrt{2}} < 0.985$. 
The two facts together imply that repeating the procedure $\poly(m)$ times we can distinguish  $\zeta_1 > \zeta_2 + 2^{-O(m)}$ from $\zeta_1 < \zeta_2$ with probability $1 - 2^{-m}$.

\paragraph{Subroutine 4 (Binary search).}
If for a state $\eta_1 \ket{0} + \eta_2 \ket{1}$ we run Subroutine 1. and 3. with $\beta > 0$ we obtain $\frac{\eta_1^2 \ket{0} + \beta \eta_2^2 \ket{1}}{\sqrt{\eta_1^4 + \beta^2 \eta_2^4}}$. 
If both of the amplitudes $ \in \{0\} \cup \left(2^{-O(m)},2^{O(m)}\right)$ we can run Subroutine 3. to distinguish between $\eta_1^2 > \beta \eta_2^2 + \sqrt{\eta_1^4 +\beta^2 \eta_2^4} \cdot 2^{-O(m)}$ and $\eta_1^2 < \beta \eta_2^2$. 
This means that we can run a binary-search-like algorithm with $\poly(\ell)$ repetitions of Subroutine 3. (with parameter $m = \poly(\ell)$) such that with probability $1 - 2^{-\ell}$ we find $\beta_0 > 0$ such that 
$$
|\eta_1^2 - \beta_0 \eta_2^2| \leq 2^{-2\ell}.    
$$ 

\paragraph{Algorithm.} 
Let $\gamma_1',\gamma_2' \in \Com$, be such that $\ket{\psi'} = \gamma_1' \ket{0} + \gamma_2' \ket{1}$.
By assumption of the lemma we can run Subroutine 4. on $\ket{\psi}$ and $\ket{\psi'}$ postselected on the second qubit being $\ket{0}$ and obtain $\beta_0 \in \{0\} \cup (2^{-O(\ell)}, 2^{O(\ell)})$
\begin{align}
&|\gamma_1^2 - \beta_0 \gamma_2^2| \leq 2^{-2\ell}. \label{eq:closeness}
\end{align}
If we run Subroutine 2. with parameter $\sqrt{\beta_0}$ the result is
\begin{equation}\label{eq:finalstate}
\frac{\gamma_1 \ket{0} + \sqrt{\beta_0} \gamma_2 \ket{1}}{\sqrt{\gamma_1^2 + \beta_0 \gamma_2^2}},
\end{equation}
which by \eqref{eq:closeness} is close (up to an irrelevant phase) to either $\ket{+}$ or $\ket{-}$.
Thus, measuring \eqref{eq:finalstate} in the Hadamard basis gives us, with probability $1 - 2^{-\Omega(\ell)}$, $\text{sgn}(\gamma_1 \gamma_2)$.

Note that the quantity of interest is equal to
\begin{equation}\label{eq:goaldot}
[\gamma_1 \re(\gamma_1') + \gamma_2 \re(\gamma_2')]^2 + [\gamma_1 \im(\gamma_1') + \gamma_2 \im(\gamma_2')]^2.
\end{equation}
The sign $\text{sgn}(\gamma_1 \gamma_2)$ together with \eqref{eq:closeness} are enough to compute an approximation to \eqref{eq:goaldot}. 
It is because the signs give us information about the relative sign ($+/-$) in each of the summand in \eqref{eq:goaldot}. 
Moreover, \eqref{eq:closeness} allows us to compute all quantities of interest in \eqref{eq:goaldot} (up to a sign), because of the normalization of $\ket{\psi}$ and $\ket{\psi'}$. 
For example, we can compute $\gamma_1,\gamma_2$ up to a relative sign using the fact that $\gamma_1^2 + \gamma_2^2 = 1$ and that $|\gamma_1^2 - \beta_0 \gamma_2^2| \leq 2^{-2\ell}$.

We conclude, by the union bound, that with probability $1 - 2^{-\Omega(\ell)}$ we compute $v$ satisfying the statement of the lemma.

\end{proof}

The next series of lemmas explains how to simulate some entangled states using separable states only. 
More concretely, we show how, for some entangled state $\rho_{AB}$ and all measurements applied to it by $\alice,\bob$, to reproduce the statistics of outcomes by $\alice',\bob'$ having access to a separable state only.

\begin{lemma}[\textbf{Werner's model}]\label{lem:wernermodel}
For every $q \leq \frac12$ the Werner state $\rho(q)$, defined in \eqref{eq:wernerstate}, is local for all projective measurements: $\left\{P,\id-P\right\}$ for $\alice$ and $\left\{Q, \id - Q\right\}$ for $\bob$. 
\end{lemma}

\begin{proof}
The first step is to show that the statement holds for the maximal value of $q$, i.e. $q=\frac12$. The strategy is as follows: $\alice',\bob'$ have access to a public source of randomness represented as $\ket{\lambda} \sim \omega(\Com^2)$, $\alice'$ returns $1$ if $\bra{\lambda} P \ket{\lambda} < \bra{\lambda} (\id - P) \ket{\lambda}$ and $0$ otherwise. $\bob'$ returns $1$ with probability $\bra{\lambda} Q \ket{\lambda}$. The proof that this gives rise to a distribution equal to arising from measuring the projections on $\rho(\frac12)$ can be found in \cite[section 3.1]{entnonsurvey}. 

Now, for any $\rho(q)$ with $q < \frac12$ we do the following. We write $\rho(q)$ as a mixture of $\rho(\frac12)$ and white noise $\id$, i.e. $\rho(q) = 2q \rho(\frac12) + \frac{1-2q}{4} \id$. As $q \leq \frac12$, $2q<1$, so we can consider the following strategy for $\alice' ,\bob'$: With probability $2q$ (that is coordinated with shared randomness) they perform the strategy as described for $\rho(\frac12)$ and with probability $\frac{1-2q}{4}$ they return a uniformly random bit. 

Direct derivation gives us $\Prob(a,b|P,Q) = \trace[(P_a\otimes Q_b) \rho(q)] = 2q \cdot \trace[(P_a \otimes Q_b) \rho(\frac12)] + \frac{1-2q}{4}$. This matches the distribution of the outputs produced by $\alice'$ and $\bob'$, as with probability $2q$ they reproduce the distribution associated with measuring $\rho(\frac12)$ and return uniformly random bits with probability $1-2q$. 

\end{proof}

The following family of entangled states will play an important role.

\begin{definition}
For $q \in [0,1]$ we define $\rho_0(q) \in \mathcal{S}(\Com^2 \otimes \Com^2)$ as
\begin{equation}\label{eq:rho0}
\rho_0(q) := q \ket{\psi_-}\bra{\psi_-} + \frac{1-q}{2} \ket{0}\bra{0} \otimes \mathbbm{1}.
\end{equation}
We note that $\rho_0(q)$ is entangled for all $q \in (0,1]$.
\end{definition}

\begin{lemma}\label{lem:wernermodelextended}
For every $q \leq \frac13$, the state $\rho_0(q)$, defined in \eqref{eq:rho0}, is local for all projective measurements: $\{P,\id-P\}$ for $\alice$ and $\{Q, \id - Q\}$ for $\bob$. 
\end{lemma}

\begin{proof}
The proof follows a similar strategy to the one of Lemma~\ref{lem:wernermodel}. 
We write $\rho_0(q)$ as a mixture of the Werner state $\rho(\frac12)$ and a separable state. 
\begin{equation}
    \rho_0(q) = 2q \rho(q) + (1-3q)(\ket{0}\bra{0} \otimes \id/2) + q (\ket{-}\bra{-} \otimes \id/2) 
\end{equation}

Now $\bob'$ does the following: with probability $2q$, he acts as described in Lemma~\ref{lem:wernermodel} for the case of $q = \frac12$, and with probability $1-2q$ outputs a random bit. Alice acts as follows: with probability $2q$ she acts as in Lemma~\ref{lem:wernermodel}, with probability $1-3q$ measures $P$ on $\ket{0}\bra{0}$ and with probability $q$ measures $P$ on $\ket{-}\bra{-}$.

\end{proof}

\begin{lemma}\label{lem:isentangled}
For every $ q > 0$ if we set $\sigma_{A,B} = \ket{0}\bra{0}$ in the following definition
\begin{align*}
\rho^* 
&:= \frac14 \left[\rho_0(q) + \rho_A \otimes \sigma_B + \sigma_A \otimes \rho_B + \sigma_A \otimes \sigma_B \right] \\
&= \frac14 \left[q \ket{\psi_-}\bra{\psi_-} + (2-q)\ket{0}\bra{0} \otimes \frac{\mathbbm{1}}{2} + q \frac{\mathbbm{1}}{2} \otimes \ket{0}\bra{0} + (2-q)\ket{00}\bra{00} \right]
\end{align*}
where $\rho_{A,B} = \trace_{B,A}(\rho_0(q))$, then $\rho^*$ is entangled.
\end{lemma}

\begin{proof}
The proof is given in \cite{hirshgenuinenonloc}. Interestingly the proof strategy is to give a \textbf{2-round} protocol for certifying non-locality of $\rho^*$. The statement follows as 2-round non-locality implies entanglement.
\end{proof}

\begin{lemma}[\textbf{Hirsch's model}]\label{lem:hirshmodel}
For $q \leq \frac12$ define a state
\begin{equation}\label{eq:formofstate}
\rho^*(q,\sigma_A,\sigma_B) = \frac14 \left[\rho_0(q) + \rho_A \otimes \sigma_B + \sigma_A \otimes \rho_B + \sigma_A \otimes \sigma_B \right],
\end{equation}
where $\sigma_{A,B} \in \mathcal{S}(\Com^2)$ are arbitrary and $\rho_{A,B} = \trace_{B,A}(\rho_0(q))$. Then $\rho^*(q,\sigma_A,\sigma_B)$ is local for all POVMs. Moreover, it is local via Algorithm~\ref{alg:alicesim} and \ref{alg:bobsim}.
\end{lemma}

\begin{algorithm}[tb]
   \caption{$\textsc{AliceSim-Ideal}(q,\rho_{AB}^*,\mathcal{A},\ket{\lambda})$}
   \label{alg:alicesim}
\begin{algorithmic}[1]
    \STATE {\bfseries Input:} parameter $q$, description of $\rho_{AB}^* = \frac14 \left[\rho_0(q) + \rho_A \otimes \sigma_B + \sigma_A \otimes \rho_B + \sigma_A \otimes \sigma_B \right] \in \mathcal{S}(\Com^2)$, POVM $\mathcal{A} = \{ \mathcal{A}_a\}_a = \{ \eta_a P_a\}_a$, state $\ket{\lambda} \sim \omega(\Com^2)$.
    \STATE Sample $a$ according to $\eta_a/2$
    \STATE Set $c_1 = \mathbbm{1}_{\{\bra{\lambda} P_a \ket{\lambda} < \bra{\lambda} (\id - P_a) \ket{\lambda}\}}$
    \STATE Sample $c_2\sim \text{Ber}(\bra{0} P_a \ket{0})$
    \STATE Sample $c_3\sim\text{Ber}(\bra{-} P_a\ket{-})$
    \STATE Pick $c$ from $\{c_1,c_2,c_3\}$ with corresponding probabilities $2q,1-3q,q$
    \IF{$c = 1$}
    \STATE {\bfseries Return} $a$
    \ELSE 
    \STATE Sample $a'$ according to $\trace[\mathcal{A}_{a'} \sigma_A]$
        \STATE {\bfseries Return} $a'$
    \ENDIF
\end{algorithmic}
\end{algorithm}

\begin{algorithm}[tb]
   \caption{$\textsc{BobSim-Ideal}(q,\rho_{AB}^*,\mathcal{B},\ket{\lambda})$}
   \label{alg:bobsim}
\begin{algorithmic}[1]
    \STATE {\bfseries Input:} parameter $q$, description of $\rho_{AB}^* = \frac14 \left[\rho_0(q) + \rho_A \otimes \sigma_B + \sigma_A \otimes \rho_B + \sigma_A \otimes \sigma_B \right] \in \mathcal{S}(\Com^2)$, POVM $\mathcal{B} = \{ \mathcal{B}_b\}_b = \{ \xi_b Q_b\}_b$, state $\ket{\lambda} \sim \omega(\Com^2)$.
    \STATE Sample $c$ according to $\xi_b/2$
    \STATE Sample $c_1 \sim \text{Ber}(\bra{\lambda} Q_b \ket{\lambda})$
    \STATE Sample $c_2\sim \text{Ber}(\frac12)$
    \STATE Pick $c$ from $\{c_1,c_2\}$ with corresponding probabilities $2q$ and $1-2q$
    \IF{$c = 1$}
        \STATE {\bfseries Return} $b$
    \ELSE
     \STATE Sample $b'$ according to $\trace[\mathcal{B}_{b'} \sigma_B]$
    \STATE {\bfseries Return} $b'$
    \ENDIF
\end{algorithmic}
\end{algorithm}

\begin{proof}
Direct computation gives us that
\begin{align}
&\Prob[a,b \ | \ \mathcal{A}, \mathcal{B}] = \nonumber \\
&\frac{\eta_a \xi_b}{4} \left( \trace[(P_a \otimes Q_b) \rho_0] + \trace[P_a \sigma_A] \trace[Q_b \sigma_B] + \trace[P_a \rho_A] \trace[Q_b \sigma_B] 
+ \trace[P_a \sigma_A] \trace[Q_b \rho_b]\right). \label{eq:desiredstatistics}
\end{align}
We consider cases depending on whether $\alice'$ returned in (i) line 8 or (ii) line 11 of Algorithm~\ref{alg:alicesim} and whether $\bob'$ returned in (i) line 7 or (ii) line 10 of Algorithm~\ref{alg:bobsim}. 

If both return in (i) then, by Lemma~\ref{lem:wernermodelextended}, the simulation recovers the statistics of $\rho_0$, the probability is equal to $\frac{\eta_a \xi_b}{4} \trace[(P_a \otimes Q_b) \rho_0]$, which is equal to the first term of \eqref{eq:desiredstatistics}. 
If $\alice'$ return in (i) and $\bob'$ returns in (ii) then the probability is $\frac14 \trace[\mathcal{A}_a \rho_A] \trace[\mathcal{B}_b \sigma_B]$.
If $\alice'$ return in (ii) and $\bob'$ returns in (i) then the probability is $\frac14 \trace[\mathcal{A}_a \sigma_A] \trace[\mathcal{B}_b \rho_B]$.
The last case is when both return in (ii) and the probability is then 
$\frac14 \trace[\mathcal{A}_a \sigma_A] \trace[B_b \sigma_B]$. Summing all the terms we arrive at \eqref{eq:desiredstatistics}.
\end{proof}


Next, we give some useful definitions that describe a distribution of approximately Haar random 1-qubit states that are representable by polynomially many random bits.
The corresponding distribution on random bits will be the shared public randomness used by $\alice',\bob'$ in the simulation.

\begin{definition}[Approximate states and approximate Haar distirbution]\label{def:apxHaar}
For $\ell \in \N$ and $\ket{\lambda} \in \Omega(\Com^2)$; expressed as $\ket{\lambda} = (\alpha_1 + i \alpha_2) \ket{0} + (\alpha_3 + i \alpha_4) \ket{1}$, where $\alpha_1,\alpha_2,\alpha_3,\alpha_4 \in [-1,1]$; we define $\ket{\hat{\lambda}_\ell} \in \{0,1\}^{\poly(\ell)}$ as a 1-qubit (not necessarily normalized) state represented by $(\hat{\alpha_1},\hat{\alpha_2},\hat{\alpha_3},\hat{\alpha_4}) \in \{0,1\}^{\poly(\ell)} $, 
where $\hat{\alpha_1}$ is interpreted as an element of $[-1,1]$ and is any canonical $\hat{\alpha_1}$ 
that satisifes $|\hat{\alpha_1} - \alpha_1| \leq 2^{-100 \ell}$ (the same properties hold for $\hat{\alpha_2},\hat{\alpha_3},\hat{\alpha_4}$).

Moreover we define $\hat{\omega}_\ell(\Com^2)$ to be a distribution on $\{0,1\}^{\poly(\ell)}$ defined by the process: sample $\ket{\lambda} \sim \omega(\Com^2)$, and return $\ket{\hat{\lambda}_\ell}$.
\end{definition}

The following is a direct consequence of Definition~\ref{def:apxHaar}.

\begin{lemma}\label{lem:dotprodapxcorrect}
For every $\ell \in \N, \ket{\psi},\ket{\lambda} \in \Omega(\Com^2)$ we have $\left|\left|\braket{\psi}{\lambda}\right|^2 - \left|\braket{\psi}{\lan}\right|^2\right| \leq 2^{-20 \ell}$.
\end{lemma}


\begin{lemma}\label{lem:efficientsim}
For every uniform family $\{C_\ell\}_\ell$ of polynomial size $\BQP$ circuits acting on $s(\ell)$ qubits there exists a uniform family of polynomial size $\PostBQP$ circuits (Algorithm~\ref{alg:dotproduct}), that for every $\ell \in \N, a \in \{0,1\}^{s(\ell)}$, every $\ket{\hat{\lambda}_{s(\ell)}} \in \text{supp} \left(\hat{\omega}_{s(\ell)}(\Com^2) \right)$ computes $v \in [0,1]$ that with probability $1 - 2^{-10s(\ell)}$ satisifies
$$
\left| v - \left|\braket{\psi_a}{\hat{\lambda}_{s(\ell)}}\right|^2 \right| \leq 2^{-10s(\ell)},
$$
where $\ket{\psi_a} \in \mathcal{S}(\Com^2)$ is defined as a state resulting from running $C_\ell$ backwards on $a$ and postselecting on all qubits but the first being $0$. 
\end{lemma}

\begin{proof}
Let $\ell \in \N, a \in \{0,1\}^{s(\ell)}$. According to \eqref{eq:phix} and \eqref{eq:relations} the eigenvector of $P_a$ is defined by the amplitudes in front of $\ket{0}\ket{0}^{s(\ell)-1}$ and $\ket{1}\ket{0}^{s(\ell)-1}$ after $C_\ell$ is run backwards on $\ket{a}$, which means that 
\begin{equation}\label{eq:rightstatecomputed}
\ket{\psi_a}\bra{\psi_a} = P_a.
\end{equation}

Note that $\ket{\psi_a}$ can be computed in $\PostBQP$. To do that run $C_\ell$ backwards on $\ket{a}$. Now we want to postselect on all qubits but the first being $\ket{0}$. We apply a negation gate on all these qubits, then compute an AND and write the result to a fresh ancilla. Ultimately, we postselect on this ancilla qubit being $\ket{1}$. Repeating this procedure $\poly(s(\ell))$ times we can collect $\poly(s(\ell))$ copies of $\ket{\psi_a}$.

Let $\alpha_a \ket{0} + \beta_a \ket{1} = \ket{\psi_a}$. The entries of matrices $H$ and $\text{Toffoli}$ have values from the set $\left\{1,\frac{1}{\sqrt{2}},\frac{-1}{\sqrt{2}}\right\}$, which implies that
\begin{equation}\label{eq:alphabetagoodrange}
|\alpha_a|,|\beta_a| \in \{0\} \cup \left(2^{-\poly(\ell)},2^{\poly(\ell)} \right),
\end{equation} 
as the amplitudes of $(U^\alice)^\dagger \ket{a}$ ($U^\alice$ is the unitary corresponding to $C^\alice$) can be expressed as a sum of at most $\poly(\ell)$ products of at most $\poly(\ell)$ values from $\left\{1,\frac{1}{\sqrt{2}},\frac{-1}{\sqrt{2}}\right\}$.

Lemma~\ref{lem:dotproductNEW} and \ref{lem:dotprodapxcorrect} together with \eqref{eq:alphabetagoodrange}, 
and the application of the union-bound guarantee that
with probability $1 - 2^{-10s}$
\begin{align*}
\left| v - \left|\braket{\psi_a}{\hat{\lambda}_{s(\ell)}}\right|^2 \right| \leq 2^{-10s(\ell)}.
\end{align*}
\end{proof}

\begin{lemma}\label{lem:statisticsareclose}
Let $q \leq 1/2, \sigma_A,\sigma_B \in \mathcal{S}(\Com^2)$ and $\rho^*_{AB} = \rho^*(q,\sigma_A,\sigma_B)$  (as defined in \eqref{eq:formofstate}) and $C^\alice,C^\bob$ be $\BQP$ circuits acting on $s$ qubits, and $\{\mathcal{A}_a\}_a, \{\mathcal{B}_y\}_y$ be their corresponding POVMs. For every $a,b\in \{0,1\}^{s}$ define probabilities
\begin{align*}
p_1 = \Prob_{\ket{\hat{\lambda}_{l}} \sim \omega_{s}(\Com^2)}
&\Big[\textsc{AliceSim} \left(q,\rho_{AB}^*,C^\alice,\ket{\hat{\lambda}_{s}} \right) = a \text{ and } \\
&\textsc{BobSim} \left(q,\rho_{AB}^*,C^\bob,\ket{\hat{\lambda}_{s}} \right) = b \Big],
\end{align*}
and
\begin{align*}
p_2 = \Prob_{\ket{\lambda} \sim \omega(\Com^2)}
&\Big[\textsc{AliceSim-Ideal} \left(q,\rho_{AB}^*,\mathcal{A},\ket{\lambda} \right) = a \text{ and } \\
&\textsc{BobSim-Ideal} \left(q,\rho_{AB}^*,\mathcal{B},\ket{\lambda} \right) = b \Big].
\end{align*}
Then, if for every $a \in \{0,1\}^{s}$, every $\ket{\hat{\lambda}_{s}} \in \text{supp} \left(\hat{\omega}_{s}(\Com^2) \right)$, $\textsc{DotProductWithEigenvector}$ computes $v \in [0,1]$ that with probability $1 - 2^{-10s}$ satisifies
$
\left| v - \left|\braket{\psi_a}{\hat{\lambda}_{s}}\right|^2 \right| \leq 100 \cdot 2^{-10s},
$ then
\begin{equation}\label{eq:diffinstatistics}
|p_1 - p_2| \leq 200 \cdot 2^{-8 s}.
\end{equation}

\end{lemma}

\begin{proof}
Fix $q \leq 1/2, \sigma_A,\sigma_B \in \mathcal{S}(\Com^2)$ and circuits $C^\alice$ and $C^\bob$ and set $\rho^*_{AB} = \rho^*_{AB}(q,\sigma_A,\sigma_B)$.

Let's analyze $\textsc{AliceSim} \left(q,\rho_{AB}^*,C^\alice,\ket{\hat{\lambda}_{s}} \right)$ first. Notice that $a$ is distributed according to running $C^\alice$ on a qubit with a density matrix $\rho = \frac{\id}{2}$, which by definition is distributed as $\eta_a/2$. Secondly, by definition $a'$ is distributed according to $\trace \left[\mathcal{A}_{a'} \ket{\sigma_A}\bra{\sigma_A} \right]$. 

\begin{algorithm}[tb]
   \caption{$\textsc{DotProductWithEigenvector} \left(a,C,\ket{\hat{\lambda}_{s}}\right)$}
   \label{alg:dotproduct}
\begin{algorithmic}[1]
    \STATE {\bfseries Input:} $a \in \{0,1\}^s$, circuit $C$ acting on $s$ qubits, $\ket{\hat{\lambda}_{s}} \in \text{supp}\left(\hat{\omega}_s(\Com^2) \right)$ given as a bitstring.
    \STATE Obtain $\poly(\ell)$ copies of a state $\ket{\psi_a} \in \mathcal{S}(\Com^2)$ by running $C$ backwards on $a$ and postselecting on all qubits but the first being $0$ $\hfill \triangleright$ \ We'll show $\ket{\psi_a}\bra{\psi_a} = P_a$
    \STATE Compute $v \approx \left|\braket{\psi_a}{\hat{\lambda}_{s}}\right|^2$ by running the algorithm from Lemma~\ref{lem:dotproductNEW} on $\poly(s)$ copies of $\ket{\psi_a}$ and $\ket{\hat{\lambda}_{s}}$
    \STATE {\bfseries Return} $v$
\end{algorithmic}
\end{algorithm}

\begin{algorithm}[tb]
   \caption{$\textsc{AliceSim} \left(q,\rho_{AB}^*,C,\ket{\hat{\lambda}_{s}}\right)$}
   \label{alg:aliceactualsim}
\begin{algorithmic}[1]
    \STATE {\bfseries Input:} parameter $q$, description of $\rho_{AB}^* = \frac14 \left[\rho_0(q) + \rho_A \otimes \sigma_B + \sigma_A \otimes \rho_B + \sigma_A \otimes \sigma_B \right] \in \mathcal{S}(\Com^2)$, circuit $C$ acting on $s$ qubits with a corresponding POVM $\{ \mathcal{A}_a\}_a = \{ \eta_a P_a\}_a$, $\ket{\hat{\lambda}_{s}} \sim \hat{\omega}_{s}(\Com^2)$ given as a bitstring.
    \STATE Sample a uniformly random bit $c_0 \sim U(\{0,1\})$
    \STATE Obtain $a$ by running $C$ on $\ket{c_0} \ket{0}^{s-1}$ and measuring all the qubits in the $Z$ basis. $\hfill \triangleright$ \ We'll show $a \sim \eta_a/2$
    \STATE $v_1 = \textsc{DotProductWithEigenvector}\left(a,C,\ket{\hat{\lambda}_{s}}\right)$
    \STATE $v_2 = \textsc{DotProductWithEigenvector}\left(a,C,\ket{0}\right)$
    \STATE $v_3 = \textsc{DotProductWithEigenvector}\left(a,C,\ket{-}\right)$
    \STATE Set $c_1 = \mathbbm{1}_{\left\{v_1 < \frac12 \right\}}$ 
    \STATE Sample $c_2\sim \text{Ber}(v_2)$ 
    \STATE Sample $c_3\sim\text{Ber}(v_3)$
    \STATE Pick $c$ from $\{c_1,c_2,c_3\}$ with corresponding probabilities $2q,1-3q,q$
    \IF{$c = 1$}
    \STATE {\bfseries Return} $a$
    \ELSE 
    \STATE Obtain $a'$ by running $C$ on $\ket{\sigma_A} \ket{0}^{s-1}$ and measuring all the qubits in the $Z$ basis.
        \STATE {\bfseries Return} $a'$
    \ENDIF
\end{algorithmic}
\end{algorithm}

\begin{algorithm}[tb]
   \caption{$\textsc{BobSim} \left(q,\rho_{AB}^*,C,\ket{\hat{\lambda}_{s}}\right)$}
   \label{alg:bobactualsim}
\begin{algorithmic}[1]
    \STATE {\bfseries Input:} parameter $q$, description of $\rho_{AB}^* = \frac14 \left[\rho_0(q) + \rho_A \otimes \sigma_B + \sigma_A \otimes \rho_B + \sigma_A \otimes \sigma_B \right] \in \mathcal{S}(\Com^2)$, circuit $C$ acting on $s$ qubits with a corresponding POVM $\{ \mathcal{B}_b\}_b = \{ \xi_b Q_b\}_b$, $\ket{\hat{\lambda}_{s}} \sim \hat{\omega}_{s}(\Com^2)$ given as a bitstring.
    \STATE Sample a uniformly random bit $c_0 \sim U(\{0,1\})$
    \STATE Obtain $b$ by running $C$ on $\ket{c_0} \ket{0}^{s-1}$ and measuring all the qubits in the $Z$ basis. $\hfill \triangleright$ \ We'll show $b \sim \xi_b/2$
    \STATE $v = \textsc{DotProductWithEigenvector}\left(b,C,\ket{\hat{\lambda}_{s}}\right)$
    \STATE Set $c_1 \sim \text{Ber}(v)$ 
    \STATE Sample $c_2\sim\text{Ber} \left(\frac12 \right)$
    \STATE Pick $c$ from $\{c_1,c_2\}$ with corresponding probabilities $2q,1-2q$
    \IF{$c = 1$}
    \STATE {\bfseries Return} $b$
    \ELSE 
    \STATE Obtain $b'$ by running $C$ on $\ket{\sigma_B} \ket{0}^{s-1}$ and measuring all the qubits in the $Z$ basis.
        \STATE {\bfseries Return} $b'$
    \ENDIF
\end{algorithmic}
\end{algorithm}

Moreover, for every $a \in \{0,1\}^{s}$ we have
\begin{equation}\label{eq:distanceishigh}
\Prob_{\ket{\lambda} \sim \omega(\Com^2)}\left[|\bra{\lambda} P_a \ket{\lambda} - \bra{\lambda} (\id - P_a) \ket{\lambda}| < 2^{-10s}\right] \leq 2^{-9s}.
\end{equation}
By assumption about $\textsc{DotProductWithEigenvector}$, \eqref{eq:distanceishigh} and the union bound we have that with probability $1-2 \cdot 2^{-9s}$ over $\ket{\lambda} \sim \omega(\Com^2)$ the value $b_1$ computed by $\textsc{AliceSim}$ is equal to $\textsc{AliceSim-Ideal}$. By a similar argument with probability $1 - 2\cdot 2^{-9s}$ the values $b_2,b_3$ computed by $\textsc{AliceSim}$ are $100 \cdot 2^{-10s}$ close to those computed by $\textsc{AliceSim-Ideal}$. An analogous argument holds for $\textsc{BobSim}$.




By the union bound over all the failure events of $\textsc{AliceSim}$ and $\textsc{BobSim}$ we obtain
$$
\big| p_1 - p_2 \big| \leq 200 \cdot 2^{-8 s}.  
$$

\end{proof}

\begin{lemma}\label{lem:howtouseBQP=PostBQP}
Let $\{C_\ell\}_\ell$ be a uniform family of polynomial size $\PostBQP$ circuits with the following property. There exists $c \in \Real_+$ such that for every $\ell$ there exists a function $f_\ell : \{0,1\}^\ell \rightarrow \Real_+$ such that, for every $a \in \{0,1\}^\ell$, $C_\ell$ run on $a$ returns $v \in [0,1)$ that with probability $1 - 2^{-c \cdot \ell}$ satisfies
\begin{equation}\label{eq:outputclosetof}
| v - f_\ell(a)| \leq 2^{-c \cdot \ell},
\end{equation}
where $v$ is represented as an element of $\{0,1\}^\ell$ and is obtained by measuring $\ell$ designated qubits of $C_\ell$ in the $Z$ basis. 

Then, if $\ \BQP = \PostBQP$, then there exists a uniform family of polynomial size $\BQP$ circuits $\{C'_\ell\}_\ell$ such that for every $\ell$, every $a \in \{0,1\}^\ell$, $C'_\ell$ run on $a$ returns $v \in [0,1)$ that with probability $1 - 2^{-c \cdot \ell}$ satisfies
\begin{align*}
| v - f_\ell(a)| \leq 10 \cdot 2^{-c \cdot \ell},
\end{align*}
where $v$ is obtained in the same way.
\end{lemma}

\begin{proof}
For $\ell \in \N$, $u \in \{0,1\}^\ell$ let $C_\ell(u)$ be a $\PostBQP$ circuit that: first repeats $N = \poly(\ell)$ times: in round $j$ obtain $v$ by running $C_\ell$, let $b_j = v \stackrel{?}{\leq} u$. 
Next, if $|\{j : b_j = 0\}| - |\{j : b_j = 1\}| \stackrel{(*)}{\leq} N/3$ then return $w = 0$, and otherwise return $w = (|\{j : b_j = 0\}| \stackrel{?}{\leq} |\{j : b_j = 1\}|)$. 
Let $p = \Prob[b_1 = 1]$. If $p \in \left(\frac49, \frac 59 \right)$ then, by the Chernoff bound, $(*)$ holds with probability $ 1- 2^{-10c \cdot \ell}$. 
On the other hand if $p \not\in \left(\frac49, \frac 59 \right)$ then with probability $1 - 2^{-10c \cdot \ell}$, $w = (p \stackrel{?}{\leq} \frac12)$. 
This means that $\Prob[w = 0] \leq \left[0,\frac13\right) \cup \left(\frac23,1\right]$. 
Thus we can associate with the family $\{C_\ell(u)\}_\ell$ a language $L(u) \subseteq \{0,1\}^*$ defined by: for $a \in \{0,1\}^\ell$, $a \in L(u)$ iff $\Prob[C_\ell(u) \text{ on } a \text{ returns } 1] \geq \frac23$.

By \eqref{eq:outputclosetof} $L(u)$ has the following property, for every $a \in \{0,1\}^\ell$, 
\begin{equation}\label{eq:propofL}
\text{if } u > f_\ell(a) + 2^{-c \cdot \ell} \text{ then } a \in L(u) \text{ and if } u > f_\ell(a) - 2^{-c \cdot \ell} \text{ then } a \not\in L(u).
\end{equation}
By the assumption that $\BQP = \PostBQP$ we have that there exists a uniform family of polynomial size $\BQP$ circuits that recognizes $L(u)$. 
By running the $\BQP$ circuit $\poly(\ell)$ times and returning the majority vote of outcomes we can decrease the error probability of recognizing $L(u)$ from $\frac13$ to $2^{-10c \cdot \ell}$. 
Let's call this family of circuits $\{C'_\ell(u)\}_\ell$.

Now we define the desired family of circuits. 
We perform a binary-search-like procedure that on an input $a$ finds $v$ such that $|v - f_\ell(a)| \leq 10 \cdot 2^{-c \cdot \ell}$. The procedure runs in $\poly(\ell)$ steps. 
Throughout the execution, it maintains an interval $(l,r)$ of possible values for $f_\ell(a)$. 
In each step it runs $C'_\ell \left(l + \frac13(r-l) \right)$ and $C'_\ell \left(l + \frac23(r-l)\right)$ and depending on the results decreases the interval to $\left(l + \frac13(r-l) - 2^{-c \cdot \ell}, r \right)$ or $\left(l, l + \frac23(r-l) + 2^{-c \cdot \ell} \right)$. By \eqref{eq:propofL} we have that with probability $1 - 2^{-c \cdot \ell}$ the new interval contains $f_\ell(a)$. 
Application of the union bound concludes the proof.
\end{proof}

\begin{lemma}\label{lem:efficientsiminBQP}
For every uniform family $\{C_\ell\}_\ell$ of polynomial size $\BQP$ circuits acting on $s(\ell)$ qubits there exists a uniform family of polynomial size $\BQP$ circuits, that for every $\ell \in \N, a \in \{0,1\}^{s(\ell)}$, every $\ket{\hat{\lambda}_{s(\ell)}} \in \text{supp} \left(\hat{\omega}_{s(\ell)}(\Com^2) \right)$ computes $v \in [0,1]$ that with probability $1 - 2^{-10s(\ell)}$ satisifies
$$
\left| v - \left|\braket{\psi_a}{\hat{\lambda}_{s(\ell)}}\right|^2 \right| \leq 10 \cdot  2^{-10s(\ell)},
$$
where $\ket{\psi_a} \in \mathcal{S}(\Com^2)$ is defined as a state resulting from running $C_\ell$ backwards on $a$ and postselecting on all qubits but the first being $0$. 
\end{lemma}

\begin{proof}
It is a direct consequence of Lemma~\ref{lem:efficientsim} and Lemma~\ref{lem:howtouseBQP=PostBQP}.    
\end{proof}

Now we are finally ready to prove Theorem~\ref{thm:maininformal}. We restate it here for convenience.

\thmmain*

\begin{remark}
As we mentioned before we show that local simulation can be implemented in $\PP$ when Alice and Bob have access to shared randomness only (not necessarily a separable state).    
This is a stronger requirement.
\end{remark}

\begin{proof}
We will prove the contraposition of the statement, i.e. $\BQP = \PostBQP \Rightarrow \ENT \neq \NEL$. We will give an example of $\mathcal{H}_A, \mathcal{H}_B, \rho^*_{AB} \in \mathcal{S}(\mathcal{H}_A \otimes \mathcal{H}_B)$ and show that for sufficiently small $\delta$, for all $k \in \N$ and all $\mathcal{G}(\rho^*_{AB},k,\cdot)$ there exists $\ell \in \N$ such that one of the two requirements from Definition~\ref{def:nonlocality} does not hold. 


Let $\hilb_A,\hilb_B = \Com^2$ and define
\begin{align*}
\rho^*_{AB} := \frac14 \left[\frac13 \ket{\psi_-}\bra{\psi_-} + \frac53\ket{0}\bra{0} \otimes \frac{\mathbbm{1}}{2} + \frac13 \cdot \frac{\mathbbm{1}}{2} \otimes \ket{0}\bra{0} + \frac53\ket{00}\bra{00} \right],
\end{align*}
Note that it is a special case of the state from Lemma~\ref{lem:isentangled} for $q = \frac13$. By Lemma~\ref{lem:hirshmodel} we know that $\rho^*_{AB}$ is local via Algorithm~\ref{alg:alicesim} and \ref{alg:bobsim}.


Assume towards contradiction that $\rho^*_{AB}$ is not-efficiently-local. This means that for all sufficiently small $\delta_0$ there exists $k \in \N$, a game $\mathcal{G}(\rho^*_{AB},\delta_0,k,\cdot)$ and $\alice, \bob \in \SampBQP$ that satisfy conditions of Definition~\ref{def:game}. Fix a polynomial $r$ that certifies that and a sufficiently small $\delta_0$. By our modelling we can assume that $\alice$ has a corresponding uniform family of circuits $\left\{C_\ell^\alice \right\}$ and so does $\bob$ with $\left\{C_\ell^\bob \right\}$.
 
Assume $\alice',\bob'$ have access to 
$\ket{\hat{\lambda}_{q(\ell)}} \sim \hat{\omega}_{q(\ell)}(\Com^2)$. 
We define $\alice'$ as follows. 
For every $x\in \{0,1\}^{p(\ell)}$, $\alice'$ runs $\textsc{AliceSim} \left(\frac13,\rho^*_{AB}, C^{\alice,x}_\ell,\ket{\hat{\lambda}_{q(\ell)}} \right)$, where $C^{\alice,x}_\ell$ is defined as a circuit that first prepares a state $\ket{x}$ in the $2$'nd through $(p(\ell)+1)$'st qubits and then runs $C_\ell^\alice$. 
This is done so that the action of $C^{\alice, x}_\ell$ is consistent with \eqref{eq:howAacts}. 
Instead of running Algorithm~\ref{alg:dotproduct} for computing $v_1,v_2,v_3$, $\alice'$ runs the $\BQP$ procedure guaranteed to exist by Lemma~\ref{lem:efficientsiminBQP}. 
$\bob'$ is defined in an analogous way, i.e. for every $y \in \{0,1\}^{p(\ell)}$, $\bob'$ runs $\textsc{BobSim} \left( \frac13,\rho^*_{AB}, 
C^{\bob,y}_\ell,\ket{\hat{\lambda}_{q(\ell)}} \right)$.

Lemmas~\ref{lem:statisticsareclose} and \ref{lem:efficientsiminBQP} guarantee that for every $a,b,x,y$ 
\begin{align}
&\Big|\Prob_{\ket{\hat{\lambda}_{q(\ell)}} \sim \omega_{q(\ell)}(\Com^2)} &&\Big[\alice' \text{ on } x,\ket{\hat{\lambda}_{q(\ell)}} \text{ returns } = a \text{ and } \nonumber \\
& &&\bob' \text{ on } y,\ket{\hat{\lambda}_{q(\ell)}} \text{ returns } = b \Big] \nonumber \\
&-\Prob_{\ket{\lambda} \sim \omega(\Com^2)} &&\Big[\textsc{AliceSim-Ideal} \left(q,\rho_{AB}^*,\mathcal{A}^x,\ket{\lambda} \right) = a \text{ and } \nonumber \\
& &&\textsc{BobSim-Ideal} \left(q,\rho_{AB}^*,\mathcal{B}^y,\ket{\lambda} \right) = b \Big]\Big| \leq 200 \cdot 2^{-8 q(\ell)}. \label{eq:finalstatisticsclose}
\end{align}

Now we arrive at a contradiction. $\mathcal{G}$ is played $k$ rounds, which means that $\ver$ collects $k$ samples. By \eqref{eq:finalstatisticsclose} and the properties of the TV-distance, we know that these $k$ samples are from a distribution that is $200k \cdot 2^{-8 q(\ell)}$ away in the TV-distance from a distribution corresponding to $\alice, \bob$. There are $2^{2q(\ell)}$ different pairs of answers for $\alice,\bob$ so by the properties of the TV-distance the probability that $\ver$ will distinguish the two distributions is at most
\begin{align*}
200k \cdot 2^{2q(\ell) - 8q(\ell)} \leq 200 \cdot 2^{-5q(\ell)},
\end{align*}
which implies, as per completeness in Definition~\ref{def:nonlocality} (where it is required that $\ver$ accepts $\alice,\bob$ with high probability) that the interaction is accepted with probability at least
\begin{equation}\label{eq:finalprob}
1 - \delta_0 - 2^{-\Omega(\ell)}.    
\end{equation}
This is a contradiction as $\ver$ should, as per soundness in Definition~\ref{def:nonlocality}, accept the interaction with $\alice', \bob'$ with probability at most $1 - \delta_0 - \frac{1}{p(\ell)}$ but it does it with at least \eqref{eq:finalprob}.

\end{proof}

\section{Implications for Delegation of Quantum Computation (DQC)}\label{sec:edqc}

As we discussed, Theorem~\ref{thm:maininformal} sheds light on the problem of delegation of quantum computation (DQC) itself.
More concretely, it implies the impossibility of some DQC protocols. 
We restate the theorem proven in this section

\thmedqcinformal*

\subsection{Extended Delegation of Quantum Computation}

We are ready to define an extended delegation of quantum computation (EDQC) formally. 
Our definition requires correctness with respect to an auxiliary input quantum state (similarly to correctness with respect to auxiliary input in \cite{advantagefromNL}).
Informally speaking a delegation protocol needs to preserve entanglement between the part of the state that the circuit operates on and the rest of the state.

\begin{definition}[\textbf{EDQC}]\label{def:EDQC}
Let $\mathcal{C}$ be a classical description of a quantum circuit on $k$ qubits that operates on two registers: $V$ of $k-1$ qubits and $Q$ of 1-qubit.\footnote{
We could have considered a more general definition, where the auxiliary register holds a state on many qubits. 
Our definition generalizes naturally. 
We chose this version for simplicity and because our result implies that the more general protocol is not possible either.}
We say that $\mathfrak{P}$ is a 1-round protocol for extended delegation of quantum computation (EDQC) if the following holds. For a security parameter $\ell \in \N$, $\mathfrak{P}(\mathcal{C})$ expects questions $q \in \{0,1\}^{n(\ell)}$ and answers $a \in \{0,1\}^{n(\ell)}$ for some polynomial $n$. $\ver$ is expected to accept or reject and upon acceptance return $b \in \{0,1\}$. Moreover, there exists $\ver \in \text{PPT}(\ell)$ such that
\begin{itemize}
    \item{\textbf{(Completeness)}} There exists $\prov \in \text{QPT}(\ell)$ such that for every $s \in \{0,1\}^{k-1}$ and every $\rho_{QE} \in \mathcal{S}(\Com^2 \otimes \mathcal{H}_E)$ if $\rho_Q$ is given to $\prov$ then the following hold
    \begin{itemize}
        \item $\ver$ accepts the interaction with probability $1$,
        \item over the randomness of $\ver$ and $\prov$ the distribution of $b$ is equal to the distribution of measuring the last qubit ($Q$ register) of $U_\mathcal{C} \ket{s}_V \rho_Q$,
        \item $\rho_E$ is equal to the post-measurement state after measuring the $Q$ register conditioned on obtaining $b$. More formally it is the normalization of
        $$\trace_{VQ} \left[(\id_V \otimes \ket{b}_Q\bra{b}_Q \otimes \id_E) (U_\mathcal{C} \otimes \id_E) \ket{s}_V \rho_{QE}\right].$$
    \end{itemize}
    
    \item{\textbf{(Soundness)}}
    There exists $c$ such that for every $\e >0$, sufficiently large $\ell$, for every $\prov \in \text{QPT}(\ell)$ that is accepted with probability $1 - \e$, there exists $\rho_Q \in \mathcal{S}(\Com^2)$ such that for every $s \in \{0,1\}^{k-1}$ the distribution of $b$ is in $O(\e^{-c})$ total variation distance of the distribution of measuring the last qubit of $U_\mathcal{C} \ket{s}_V \rho_Q$.
\end{itemize}
\end{definition}

Now we state the main theorem of this section, which is the formal version of Theorem~\ref{thm:edqcinformal}.

\begin{theorem}\label{thm:edqcdoesntexist}
The existence of EDQC implies $\BQP \neq \PostBQP$.
\end{theorem}

On a high level, using an EDQC, we convert the semi-quantum game protocol from Section~\ref{sec:buscemiSQG} to a 1-round protocol that certifies entanglement of all entangled states against $\BQP$ adversaries, i.e. the existence of EDQC implies $\ENT = \NEL$. 
Then, invoking Theorem~\ref{thm:maininformal}, we deduce that this implies $\BQP \neq \PostBQP$.


\begin{proof}

Let $\mathfrak{P}$ be a single round EDQC, guaranteeing soundness against all provers in $\BQP$. We show that this implies that every entangled state is $\NEL$. More formally let $\rho_{AB} \in \mathcal{S}(\Com^2 \otimes \Com^2)$ be an entangled state. For every sufficiently small $\delta$ we will show that there exists $k \in \N$ and a game $\mathcal{G}$ that satisfies Definition~\ref{def:nonlocality}. 

Let $\mathcal{C}$ be a quantum circuit acting on $4$ qubits, with registers: $A$ of $3$ qubits and $Q$ of $1$ qubit. 
To every $s \in \{0,1\}^3$ we associate $\tau_s \in \{\ket{0}, \ket{1}, \ket{+}, \ket{-}, \ket{i}, \ket{-i} \}$ such that the association is surjective. 
For $s \in \{0,1\}^3$ and $\rho_Q \in \mathcal{S}(\Com^2)$ the circuit works as follows, $U_{\mathcal{C}} \ket{s}_A \rho_Q$ first creates $\tau_s$ out of $s$ and then returns the result of the Bell measurement on $\tau_s \otimes \rho_Q$. 
Equivalently its action can be expressed as a projection onto $\ket{\phi^+} = \frac{1}{\sqrt{2}}(\ket{00} + \ket{11})$, i.e. probability of returning $1$ is equal to $\trace \left[\ket{\phi^+}\bra{\phi^+}(\tau_s \otimes \rho_Q) \right]$.

\paragraph{Game.} Let $\delta \in (0,1)$ and let $W$ be an entanglement witness of $\rho_{AB}$ with parameter $\eta$ . As discussed in Section~\ref{sec:buscemiSQG} one can express the witness as:
\begin{equation}\label{eq:EW}
W = \sum_{s,t \in \{0,1\}^3} \beta_{s,t} \ \tau_s^\top \otimes \omega_t^\top,
\end{equation}
for some real coefficients $\beta_{s,t}$.\footnote{Observe that in expression \eqref{eq:EW} there are different pairs $(s,t)$ that map to the same pair of states. It is because $|\{0,1\}^3| = 8$ but $|\{\ket{0}, \ket{1}, \ket{+}, \ket{-}, \ket{i}, \ket{-i} \}| = 6$. This is inconsequential.}

We define $k = O \left(\frac{\log(1/\delta)}{\eta \min_{s,t} |\beta_{s,t}|} \right)$. The game proceeds as follows: in each repetition $\ver$ samples $s,t \sim U(\{0,1\}^3)$ and then proceeds with running two independent copies of $\mathfrak{P}(\mathcal{C},s), \mathfrak{P}(\mathcal{C},t)$ with $\alice$ and $\bob$ respectively, collects the answers, i.e. bits $b_A,b_B$. At the end $\ver$ computes statistics 
$
\widehat{\Prob} \left[ b_A = 1, b_B = 1 \ | \ \tau_s , \omega_t \right]
$
and a score function corresponding to $W$
$$
\hat{I} = \sum_{s,t} \beta_{s,t} \ \widehat{\Prob} \left[b_A = 1, b_B= 1 \ | \ \tau_s , \omega_t \right].
$$
Finally, $\ver$ declares that $\alice,\bob$ held an entangled state if and only if $\hat{I} < 0$. This can be seen as compiling the semi-quantum game from \cite{Buscemi_SemiQuantumGame} into a Bell-like game with the help of the EDQC $\mathfrak{P}$.

\paragraph{Completeness.} Assume the parties have access to $\rho_{AB}$, i.e. mode (i). By definition of $\mathfrak{P}$ there exist $\alice,\bob \in \text{QPT}(\ell)$ satisfying the completeness property of Definition~\ref{def:EDQC}. We claim that if they run their protocol on their respective shares of $\rho_{AB}$ then they certify completeness of Definition~\ref{def:nonlocality}. 

This strategy yields the following idealized (assuming perfect statistics) value:
\begin{align}
I
&= \sum_{s,t \in \{0,1\}^3} \beta_{s,t} \ \Prob \Big[b_A = 1, b_B = 1 \ | \ s,t \Big] \nonumber \\
&= \sum_{s,t} \beta_{s,t} \ \trace \left[ \left(\ket{\phi^+}_A\bra{\phi^+}_A \otimes \ket{\phi^+}_B\bra{\phi^+}_B \right)(\tau_s \otimes \rho_{AB} \otimes \omega_t) \right] \nonumber \\
&= \sum_{s,t} \beta_{s,t} \ \trace \left[\left(\tau_s^\top \otimes \omega_t^\top \right) \rho_{AB} \right]/4 \nonumber \\
&= \trace \left[W \rho_{AB} \right]/4 && \text{By \eqref{eq:EW}} \nonumber \\
&\leq - \eta/4, && \text{By \eqref{eq:EQprop}}  \label{eq:negativescore}
\end{align}
where in the crucial second equality we used the following. To compute $ \Prob[b_A = 1, b_B = 1 \ | \ s,t]$ quantum mechanics allows as to think that first $\alice$ performs the measurement and then $\bob$ performs his measurement on the post measurement state after actions of $\alice$. By the second property of completeness we have that for every $s$, $b_A$, i.e. the bit collected by $\ver$ from $\alice$ is distributed according to $\trace[\ket{\phi^+}_A \bra{\phi^+}_A (\tau_s \otimes \trace_B[\rho_{AB}])]$. Next, the third property of completeness guarantees that the postmeasurement state of $\bob$'s share of $\rho_{AB}$ is equal to the post measurement state of performing the Bell measurement on $\tau_s \otimes \rho_A$ conditioned on obtaining outcome $b_A = 1$. This means that $\bob$'s state is 
$$
\frac{\trace_A \left[(\ket{\phi^+}_A\bra{\phi^+}_A \otimes \id) (\tau_s \otimes \rho_{AB}) \right] }{\trace \left[\ket{\phi^+}_A\bra{\phi^+}_A (\tau_s \otimes \rho_A) \right]}.
$$
Thus the overall probability is exactly
\begin{align*}
&\Prob[b_A = 1, b_B = 1 \ | \ s,t] \\
&= \trace \left[\ket{\phi^+}_A\bra{\phi^+}_A (\tau_s \otimes \rho_A) \right] \cdot \frac{\trace \left[\ket{\phi^+}_B \bra{\phi^+}_B) \left(\trace_A \left[(\ket{\phi^+}_A\bra{\phi^+}_A \otimes \id) (\tau_s \otimes \rho_{AB}) \right] \otimes \omega_t \right) \right]}{\trace \left[\ket{\phi^+}_A\bra{\phi^+}_A (\tau_s \otimes \rho_A) \right]} \\
&= \trace \left[\ket{\phi^+}_B \bra{\phi^+}_B \left(\trace_A[(\ket{\phi^+}_A\bra{\phi^+}_A \otimes \id) (\tau_s \otimes \rho_{AB})] \otimes \omega_t \right) \right] \\
&= \trace\left[ \left(\ket{\phi^+}_A\bra{\phi^+}_A \otimes \ket{\phi^+}_B\bra{\phi^+}_B \right)(\tau_s \otimes \rho_{AB} \otimes \omega_t) \right],
\end{align*}
where we used the properties of partial trace in the last equality. From \eqref{eq:negativescore}, setting of $k = O \left(\frac{\log(1/\delta)}{\eta \min_{s,t} |\beta_{s,t}|} \right)$ and a standard application of the Chernoff bound we get that with probability $1 - \delta$ we have $|\hat{I}-I| < \eta/4$, hence $\hat{I}\leq 0$ and thus the interaction is accepted.

\paragraph{Soundness.}
 Let $\textsc{Enc}$ be the $\ver$'s deterministic algorithm for generating $q$ that takes as input $s \in \{0,1\}^{3}$ and randomness $r \in \{0,1\}^{\poly(\ell)}$, i.e. $q = \textsc{Enc}(s,r)$ and similarly $\textsc{Dec}$ be the $\ver$'s deterministic algorithm for generating $b$, i.e. $b = \textsc{Dec}(s,r,a)$.
Let $\alice, \bob \in \qpt(\ell)$ be run in mode (ii), i.e. access to a separable state $\sigma_{AB} = \sum_{k=1}^\infty p_k \ \sigma_A^{(k)} \otimes \sigma_B^{(k)} \in \mathcal{S}(\mathcal{H}_A \otimes \mathcal{H}_B)$.\footnote{We allow $\sigma_{AB}$ to be a convex combination of infinitely many product states.}

Assume towards contradiction that there exists a negligible function $\text{negl}$ such that $\ver$ accepts with probability at least $1 - \delta - \text{negl}(\ell)$. Let $\ell$ be large enough so that $\text{negl}(\ell) \leq \delta$. Then $\ver$ accepts with probability $ 1- 2\delta$.


For simplicity of notation denote the size of the questions and answers in the protocol as $n = n(\ell)$.
For every $q \in \{0,1\}^{n}$ let $\{\mathcal{A}_a(q)\}_{a \in \{0,1\}^n}$ be the effective POVM acting on $\mathcal{H}_A$ that defines $\alice$'s actions. 
Similarly for every $q' \in \{0,1\}^{n}$ we define $\{\mathcal{B}_{a'}(q')\}_{a' \in \{0,1\}^n}$ as the effective POVM acting on $\mathcal{H}_B$ for $\bob$. 
Also let $\{\mathcal{C}_b\}_{b \in \{0,1\}}$ be the effective POVM of circuit $\mathcal{C}$ acting on $k$ qubits. 
Denote by $R$ the length of the randomness $r$ used in $\text{ENC}$. 
We express the probability of $b_A = 1, b_B = 1$. For every $s,t \in \{0,1\}^3$
\begin{align}
&\Prob \Big[b_A = 1, b_B = 1 \ | \ s,t \Big] \nonumber \\
&= 2^{-2R} \sum_{r,r'} \sum\limits_{\substack{a: \text{DEC}(s,r,a) = 1 \\ a': \text{DEC}(t,r',a') = 1}} \sum_k p_k 
\trace \left[ \left(\mathcal{A}_{a}(\text{ENC}(s,r)) \otimes \mathcal{B}_{a'}(\text{ENC}(t,r') \right) \left(\sigma_A^{(k)} \otimes \sigma_B^{(k)} \right) \right] \nonumber \\
&= \sum_k p_k 2^{-2R} \sum_{r} \sum\limits_{\substack{a: \text{DEC}(s,r,a) = 1}} \sum_{r'} \sum\limits_{\substack{a': \text{DEC}(t,r',a') = 1}} \nonumber \\
&\trace \left[ \mathcal{A}_{a}(\text{ENC}(s,r)) \sigma_A^{(k)} \right] \trace \left[\mathcal{B}_{a'}(\text{ENC}(t,r')) \sigma_B^{(k)} \right] \nonumber \\
&= \sum_k p_k \left(2^{-R} \sum_{r} \sum\limits_{\substack{a: \text{DEC}(s,r,a) = 1}} 
\trace \left[ \mathcal{A}_{a}(\text{ENC}(s,r)) \sigma_A^{(k)} \right]  \right) \cdot \nonumber \\
&\cdot \left( 2^{-R} \sum_{r'} \sum\limits_{\substack{a': \text{DEC}(t,r',a') = 1}} \trace \left[\mathcal{B}_{a'}(\text{ENC}(t,r')) \sigma_B^{(k)} \right] \right). \label{eq:longeq}
\end{align}
Now, if $\ver$ accepts the whole interaction with probability $1 - 2\delta$ then in particular in a single round $\ver$ accepts the delegation part (with $\alice$ and $\bob$) of the protocol with probability at least $1 - 2\delta$. Thus by the Markov inequality there exists a subset $S \subseteq \N$ such that $\sum_{k \in S} p_k \geq 1-2\delta$ and for every $k \in S$,  $\alice$'s circuit, when given $\sigma_A^{(k)}$ succeeds in $\mathfrak{P}$ with probability $1 - 4\delta$ for every $s$. Thus if $\delta$ is sufficiently small then for every $k \in S$ soundness of $\mathfrak{P}$ holds, which implies that there exists a 1-qubit density matrix 
$\rho_Q^k \in \mathcal{S}(\Com^2)$ such that 
\begin{equation}\label{eq:closetocircuit}
2^{-R} \sum_{r} \sum\limits_{\substack{a: \text{DEC}(s,r,a) = 1}} 
\trace \left[ \mathcal{A}_{a}(\text{ENC}(s,r)) \sigma_A^{(k)} \right]  = \trace \left[\mathcal{C}_1 \left(\ket{s} \otimes \rho_Q^k \right) \right] \pm O(\delta^{-c}).
\end{equation}
and crucially \textbf{the same} $\rho_Q^k$ can be taken for all $s$. Similar argument holds for $\bob$. Now we can use \eqref{eq:closetocircuit} in \eqref{eq:longeq}. We split $\N$ into two groups, $S$ and $\N \setminus S$. For $k \in S$ we use \eqref{eq:closetocircuit} and for $k \in \N \setminus S$ we bound $\left| 2^{-R} \sum_{r} \sum\limits_{\substack{a: \text{DEC}(s,r,a) = 1}} 
\trace \left[ \mathcal{A}_{a}(\text{ENC}(s,r)) \sigma_A^{(k)} \right] - \trace \left[\mathcal{C}_1 \left(\ket{s} \otimes \rho_Q^k \right) \right] \right|$ by $1$. The same operation is performed for $\bob$. The result is

\begin{align}
&\Prob \Big[ b_A = 1, b_B = 1 \ | \ s,t \Big] \nonumber \\
&=\sum_k p_k \ \trace \left[\mathcal{C}_1 \left(\ket{s} \otimes \rho_Q^k \right) \right] \trace \left[\mathcal{C}_1 \left(\ket{t} \otimes \rho^{'k}_Q \right) \right] \pm (8\delta + O(\delta^{-2c})) \nonumber \\
&\stackrel{(1)}{=} \sum_k p_k \ \trace \left[\ket{\phi^+}\bra{\phi^+} \left(\tau_s \otimes \rho_Q^k \right) \right] \trace \left[\ket{\phi^+}\bra{\phi^+} \left(\omega_t \otimes \rho_Q^{'k} \right) \right] \pm O(\delta^{-2c}) \nonumber \\
&\stackrel{(2)}{=} \sum_k p_k \ \trace \left[\mathcal{A}_1^k \tau_s \right] \trace \left[\mathcal{B}_1^k \omega_t \right] \pm O(\delta^{-2c}) \nonumber \\
&= \sum_k p_k \ \trace \left[ \left(\mathcal{A}_1^k \otimes \mathcal{B}_1^k \right) \left(\tau_s \otimes \omega_t \right) \right] \pm O(\delta^{-2c}), \label{eq:finalprobbound}
\end{align}
where in (1) we used the definition of $\mathcal{C}$ and in (2) we denoted by $\mathcal{A}^k$ and $\mathcal{B}^k$ the effective POVMs acting on $\tau_s$ and $\omega_t$ respectively. Again, crucially, the same $\mathcal{A}^k$ ($\mathcal{B}^k$) can be taken for all $s$ ($t$) because the same was true for $\rho_Q^k$  ($\rho_Q^{'k}$). With \eqref{eq:finalprobbound} we can express the score function 
\begin{align}
I 
&= \sum_{s,t \in \{0,1\}^3} \beta_{s,t} \ \Prob \Big[b_A = 1, b_B = 1 \ | \ s,t \Big] \nonumber \\
&\geq \sum_{s,t} \beta_{s,t} \ \sum_k p_k \ \trace \left[ \left(\mathcal{A}_1^k \otimes \mathcal{B}_1^k \right)(\tau_s \otimes \omega_t) \right] - \max_{s,t}|\beta_{s,t}| \cdot O(\delta^{-2c}) \nonumber \\
&\stackrel{(1)}{\geq} \sum_k p_k \ \trace \left[ \left(\mathcal{A}_1^k \otimes \mathcal{B}_1^k \right) W^\top \right] - \eta/2 \nonumber \\
&\geq \sum_k p_k \ \trace \left[W \left[ \left(\mathcal{A}_1^k \right)^\top \otimes \left(\mathcal{B}_1^k \right)^\top\right]\right] - \eta/2 \label{eq:scorelwrbnd},
\end{align}
where in (1) we used that $\delta$ is sufficiently small, as per Definition~\ref{def:nonlocality} the upper bound on $\delta$ can be chosen after $W$ was picked.
Because of \eqref{eq:scorelwrbnd}, \eqref{eq:EQprop} and the fact that $(\mathcal{A}_1^k)^\top, (\mathcal{B}_1^k)^\top$ are positive, Hermitian, we have that $I \geq \eta - \eta/2 \geq \eta/2$. We conclude the proof by a standard application of the Chernoff bound to argue that with probability $1 - \delta$ we have $\hat{I} > 0$. This gives a contradiction with the assumption that $\ver$ accepted the interaction with probability $ 1- 2\delta$.

This proves that in fact if $\mathfrak{P}$ is complete and sound against all $\qpt$ cheating provers, every entangled state is $\NEL$. Now a direct application of Theorem~\ref{thm:maininformal} implies that if such $\mathfrak{P}$ exists $\BQP \neq \PostBQP$.

\end{proof}

\subsection{Requirements of EDQC}\label{sec:whyEDQC}

In this section, we briefly discuss how the EDQC introduced in Definition~\ref{def:EDQC} relates to other notions of DQC from the literature. 
We argue that all non-standard requirements of EDQC were considered in the literature before. 
There are 3 crucial differences between EDQC and more standard versions of DQC. 
EDQC (i) is 1-round, (ii) requires completeness with respect to the auxiliary input, (iii) the soundness implies that the prover is not allowed to select the auxiliary input adaptively.

Let's elaborate on each point. 
The first difference is that (i) we require the protocol to be 1-round. 
Although the original delegation protocol from~\cite{mahadev} requires several rounds of interaction, Alagic et al.~\cite{alagic} described how this protocol can be transformed into a 1-round protocol in the QROM. 
This is possible via a technique similar to the Fiat-Shamir transform and parallel repetitions of the measurement protocol from \cite{mahadev}. 

As we discussed in Section~\ref{sec:openproblems} our proof technique could in principle be generalized beyond 1-round setting.
We focused on a 1-round local simulation but in \cite{hirshsigmalocal} the authors describe a local simulation for 3-message protocols. 
If one proves that this local model can be implemented in $\PP$ then Theorem~\ref{thm:edqcdoesntexist} automatically extends to $3$-message EDQCs. 
Similarly if one introduces a local model for more rounds the theorem would naturally extend to EDQCs with that number of rounds. 

EDQC requires (ii) completeness with respect to an auxiliary input.
The prover can input a state of his choice as part of the input to the circuit.
As we mentioned the same requirement was considered in \cite{advantagefromNL}, where it was shown that existing qFHE schemes (\cite{mahadevFHE} and \cite{brakerskiQFHE}) satisfy it.
However, the soundness of qFHE is different from the soundness of EDQC.
Recently (\cite{nonlocalitytoDQC}) it was shown that qFHE can be used as a black box for DQC.
This is some evidence that qFHE implies EDQC.
If this was indeed the case our result would imply that there's no $\PP$-sound qFHE.
However, as we discuss below, the relationship between the blindness of DQC and that of EDQC is not clear.

We give an idea of why it should be possible to realize EDQC soundness based on DQC protocols that rely on Kitaev's local Hamiltonian reduction. 
To realize this functionality one can remove terms from the Hamiltonian that correspond to the portion of input that the verifier does not control (removing a part of $H_{\text{in}}$). 
An in-depth description of the guarantees, when the penalty terms corresponding to the input are partially removed, can also be found in \cite[Lemma 4]{abs-2112-09625}. 

The final difference from the standard setup is that (iii) $\rho$ from the soundness property of Definition~\ref{def:EDQC} is assumed to have no dependence on the input $s$. 
Intuitively this means that the adversary is not allowed to choose the circuit input states adaptively, i.e., dependent on $s$. 
This property is reminiscent of blindness - the second after verifiability property of interest for DQC. 
Our property does not directly imply blindness, at least in the case of multi-round protocols, as, for instance, if the prover were to commit to a state and receive the input in the clear our property would still hold but the protocol would not be blind.
In particular, this means that it is not clear that EDQC is a stronger assumption than blindness.

It is also not clear if the blindness implies the soundness of EDQC.
The main issue is that the state $\rho$ (Definition~\ref{def:EDQC}) is an abstract state, that the prover does not necessarily need to have access to. 
Hence, we can not argue that $\rho$ is only dependent on the randomness of the verifier and the question $q$. 
If the definition is stronger and implies that the prover \emph{holds} the state, similarly to guarantees in \cite{VidickZ21}, it might be possible to show that blindness implies our property. 
However, as we mentioned, \cite{advantagefromNL} gives some evidence that qFHE might imply EDQC.

\bibliography{references}

\begin{thebibliography}{59}
\providecommand{\natexlab}[1]{#1}
\providecommand{\url}[1]{\texttt{#1}}
\expandafter\ifx\csname urlstyle\endcsname\relax
  \providecommand{\doi}[1]{doi: #1}\else
  \providecommand{\doi}{doi: \begingroup \urlstyle{rm}\Url}\fi

\bibitem[Aaronson()]{scott25challenge}
S.~Aaronson.
\newblock {25\$ challenge}.
\newblock \url{https://scottaaronson.blog/?p=284}.
\newblock Accessed: 2023-07-15.

\bibitem[Aaronson(2005)]{PostBQPPP}
S.~Aaronson.
\newblock Quantum computing, postselection, and probabilistic polynomial-time.
\newblock \emph{Proceedings of the Royal Society A: Mathematical, Physical and Engineering Sciences}, 461, 01 2005.
\newblock \doi{10.1098/rspa.2005.1546}.

\bibitem[Aaronson(2016)]{Aaronson2016TheCO}
S.~Aaronson.
\newblock The complexity of quantum states and transformations: From quantum money to black holes.
\newblock \emph{Electron. Colloquium Comput. Complex.}, TR16, 2016.
\newblock URL \url{https://api.semanticscholar.org/CorpusID:1869239}.

\bibitem[Aaronson et~al.(2019)Aaronson, Cojocaru, Gheorghiu, and Kashefi]{aaronsonimposs}
S.~Aaronson, A.~Cojocaru, A.~Gheorghiu, and E.~Kashefi.
\newblock {Complexity-Theoretic Limitations on Blind Delegated Quantum Computation}.
\newblock In C.~Baier, I.~Chatzigiannakis, P.~Flocchini, and S.~Leonardi, editors, \emph{46th International Colloquium on Automata, Languages, and Programming (ICALP 2019)}, volume 132 of \emph{Leibniz International Proceedings in Informatics (LIPIcs)}, pages 6:1--6:13, Dagstuhl, Germany, 2019. Schloss Dagstuhl--Leibniz-Zentrum fuer Informatik.
\newblock ISBN 978-3-95977-109-2.
\newblock \doi{10.4230/LIPIcs.ICALP.2019.6}.
\newblock URL \url{http://drops.dagstuhl.de/opus/volltexte/2019/10582}.

\bibitem[Aaronson et~al.(2022)Aaronson, Bouland, Fefferman, Ghosh, Vazirani, Zhang, and Zhou]{pseudoentanglement}
S.~Aaronson, A.~Bouland, B.~Fefferman, S.~Ghosh, U.~Vazirani, C.~Zhang, and Z.~Zhou.
\newblock Quantum pseudoentanglement.
\newblock \emph{arXiv preprint arXiv:2211.00747}, 2022.

\bibitem[Aharonov(2003)]{ToffoliH}
D.~Aharonov.
\newblock A simple proof that toffoli and hadamard are quantum universal.
\newblock \emph{arXiv: Quantum Physics}, 2003.

\bibitem[Aharonov and Vazirani(2012)]{vaziranifalsifiablequantum}
D.~Aharonov and U.~Vazirani.
\newblock Is quantum mechanics falsifiable? a computational perspective on the foundations of quantum mechanics.
\newblock 06 2012.

\bibitem[Alagic et~al.(2020)Alagic, Childs, Grilo, and Hung]{alagic}
G.~Alagic, A.~M. Childs, A.~B. Grilo, and S.-H. Hung.
\newblock Non-interactive classical verification of quantum computation.
\newblock \emph{IACR Cryptol. ePrint Arch.}, 2020:\penalty0 1422, 2020.

\bibitem[Ambainis et~al.(2008)Ambainis, Leung, Mancinska, and Ozols]{qracs}
A.~Ambainis, D.~Leung, L.~Mancinska, and M.~Ozols.
\newblock Quantum random access codes with shared randomness.
\newblock 10 2008.

\bibitem[Ananth et~al.(2022)Ananth, Qian, and Yuen]{cryptofromPRS}
P.~Ananth, L.~Qian, and H.~Yuen.
\newblock Cryptography from pseudorandom quantum states.
\newblock In \emph{Advances in Cryptology – CRYPTO 2022: 42nd Annual International Cryptology Conference, CRYPTO 2022, Santa Barbara, CA, USA, August 15–18, 2022, Proceedings, Part I}, page 208–236, Berlin, Heidelberg, 2022. Springer-Verlag.
\newblock ISBN 978-3-031-15801-8.
\newblock \doi{10.1007/978-3-031-15802-5_8}.
\newblock URL \url{https://doi.org/10.1007/978-3-031-15802-5_8}.

\bibitem[Arnon-Friedman et~al.(2023)Arnon-Friedman, Brakerski, and Vidick]{arnon-friedmancompent}
R.~Arnon-Friedman, Z.~Brakerski, and T.~Vidick.
\newblock {Computational Entanglement Theory}.
\newblock 10 2023.

\bibitem[Augusiak et~al.(2014)Augusiak, Demianowicz, and Acín]{entnonsurvey}
R.~Augusiak, M.~Demianowicz, and A.~Acín.
\newblock Local hidden--variable models for entangled quantum states.
\newblock \emph{Journal of Physics A: Mathematical and Theoretical}, 47, 05 2014.
\newblock \doi{10.1088/1751-8113/47/42/424002}.

\bibitem[Babai and Moran(1988)]{babaiAMgames}
L.~Babai and S.~Moran.
\newblock Arthur-merlin games: A randomized proof system, and a hierarchy of complexity classes.
\newblock \emph{Journal of Computer and System Sciences}, 36\penalty0 (2):\penalty0 254--276, 1988.
\newblock ISSN 0022-0000.
\newblock \doi{https://doi.org/10.1016/0022-0000(88)90028-1}.
\newblock URL \url{https://www.sciencedirect.com/science/article/pii/0022000088900281}.

\bibitem[Barrett(2002)]{barrettmodel}
J.~Barrett.
\newblock Nonsequential positive-operator-valued measurements on entangled mixed states do not always violate a bell inequality.
\newblock \emph{Phys. Rev. A}, 65:\penalty0 042302, Mar 2002.
\newblock \doi{10.1103/PhysRevA.65.042302}.
\newblock URL \url{https://link.aps.org/doi/10.1103/PhysRevA.65.042302}.

\bibitem[Bell(1964)]{bell}
J.~S. Bell.
\newblock On the einstein podolsky rosen paradox.
\newblock \emph{Physics Physique Fizika}, 1:\penalty0 195--200, Nov 1964.
\newblock \doi{10.1103/PhysicsPhysiqueFizika.1.195}.
\newblock URL \url{https://link.aps.org/doi/10.1103/PhysicsPhysiqueFizika.1.195}.

\bibitem[Bernstein and Vazirani(1993)]{vaziraniBQP}
E.~Bernstein and U.~Vazirani.
\newblock Quantum complexity theory.
\newblock In \emph{Proceedings of the Twenty-Fifth Annual ACM Symposium on Theory of Computing}, STOC '93, page 11–20, New York, NY, USA, 1993. Association for Computing Machinery.
\newblock ISBN 0897915917.
\newblock \doi{10.1145/167088.167097}.
\newblock URL \url{https://doi.org/10.1145/167088.167097}.

\bibitem[Bouland et~al.(2020)Bouland, Fefferman, and Vazirani]{bouland_wormholes}
A.~Bouland, B.~Fefferman, and U.~Vazirani.
\newblock {Computational Pseudorandomness, the Wormhole Growth Paradox, and Constraints on the AdS/CFT Duality}.
\newblock In T.~Vidick, editor, \emph{11th Innovations in Theoretical Computer Science Conference (ITCS 2020)}, volume 151 of \emph{Leibniz International Proceedings in Informatics (LIPIcs)}, pages 63:1--63:2, Dagstuhl, Germany, 2020. Schloss Dagstuhl -- Leibniz-Zentrum f{\"u}r Informatik.
\newblock ISBN 978-3-95977-134-4.
\newblock \doi{10.4230/LIPIcs.ITCS.2020.63}.
\newblock URL \url{https://drops.dagstuhl.de/entities/document/10.4230/LIPIcs.ITCS.2020.63}.

\bibitem[Bowles et~al.(2018)Bowles, \ifmmode \check{S}\else \v{S}\fi{}upi\ifmmode~\acute{c}\else \'{c}\fi{}, Cavalcanti, and Ac\'{\i}n]{fourparties}
J.~Bowles, I.~\ifmmode \check{S}\else \v{S}\fi{}upi\ifmmode~\acute{c}\else \'{c}\fi{}, D.~Cavalcanti, and A.~Ac\'{\i}n.
\newblock Device-independent entanglement certification of all entangled states.
\newblock \emph{Phys. Rev. Lett.}, 121:\penalty0 180503, Oct 2018.
\newblock \doi{10.1103/PhysRevLett.121.180503}.
\newblock URL \url{https://link.aps.org/doi/10.1103/PhysRevLett.121.180503}.

\bibitem[Bowles et~al.(2020)Bowles, Hirsch, and Cavalcanti]{broadcasting}
J.~Bowles, F.~Hirsch, and D.~Cavalcanti.
\newblock Single-copy activation of bell nonlocality via broadcasting of quantum states.
\newblock \emph{Quantum}, 5:\penalty0 499, 2020.
\newblock URL \url{https://api.semanticscholar.org/CorpusID:220919966}.

\bibitem[Brakerski(2018)]{brakerskiQFHE}
Z.~Brakerski.
\newblock Quantum fhe (almost) as secure as classical.
\newblock In \emph{Advances in Cryptology – CRYPTO 2018: 38th Annual International Cryptology Conference, Santa Barbara, CA, USA, August 19–23, 2018, Proceedings, Part III}, page 67–95, Berlin, Heidelberg, 2018. Springer-Verlag.
\newblock ISBN 978-3-319-96877-3.
\newblock \doi{10.1007/978-3-319-96878-0_3}.
\newblock URL \url{https://doi.org/10.1007/978-3-319-96878-0_3}.

\bibitem[Brakerski(2023)]{zvikaequiv}
Z.~Brakerski.
\newblock Black-hole radiation decoding is quantum cryptography, 2023.

\bibitem[Brakerski et~al.(2018)Brakerski, Christiano, Mahadev, Vazirani, and Vidick]{BCM}
Z.~Brakerski, P.~F. Christiano, U.~Mahadev, U.~V. Vazirani, and T.~Vidick.
\newblock Certifiable randomness from a single quantum device.
\newblock \emph{CoRR}, abs/1804.00640, 2018.
\newblock URL \url{http://arxiv.org/abs/1804.00640}.

\bibitem[Brakerski et~al.(2021)Brakerski, Christiano, Mahadev, Vazirani, and Vidick]{brakerskicertifiedrandomness}
Z.~Brakerski, P.~Christiano, U.~Mahadev, U.~Vazirani, and T.~Vidick.
\newblock A cryptographic test of quantumness and certifiable randomness from a single quantum device.
\newblock \emph{J. ACM}, 68\penalty0 (5), aug 2021.
\newblock ISSN 0004-5411.
\newblock \doi{10.1145/3441309}.
\newblock URL \url{https://doi.org/10.1145/3441309}.

\bibitem[Brakerski et~al.(2023)Brakerski, Canetti, and Qian]{brakerskiquantumcrypto}
Z.~Brakerski, R.~Canetti, and L.~Qian.
\newblock On the computational hardness needed for quantum cryptography.
\newblock In Y.~Kalai, editor, \emph{14th Innovations in Theoretical Computer Science Conference, ITCS 2023}, Leibniz International Proceedings in Informatics, LIPIcs. Schloss Dagstuhl- Leibniz-Zentrum fur Informatik GmbH, Dagstuhl Publishing, Jan. 2023.
\newblock \doi{10.4230/LIPIcs.ITCS.2023.24}.
\newblock Publisher Copyright: {\textcopyright} Zvika Brakerski, Ran Canetti, and Luowen Qian; licensed under Creative Commons License CC-BY 4.0.; 14th Innovations in Theoretical Computer Science Conference, ITCS 2023 ; Conference date: 10-01-2023 Through 13-01-2023.

\bibitem[Branciard et~al.(2013)Branciard, Rosset, Liang, and Gisin]{MDIEW}
C.~Branciard, D.~Rosset, Y.-C. Liang, and N.~Gisin.
\newblock Measurement-device-independent entanglement witnesses for all entangled quantum states.
\newblock \emph{Phys. Rev. Lett.}, 110:\penalty0 060405, Feb 2013.
\newblock \doi{10.1103/PhysRevLett.110.060405}.
\newblock URL \url{https://link.aps.org/doi/10.1103/PhysRevLett.110.060405}.

\bibitem[Buscemi(2012)]{Buscemi_SemiQuantumGame}
F.~Buscemi.
\newblock All entangled quantum states are nonlocal.
\newblock \emph{Phys. Rev. Lett.}, 108:\penalty0 200401, May 2012.
\newblock \doi{10.1103/PhysRevLett.108.200401}.
\newblock URL \url{https://link.aps.org/doi/10.1103/PhysRevLett.108.200401}.

\bibitem[Clauser et~al.(1969)Clauser, Horne, Shimony, and Holt]{CHSH69}
J.~F. Clauser, M.~A. Horne, A.~Shimony, and R.~A. Holt.
\newblock Proposed experiment to test local hidden-variable theories.
\newblock \emph{Phys. Rev. Lett.}, 23:\penalty0 880--884, Oct 1969.
\newblock \doi{10.1103/PhysRevLett.23.880}.
\newblock URL \url{https://link.aps.org/doi/10.1103/PhysRevLett.23.880}.

\bibitem[Einstein et~al.(1935)Einstein, Podolsky, and Rosen]{einsteinepr}
A.~Einstein, B.~Podolsky, and N.~Rosen.
\newblock Can quantum-mechanical description of physical reality be considered complete?
\newblock \emph{Phys. Rev.}, 47:\penalty0 777--780, May 1935.
\newblock \doi{10.1103/PhysRev.47.777}.
\newblock URL \url{https://link.aps.org/doi/10.1103/PhysRev.47.777}.

\bibitem[Fitzsimons and Kashefi(2012)]{receivemeasure}
J.~Fitzsimons and E.~Kashefi.
\newblock Unconditionally verifiable blind quantum computation.
\newblock \emph{Physical Review A}, 96:\penalty0 012303, 2012.

\bibitem[Gheorghiu and Hoban(2020)]{gheorghiu2020estimating}
A.~Gheorghiu and M.~J. Hoban.
\newblock Estimating the entropy of shallow circuit outputs is hard.
\newblock \emph{arXiv preprint arXiv:2002.12814}, 2020.

\bibitem[Gheorghiu and Vidick(2019)]{vidick}
A.~Gheorghiu and T.~Vidick.
\newblock Computationally-secure and composable remote state preparation.
\newblock 11 2019.
\newblock \doi{10.1109/FOCS.2019.00066}.

\bibitem[Gheorghiu et~al.(2015)Gheorghiu, Kashefi, and Wallden]{MIPforBQP2}
A.~Gheorghiu, E.~Kashefi, and P.~Wallden.
\newblock Robustness and device independence of verifiable blind quantum computing.
\newblock \emph{New Journal of Physics}, 17\penalty0 (8):\penalty0 083040, aug 2015.
\newblock \doi{10.1088/1367-2630/17/8/083040}.
\newblock URL \url{https://dx.doi.org/10.1088/1367-2630/17/8/083040}.

\bibitem[Gheorghiu et~al.(2017)Gheorghiu, Kapourniotis, and Kashefi]{andrusurvey}
A.~Gheorghiu, T.~Kapourniotis, and E.~Kashefi.
\newblock Verification of quantum computation: An overview of existing approaches.
\newblock \emph{Theory of Computing Systems}, 63:\penalty0 715--808, 2017.

\bibitem[Gisin(1991)]{gisinpurestatesarenonlocal}
N.~Gisin.
\newblock Bell's inequality holds for all non-product states.
\newblock \emph{Physics Letters A}, 154\penalty0 (5):\penalty0 201--202, 1991.
\newblock ISSN 0375-9601.
\newblock \doi{https://doi.org/10.1016/0375-9601(91)90805-I}.
\newblock URL \url{https://www.sciencedirect.com/science/article/pii/037596019190805I}.

\bibitem[Gluch et~al.(2023)Gluch, Barooti, and Urbanke]{abs-2112-09625}
G.~Gluch, K.~Barooti, and R.~Urbanke.
\newblock Breaking a classical barrier for classifying arbitrary test examples in the quantum model.
\newblock In F.~Ruiz, J.~Dy, and J.-W. van~de Meent, editors, \emph{Proceedings of The 26th International Conference on Artificial Intelligence and Statistics}, volume 206 of \emph{Proceedings of Machine Learning Research}, pages 11457--11488. PMLR, 25--27 Apr 2023.
\newblock URL \url{https://proceedings.mlr.press/v206/gluch23a.html}.

\bibitem[Hirsch et~al.(2013)Hirsch, Quintino, Bowles, and Brunner]{hirshgenuinenonloc}
F.~Hirsch, M.~T. Quintino, J.~Bowles, and N.~Brunner.
\newblock Genuine hidden quantum nonlocality.
\newblock \emph{Physical review letters}, 111 16:\penalty0 160402, 2013.

\bibitem[Hirsch et~al.(2016)Hirsch, Quintino, Bowles, Vértesi, and Brunner]{hirshsigmalocal}
F.~Hirsch, M.~T. Quintino, J.~Bowles, T.~Vértesi, and N.~Brunner.
\newblock Entanglement without hidden nonlocality.
\newblock \emph{New Journal of Physics}, 18\penalty0 (11):\penalty0 113019, nov 2016.
\newblock \doi{10.1088/1367-2630/18/11/113019}.
\newblock URL \url{https://dx.doi.org/10.1088/1367-2630/18/11/113019}.

\bibitem[Horodecki et~al.(1996)Horodecki, Horodecki, and Horodecki]{horodeckiEW}
M.~Horodecki, P.~Horodecki, and R.~Horodecki.
\newblock Separability of mixed states: necessary and sufficient conditions.
\newblock \emph{Physics Letters A}, 223\penalty0 (1):\penalty0 1--8, 1996.
\newblock ISSN 0375-9601.
\newblock \doi{https://doi.org/10.1016/S0375-9601(96)00706-2}.
\newblock URL \url{https://www.sciencedirect.com/science/article/pii/S0375960196007062}.

\bibitem[Ji et~al.(2018{\natexlab{a}})Ji, Liu, and Song]{cryptoeprint:2018/544}
Z.~Ji, Y.-K. Liu, and F.~Song.
\newblock Pseudorandom quantum states.
\newblock Cryptology ePrint Archive, Paper 2018/544, 2018{\natexlab{a}}.
\newblock URL \url{https://eprint.iacr.org/2018/544}.
\newblock \url{https://eprint.iacr.org/2018/544}.

\bibitem[Ji et~al.(2018{\natexlab{b}})Ji, Liu, and Song]{prsfirst}
Z.~Ji, Y.-K. Liu, and F.~Song.
\newblock Pseudorandom quantum states.
\newblock Cryptology ePrint Archive, Paper 2018/544, 2018{\natexlab{b}}.
\newblock URL \url{https://eprint.iacr.org/2018/544}.
\newblock \url{https://eprint.iacr.org/2018/544}.

\bibitem[Kalai et~al.(2023)Kalai, Lombardi, Vaikuntanathan, and Yang]{advantagefromNL}
Y.~Kalai, A.~Lombardi, V.~Vaikuntanathan, and L.~Yang.
\newblock Quantum advantage from any non-local game.
\newblock In \emph{Proceedings of the 55th Annual ACM Symposium on Theory of Computing}, STOC 2023, page 1617–1628, New York, NY, USA, 2023. Association for Computing Machinery.
\newblock ISBN 9781450399135.
\newblock \doi{10.1145/3564246.3585164}.
\newblock URL \url{https://doi.org/10.1145/3564246.3585164}.

\bibitem[Kretschmer(2021{\natexlab{a}})]{Kretschmer}
W.~Kretschmer.
\newblock Quantum pseudorandomness and classical complexity.
\newblock Schloss Dagstuhl - Leibniz-Zentrum für Informatik, 2021{\natexlab{a}}.
\newblock \doi{10.4230/LIPICS.TQC.2021.2}.
\newblock URL \url{https://drops.dagstuhl.de/opus/volltexte/2021/13997/}.

\bibitem[Kretschmer(2021{\natexlab{b}})]{Kretschmer2021QuantumPA}
W.~Kretschmer.
\newblock Quantum pseudorandomness and classical complexity.
\newblock \emph{ArXiv}, abs/2103.09320, 2021{\natexlab{b}}.
\newblock URL \url{https://api.semanticscholar.org/CorpusID:232257841}.

\bibitem[Lyubashevsky et~al.(2013)Lyubashevsky, Peikert, and Regev]{LWEassumption}
V.~Lyubashevsky, C.~Peikert, and O.~Regev.
\newblock On ideal lattices and learning with errors over rings.
\newblock \emph{J. ACM}, 60\penalty0 (6), nov 2013.
\newblock ISSN 0004-5411.
\newblock \doi{10.1145/2535925}.
\newblock URL \url{https://doi.org/10.1145/2535925}.

\bibitem[Mahadev(2017)]{mahadevFHE}
U.~Mahadev.
\newblock Classical homomorphic encryption for quantum circuits.
\newblock \emph{CoRR}, abs/1708.02130, 2017.
\newblock URL \url{http://arxiv.org/abs/1708.02130}.

\bibitem[Mahadev(2018)]{mahadev}
U.~Mahadev.
\newblock Classical verification of quantum computations.
\newblock \emph{2018 IEEE 59th Annual Symposium on Foundations of Computer Science (FOCS)}, pages 259--267, 2018.

\bibitem[Morimae and Yamakawa(2022)]{morimaecommitmentsnoOWF}
T.~Morimae and T.~Yamakawa.
\newblock Quantum commitments and signatures without one-way functions.
\newblock In \emph{Advances in Cryptology – CRYPTO 2022: 42nd Annual International Cryptology Conference, CRYPTO 2022, Santa Barbara, CA, USA, August 15–18, 2022, Proceedings, Part I}, page 269–295, Berlin, Heidelberg, 2022. Springer-Verlag.
\newblock ISBN 978-3-031-15801-8.
\newblock \doi{10.1007/978-3-031-15802-5_10}.
\newblock URL \url{https://doi.org/10.1007/978-3-031-15802-5_10}.

\bibitem[Natarajan and Vidick(2016)]{MIPforBQP3}
A.~Natarajan and T.~Vidick.
\newblock Robust self-testing of many-qubit states.
\newblock \emph{CoRR}, abs/1610.03574, 2016.
\newblock URL \url{http://arxiv.org/abs/1610.03574}.

\bibitem[Natarajan and Zhang(2023)]{nonlocalitytoDQC}
A.~Natarajan and T.~Zhang.
\newblock Bounding the quantum value of compiled nonlocal games: from chsh to bqp verification.
\newblock 2023.
\newblock URL \url{https://api.semanticscholar.org/CorpusID:257353407}.

\bibitem[Palazuelos(2012)]{Superactivation_MultipleCopies}
C.~Palazuelos.
\newblock Superactivation of quantum nonlocality.
\newblock \emph{Phys. Rev. Lett.}, 109:\penalty0 190401, Nov 2012.
\newblock \doi{10.1103/PhysRevLett.109.190401}.
\newblock URL \url{https://link.aps.org/doi/10.1103/PhysRevLett.109.190401}.

\bibitem[Raz and Tal(2019)]{talPH}
R.~Raz and A.~Tal.
\newblock Oracle separation of bqp and ph.
\newblock In \emph{Proceedings of the 51st Annual ACM SIGACT Symposium on Theory of Computing}, STOC 2019, page 13–23, New York, NY, USA, 2019. Association for Computing Machinery.
\newblock ISBN 9781450367059.
\newblock \doi{10.1145/3313276.3316315}.
\newblock URL \url{https://doi.org/10.1145/3313276.3316315}.

\bibitem[Regev(2005)]{Regev05}
O.~Regev.
\newblock On lattices, learning with errors, random linear codes, and cryptography.
\newblock In H.~N. Gabow and R.~Fagin, editors, \emph{STOC}, pages 84--93. ACM, 2005.
\newblock ISBN 1-58113-960-8.
\newblock URL \url{http://dblp.uni-trier.de/db/conf/stoc/stoc2005.html#Regev05}.

\bibitem[Regev(2009)]{regevLWE}
O.~Regev.
\newblock On lattices, learning with errors, random linear codes, and cryptography.
\newblock \emph{J. ACM}, 56\penalty0 (6), sep 2009.
\newblock ISSN 0004-5411.
\newblock \doi{10.1145/1568318.1568324}.
\newblock URL \url{https://doi.org/10.1145/1568318.1568324}.

\bibitem[Reichardt et~al.(2013)Reichardt, Unger, and Vazirani]{MIPforBQP1}
B.~W. Reichardt, F.~Unger, and U.~V. Vazirani.
\newblock Classical command of quantum systems.
\newblock \emph{Nat.}, 496\penalty0 (7446):\penalty0 456--460, 2013.
\newblock \doi{10.1038/nature12035}.
\newblock URL \url{https://doi.org/10.1038/nature12035}.

\bibitem[Toda(1989)]{todastheorem}
S.~Toda.
\newblock On the computational power of pp and (+)p.
\newblock \emph{30th Annual Symposium on Foundations of Computer Science}, pages 514--519, 1989.
\newblock URL \url{https://api.semanticscholar.org/CorpusID:30272970}.

\bibitem[Vidick and Zhang(2021)]{VidickZ21}
T.~Vidick and T.~Zhang.
\newblock Classical proofs of quantum knowledge.
\newblock In A.~Canteaut and F.~Standaert, editors, \emph{Advances in Cryptology - {EUROCRYPT} 2021 - 40th Annual International Conference on the Theory and Applications of Cryptographic Techniques, Zagreb, Croatia, October 17-21, 2021, Proceedings, Part {II}}, volume 12697 of \emph{Lecture Notes in Computer Science}, pages 630--660. Springer, 2021.
\newblock \doi{10.1007/978-3-030-77886-6\_22}.
\newblock URL \url{https://doi.org/10.1007/978-3-030-77886-6\_22}.

\bibitem[Werner(1989)]{wernerprojective}
R.~F. Werner.
\newblock Quantum states with einstein-podolsky-rosen correlations admitting a hidden-variable model.
\newblock \emph{Phys. Rev. A}, 40:\penalty0 4277--4281, Oct 1989.
\newblock \doi{10.1103/PhysRevA.40.4277}.
\newblock URL \url{https://link.aps.org/doi/10.1103/PhysRevA.40.4277}.

\bibitem[Winter(1999)]{win99}
A.~Winter.
\newblock Coding theorem and strong converse for quantum channels.
\newblock \emph{IEEE Transactions on Information Theory}, 45\penalty0 (7):\penalty0 2481--2485, 1999.
\newblock \doi{10.1109/18.796385}.

\bibitem[Yao(2003)]{preparesend}
A.~C. Yao.
\newblock Interactive proofs for quantum computation.
\newblock In T.~Ibaraki, N.~Katoh, and H.~Ono, editors, \emph{Algorithms and Computation, 14th International Symposium, {ISAAC} 2003, Kyoto, Japan, December 15-17, 2003, Proceedings}, volume 2906 of \emph{Lecture Notes in Computer Science}, page~1. Springer, 2003.
\newblock \doi{10.1007/978-3-540-24587-2\_1}.
\newblock URL \url{https://doi.org/10.1007/978-3-540-24587-2\_1}.

\end{thebibliography}
\bibliographystyle{abbrvnat}

\newpage

\appendix

\section{Generalization of the not-efficiently-local definition}\label{apx:generalizationofNEL}

In our proof of Theorem~\ref{thm:maininformal} we effectively take any two circuits $C^\alice$ and $C^\bob$ and some state $\rho_{AB}$ and simulate the statistics generated by them in mode (i) with two $\PostBQP$ circuits of size $\poly(|C^\alice|)$ and $\poly(|C^\bob|)$, where $|\cdot|$ denotes the size of the circuit, i.e. number of gates plus the number of qubits on which the circuit operates. This means that we can show a more general result. 

What does more general mean? 
Even when one agrees with the quantum Church-Turing hypothesis, the choice of the $\BQP$ class in the completeness part of Definition~\ref{def:nonlocality} is somewhat arbitrary. 
In the following, we define nonlocality slightly differently. 
Honest $\alice,\bob$ still apply some quantum circuit on their input and their share of $\rho_{AB}$ but their circuits are no longer limited to be of polynomial size with respect to the question size, i.e. polynomial in $\ell$.
Next, we say that a game is sound if no cheating provers with \textbf{polynomially larger circuits} can fool the verifier. Where the parameter is now the size of the honest circuit. To summarize, we still assume the quantum Church-Turing hypothesis but we don't impose any restrictions on the sizes of honest circuits. More formally 

\begin{definition}(\textbf{not-efficiently-local})\label{def:nonlocality2.0}
For a state $\rho_{AB} \in \mathcal{S}(\mathcal{H}_A \otimes \mathcal{H}_B)$ we say that $\rho_{AB}$ is not-efficiently-local if for every sufficiently small $\delta$ there exists $k \in \N$, game $\mathcal{G}(\rho_{AB},k,\cdot)$ and provers $\alice_{\text{honest}}(\cdot),\bob_{\text{honest}}(\cdot)$ such that for every polynomial $p(\cdot)$ there exists a polynomial $q(\cdot)$ such that for every $\ell \in \N$
\begin{enumerate}
    \item{(\textbf{Completeness})} If $\mathcal{G}(\rho_{AB},k,\ell)$ was run in mode (i) with $\alice_{\text{honest}}(\ell),\bob_{\text{honest}}(\ell)$ then $$\Prob[\ver \text{ accepts}] \geq 1 - \delta.$$
    \item{(\textbf{Soundness})} For \textbf{every} $\alice,\bob$ such that $|\alice| \leq p(|\alice_{\text{honest}}(\ell)|), |\bob| \leq p(|\bob_{\text{honest}}(\ell)|)$ if $\mathcal{G}(\rho_{AB},k,\ell)$ was run in mode (ii) with $\alice,\bob$ then $$\Prob[\ver \text{ accepts}] \leq 1 - \delta - \frac{1}{q(\ell)}.$$
\end{enumerate}
\end{definition}
\noindent
With this new definition, our proof directly implies that 
$$
\ENT =  \NEL \Rightarrow \BQP \neq \PP,
$$
where the $\NEL$ is understood as in Definition~\ref{def:nonlocality2.0}.

\section{RSP soundness}\label{apx:nonlocality}

In this section, we move to the most technically involved part of the corresponding chapter, where we show that every prover winning in the protocol with high probability had to prepare an eigenstate of $X,Y,$ or $Z$, and he doesn't know which one was prepared.

First, we list some technical definitions and define what the computational distinguishability of states really means.

\begin{definition}
For two families of (not necessarily normalized) density operators $\{\rho_\ell\}_{\ell \in \N}$ and $\{\sigma_\ell\}_{\ell \in \N}$ we say that they are \textit{computationally distinguishable with advantage at most} $\delta(\ell)$, if for any polynomial-time uniformly generated family of circuits $\{D_\ell\}_{\ell \in \N}$ (\textit{distinguisher}), there exists $\ell_0$ such that for all $\ell > \ell_0$ we have
$$
\frac12 \left|\trace(D_\ell^\dagger (\ket{0}\bra{0} \otimes \id) D_\ell \rho_\ell) - \trace(D_\ell^\dagger (\ket{0}\bra{0} \otimes \id) D_\ell \sigma_\ell) \right| \leq \delta(\ell).
$$
\end{definition}
Next, we define particular notions of rigidity
\begin{definition}
Let $\hilb_A, \hilb_{A'}$ be finite-dimensional Hilbert spaces. Let $R \in L(\hilb_A), S \in L(\hilb_{A'})$ be functions of $\delta$. We say that $R,S$ are $\delta$ isometric wrt $\ket{\psi} \in \hilb_A \otimes \hilb_B$, and write $R \simeq_\delta S$, if there exists an isometry $V : \hilb_A \rightarrow \hilb_{A'}$ such that $\|(R - V^\dagger S V) \otimes \id_B \ket{\psi}) \|^2 = O(\delta) $. Moreover, if $V$ is the identity then we say that $R,S$ are $\delta$-equivalent, and write $R \approx_\delta S$, i.e. if $\|(R-S)\otimes \id_B \ket{\psi}\|^2 = O(\delta)$.
\end{definition}

\subsection{Modeling}

In this section, we explain how an arbitrary prover strategy can be modeled.
Similar to what is done in \cite{vidick}, we formalize the behavior of the prover as a device. 

\begin{definition}
A device $D = (\phi, \Pi , M , Z ,X,Y,\XYm,\XYp)$ is specified as follows, \\ 
\begin{enumerate}
    \item $\phi$ is positive semidefinite in $\mathcal{P}(\mathcal{H}_D \otimes \mathcal{H}_Y)$, where $\mathcal{H}_D$ is arbitrary and $\mathcal{H}_Y$ is of the same dimension as the length of the image $y$ sent by the prover. 
    \item $\phi_y$ is the not normalized post measurement state of $D$ after sending the image $y$ and can be written as,
    \begin{align*}
        \phi_y = (\id_D \otimes \bra{y}_Y)\phi(\id_D \otimes \ket{y}_Y).
    \end{align*}
    \item $\{\Pi^{(b,x)}_y\}_{(b,x)}$ are a set of POVMs such that probability of the device returning $(b,x)$ on the preimage test is $\trace(\Pi^{(b,x)}_y \phi_y)$. 
    \item $\{M^{d}_y\}$ are a family of POVMs that classify the distribution of the parity $d$ sent by prover.
    \item Similarly, families of binary observables $Z_y = Z^0_y - Z^1_y$, $X_y = X^0_y - X^1_y$, $Y_y = Y^0_y - Y^1_y, (\XYm)_y = (\XYm)_y^0 - (\XYm)_y^1, (\XYp)_y = (\XYp)_y^0 - (\XYp)_y^1$ where $X_y$ classifies the distribution of the output $b$ when $c = X$ and $Z$'s and $Y,\XYm,\XYp$'s are defined analogously.  
\end{enumerate}

\end{definition}

\begin{definition}
Let $D = (\phi , \Pi , M,Z,X,Y,\XYm,\XYp)$ be a device and $W \in \{X,Y\}$. For $v \in \{0,1\}$ define a subnormalized density matrix as
$$
\phi_{y,W,v} = \sum_{d: \widehat{W}(d) = W, \hat{v}(d) = v} (\id_Y \otimes M_y^d) \phi_y (\id_Y \otimes M_y^d).
$$
we will sometimes omit $y$ to mean the same state. Moreover we define $\phi_X = \phi_{X,0} + \phi_{X,1}, \phi_Y = \phi_{Y,0} + \phi_{Y,1}$. 
\end{definition}
A device $D$ is called efficient if all the observables can be realized efficiently, moreover $\phi$ can be prepared efficiently. For the majority of the proofs we consider a prover~(device) that wins the preimage test with probability 1. Later we will explain how these requirements can be relaxed for a non-perfect prover. 

The first lemma we prove is an analog of Claim~7.2 from \cite{mahadev}. This lemma is a form of simplification of the modeling of the device. It is used later on when combining all the results in the proof of Theorem~\ref{thm:mainRSP}.

\begin{lemma}\label{lem:simplification}
Let $D = (\phi,\Pi , M,Z,X,Y,\XYm,\XYp)$ be an efficient device that passes the preimage test with probability $1-\epsilon$, for some $0\leq \epsilon \leq 1$. There exists a device $D' = (\phi' , \Pi , M,Z,X,Y,\XYm,\XYp)$ such that $\|\phi' - \phi\|_1 = O(\sqrt{\epsilon})$, $D'$ wins the preimage test with probability $1-\text{negl}(\ell)$ and after returning $y$, the state is of the form, 
\begin{align}\label{eq:stateform}
    \ket{\phi'_y} = \sum_{b\in\{0,1\}} \ket{b,x_b}\ket{\phi_{y,b}}, \text{ when }k\in \mathcal{K}_\mathcal{F}\\
    \ket{\phi'_y} = \ket{\hat{b},x_{\hat{b}}} \ket{\phi_{y,\hat{b}}}, \text{ when } k\in \mathcal{K}_\mathcal{G}
\end{align}
Where $(b,x_b)$ ($(\hat{b},x_{\hat{b}})$ respectively) are the preimages of $y$ under the function defined by $k$, and $\ket{\phi_{y,b}}$ are arbitrary.   
\end{lemma}
\begin{proof}
We sketch how to design $D'$ that satisfies the statement. 
Given the state $\phi$, $D'$ measures the range register and returns $y$ as $D$ would. 
After that, $D'$ evaluates a procedure that checks that the input is well-formed on the superposition to make sure he would succeed in the preimage test. 
The idea is that by preparing polynomially many states $\phi$ and doing this, the probability of $D'$ failing the preimage test is $O(\epsilon^{\poly(\ell)})$, which is negligible in $\ell$. 
Also as evaluating a check on some registers of the state would return $1$, these registers contain preimages of $y$, which means the state has the form~\ref{eq:stateform}. 

What remains is to show that the state $D'$ has after sending $y$ is $O(\sqrt{\epsilon})$ apart form $\phi$ in $L_1$ norm. Let $X$ be the positive operator corresponding to evaluating the check procedure on a part of the state and measuring the output register. As $D$ wins the preimage test with probability at least $1-\epsilon$, $\trace(\phi X) \geq 1-\epsilon$, as if $D$ wins the preimage test, the output of the $\CHK$ would also be $1$. This allows us to use the gentle measurement lemma from \cite{win99} which tells us, 
\begin{align*}
    \|\phi - \sqrt{X} \phi \sqrt{X}\|_1 \leq \sqrt{8 \epsilon}
\end{align*} 
which concludes the proof. 
\end{proof}

We continue by proving that no efficient prover can distinguish what the values of $\widehat{W}$ and $\hat{v}$ are using the information in their view.    

\begin{lemma}\label{lemma:guessingwonthelp}
Let $D$ be an efficient device succeeding with probability $1$ in the preimage test. For any $W \in \{X,Y\}$ and $v \in \{0,1\}$, no polynomial-time quantum procedure can predict $\widehat{W}(d)$ given $(y,d,\phi_{W,v})$, with advantage significantly more than $\frac12$. 
\end{lemma}
\begin{proof}
Assuming by contradiction that such a distinguisher exists we construct an adversary that violates the adaptive hardcore bit property. 

First, assume that a distinguisher $\adv$ can predict $\widehat{W}(d)$ with probability significantly larger than $\frac12$. Given the state, the adversary measures $\{\Pi_y^{(b,x_b)}\}$ to obtain $(b,x_b)$ and then measures $\{M_y^{d}\}$ to obtain $d$. Next the adversary runs $\adv$ to obtain $\widetilde{W}$ and returns $(b,x_b,d,\widetilde{W})$. Because of the collapsing property, $\widetilde{W}$ still has a noticeable advantage of being equal to $\widehat{W}(d)$. This means that $(b,x_b,d,\widetilde{W}) \in H_k^{0}$ (see Definition~\eqref{def:adaptivebit}) with probability significantly larger than $\frac12$, which contradicts the hardcore bit property~\eqref{eq:hardcorebit2}. 
\end{proof}

For the sake of convenience, we will sometimes drop the $y$ subscript from the measurement operators and the states.

\subsection{$Z$ measurement}

\begin{lemma}\label{lem:weirdsquare}
Let $D = (\phi , \Pi , M,Z,X,Y,\XYm,\XYp)$ be an efficient device that wins in the preimage test  with probability $1$ and in the $Z$-measurement test with probability $1 - \e$. Then on average over $y\in \mathcal{Y}$ we have, 
\begin{align*}
    \sum_{b,d}\trace((M^d \Pi^b - Z^b M^d)^\dagger(M^d \Pi^b - Z^b M^d) \phi ) \leq O(\e).
\end{align*}
\end{lemma}
\begin{proof}
We first argue the property for $G = 1$, i.e. injective functions case and then move to the $G = 0$, i.e. the claw free functions case. 

\paragraph{Case $G = 1$.} We define $\Pi^b = \Pi_{y}^{(b,x_b)}$, and $\Pi = \Pi^0 - \Pi^1$. As the prover wins with probability 1 in the preimage test, let $(\tilde{b},x_{\tilde{b}})$ be the preimage of $y$, we have that $\trace(\Pi^{1-\tilde{b}} \phi) = 0$ hence,
\begin{equation}\label{eq:zeroonorthogonal}
\Pi^{1-\tilde{b}} \phi \Pi^{1-\tilde{b}} = 0.
\end{equation}
Now winning in the $Z$-measurement test with probability $1-\epsilon$ means that,
$$ \sum_d\trace(Z^{\tilde{b}}M^d \phi M^d) \geq 1-\epsilon, 
$$
which by the collapsing property implies
\begin{equation}\label{eq:atleast1-eps}
\sum_{d,b}\trace(Z^{\tilde{b}}M^d \Pi^b \phi \Pi^b M^d) \geq 1-\epsilon-\text{negl}(\ell) \geq 1 - O(\e). 
\end{equation}
Because of \eqref{eq:zeroonorthogonal} we can write 
\begin{equation}\label{eq:equivalentexpression}
\sum_{d,b}\trace(Z^{\tilde{b}}M^d \Pi^b \phi \Pi^b M^d) = \sum_d\trace(Z^{\tilde{b}}M^d\Pi^{\tilde{b}}\phi \Pi^{\tilde{b}}M^d + Z^{1-\tilde{b}}M^d\Pi^{1-\tilde{b}}\phi \Pi^{1-\tilde{b}}M^d).
\end{equation}
Combining \eqref{eq:atleast1-eps} and \eqref{eq:equivalentexpression} we get
$$
\sum_{b,d}\trace(Z^b M^d \Pi^b \phi \Pi^b M^d) \geq 1 - O(\e).
$$

\paragraph{Case $G=0$.} Now because of the collapsing property again, when $G = 0$ we also have 
\begin{align}
    \sum_{b,d} \trace(Z^b M^d \Pi^b \phi \Pi^b M^d) = 1 - O(\e) - \negl(\ell) \geq 1 - O(\e). \label{eq:trace_collapsing}
\end{align}
as otherwise this would contradict indistinguishability of $G= 0$ and $G= 1$ (as all the observables are efficiently computable). We claim that due to the collapsing property, $\Pi \phi \Pi$ and $\phi$ are indistinguishable. This is proven as follows, 
\begin{align*}
    &\Pi \phi \Pi \\
    &= (\Pi^0 - \Pi^1) \phi (\Pi^0 - \Pi^1) \\
    &= \underbrace{\Pi^0 \phi \Pi^0 + \Pi^1 \phi \Pi^1}_{\text{computationally indisting. from}\phi} - \Pi^0 \phi \Pi^1 - \Pi^1 \phi \Pi^0\\ 
    &=_{\text{CID}} \phi - \Pi^0\Pi^0 \phi \Pi^0 \Pi^1 - \Pi^0\Pi^1 \phi \Pi^1 \Pi^1 + \\
    &- \Pi^1\Pi^0 \phi \Pi^0 \Pi^0 - \Pi^1\Pi^1 \phi \Pi^1 \Pi^0 && \phi =_{\text{CID}} \Pi^0 \phi \Pi^0 + \Pi^1 \phi \Pi^1\\ 
    &= \phi + 0 &&\text{Because }\Pi^b\Pi^{1-b} = 0,
\end{align*}
where we repeatedly used the collapsing property and used CID to mean computationally indistinguishable. Using the fact that $\Pi,M,Z$ are all efficiently computable observables we have, 
\begin{align}\label{eq:abs_trace}
\sum_{b}\left|\sum_{d}\trace(Z^b M^d (\phi - \Pi \phi \Pi) M^d) \right| \leq \text{negl}(\ell).
\end{align}
The following useful identity follows from a direct computation
\begin{equation}\label{eq:identity}
 \phi - \Pi \phi \Pi = 2(\Pi^0 \phi \Pi^1 + \Pi^1 \phi \Pi^0).   
\end{equation}
We can bound
\begin{align}
&2 - O(\e) \nonumber \\
&\leq 2 \cdot \sum_{b,d} \trace(Z^b M^d \Pi^b \phi \Pi^b M^d) && \text{By 
\eqref{eq:trace_collapsing}} \nonumber \\
&=2\cdot \sum_{d}\trace(Z^0M^d \Pi^0 \phi \Pi^0 M^d) + \trace(Z^1M^d \Pi^1 \phi \Pi^1 M^d) && \text{By def.} \nonumber\\ 
&= 2\cdot \sum_d \trace(Z^0 M^d (\id-\Pi^1) \phi \Pi^0 M^d) + \trace(Z^1 M^d \Pi^1 \phi (\id-\Pi^0) M^d) && \text{By def.}\nonumber\\
&= \sum_{b,d} [\trace(Z^b M^d \Pi^b\phi M^d) + \trace(Z^b M^d \phi \Pi^b M^d) + \trace(Z^b M^d (\phi - \Pi\phi\Pi) M^d)] &&\text{By } \eqref{eq:identity} \label{eq:almost2}
\end{align}
Now \eqref{eq:abs_trace} and \eqref{eq:almost2} gives
$$
\sum_{b,d} \left[\trace(Z^b M^d \Pi^b \phi M^d) + \trace(Z^b M^d \phi \Pi^b M^d) \right] \geq 2 - O(\e).  
$$
If we expand the square in the claim of the lemma and regroup the terms we get
\begin{align*}
&\sum_{b,d}\trace((M^d \Pi^b - Z^b M^d)^\dagger(M^d \Pi^b - Z^b M^d) \phi ) \\
&= \sum_{b,d} \trace( [\Pi^b M^d M^d \Pi^b - \Pi^b M^d Z^b M^d - M^d Z^b M^d \Pi^b + M^d Z^b Z^b M^d] \phi) \\
&= \sum_{b,d} \left[ \trace(\Pi^b M^d \Pi^b \phi) + \trace(M^d Z^b M^d \phi) \right] - \sum_{b,d} \left[ \trace(Z^b M^d \phi \Pi^b M^d) + \trace( Z^b M^d \Pi^b \phi M^d) \right].
\end{align*}
By \eqref{eq:almost2} the second sum is at least $2 - O(\e)$ and by definition the first sum is upper bounded by $2$. The claim follows.
\end{proof}

\subsection{$X,Y$ Measurements} 
The goal in this section is to first prove that the measurement operator $Z$, almost anti-commutes with $X$ and $Y$ operators, and moreover $X,Y$ anti-commute. This would guarantee that there exists an efficiently computable isometry $V$ which maps these observables onto the Pauli operators $\sigma_X , \sigma_Y , \sigma_Z$ (on support of $\phi$).

\begin{lemma}\label{lemma:tracesareclose}
Let $D = (\phi, \Pi , M,Z,X,Y,\XYm,\XYp)$ be an efficient device that wins in the preimage test with probability $1$ and in the $Z$ test with probability $1 - \e$. On average over $pk$, $y$,
$$
     \sum_{W \in \{X,Y\}} \sum_{b \in \{0,1\}} \left| \trace(W^0 Z^b \phi_W Z^b) - \trace(W^1 Z^b \phi_W Z^b) \right| = O(\sqrt{\e}).
$$
\end{lemma}

\begin{proof}
Note that we can express the post-measurement state of the device after the preimage test as
\begin{equation}\label{eq:stateafterpreimage}
\tilde{\phi}_{\text{YBXD}} = \sum_y \trace(\phi_y) \ket{y}\bra{y}_Y \otimes \sum_{b \in \{0,1\}} \ket{b,x_b}\bra{b,x_b}_{BX} \otimes \Pi^{(b,x_b)}_y \frac{ \phi_y}{\trace(\phi_y)} \Pi^{(b,x_b)}_y.
\end{equation}
We define this state to make the analysis easier. By definition $\phi \in \mathcal{H}_D \otimes \mathcal{H}_Y$ so it might be at first sight surprising that $\tilde{\phi}_{\text{YBXD}}$ is a state on more registers. The reason we write it like this is that according to the protocol we first measure $y$ which is in the $Y$ register, then measure $\{\Pi_y^{(b,x_b)}\}$ to obtain outcome $(b,x_b)$ that we artificially append to the state in a register we call $BX$ and the post measurement state is in the $D$ register.

For every $W \in \{X,Y\}, v \in \{0,1\}$ let
\begin{equation}
\sigma_{W,v} := \sum_y \sum_{b \in \{0,1\}} \ket{b,x_b}\bra{b,x_b}_{BX} \otimes \sum_{d : \widehat{W}(d) = W} \ket{y}\bra{y} \otimes \ket{d}\bra{d} \otimes  W^vM^d_y \Pi^{(b,x_b)}_y \phi_y \Pi^{(b,x_b)}_yM_y^d W^v,    
\end{equation}
these states can be understood informally as a result of sequentially measuring $y,\Pi,M,W$ conditioned on obtaining $v$ when measuring $W$. We will show that for every $W \in \{X,Y\}$, $\sigma_{W,0}$ and $\sigma_{W,1}$ are computationally indistinguishable.

Suppose towards contradiction that there exists $W \in \{X,Y\}$ and en efficient observable $O$ such that $\trace(O(\sigma_{W,0} - \sigma_{W,1})) \geq \mu(\ell)$ for some non-negligible function $\mu$. We will define an efficient procedure $\mathcal{A}$ that will use observable $O$ and the existence of which will show a contradiction with \eqref{eq:hardcorebit1}.

$\mathcal{A}$ first prepares $\phi$, then measures $y$, then applies $\{\Pi^{(b,x_b)}_y\}$. $\mathcal{A}$ aborts if $f_{pk}(b,x) \neq y$, i.e. $(b,x_b)$ is not the preimage of $y$. Next, $\mathcal{A}$ measures $\{M_y^d\}$, obtaining $d \in \Z^w_4$, and then it measures $\{W^v\}$ obtaining $v \in \{0,1\}$. Finally it measures $O$, obtaining $u \in \{0,1\}$, and returns $(b,x,d,W,u \oplus v)$.

Observe that after the $\Pi$ measurement $\mathcal{A}$ prepared $\tilde{\phi}_{\text{YBXD}}$ and after the $W^v$ measurement either $\sigma_{W,0}$ or $\sigma_{W,1}$ was prepared. Assumption $\trace(O(\sigma_{W,0} - \sigma_{W,1})) \geq \mu(\ell)$ is exactly a contradiction with \eqref{eq:hardcorebit1}. 

This implies that states $\sigma_{W,0}$ and $\sigma_{W,1}$ are computationally indistinguishable, so in particular
$$
\left| \sum_{d : \widehat{W}(d) = W} \sum_{b \in \{0,1\}} \trace(W^0 M^d \Pi^b \phi \Pi^b M^d) - \trace(W^1 M^d \Pi^b \phi \Pi^b M^d) \right| \leq \negl(\ell).
$$
The statement follows from Lemma~\ref{lem:weirdsquare} and the Cauchy-Schwarz inequality.
\end{proof}

\begin{lemma}\label{lemma:zanticommuteswithxy}
Let $D = (\phi, \Pi , M,Z,X,Y,\XYm,\XYp)$ be an efficient device winning in the preimage test with probability $1$ and in the $X,Y,Z$ tests with probability $1 - \e$. On average over $pk$, $y$ we have,
\begin{align}
     \sum_{W \in \{X,Y\}} \trace(\{Z,W\}^2 \phi_{W}) \leq O(\e^{1/4}).
\end{align}
\end{lemma}
\begin{proof}
Let $W \in \{X,Y\}$. $D$ succeeding with probability $1-\epsilon$ in the $X,Y$ tests means that
\begin{align}\label{eq:winningprobXY}
    \sum_{v\in\{0,1\}}\trace(W^v \phi_{W,v}) \geq 1-\epsilon. 
\end{align}

Using Lemma~\ref{lemma:tracesareclose}, we have that, 
\begin{align}\label{eq:notnormalized}
    \left|\sum_b \trace(W^0 Z^b \phi_W Z^b) - \sum_b \trace(W^1 Z^b \phi_W Z^b)\right| = O(\sqrt{\epsilon})
\end{align}
Let $\tilde{\phi}_W = \phi_W / \trace(\phi_W)$, be the normalization of the state. Now we argue, for all but a negligible number of $y$ values, the renormalization is almost uniform. This is due to Lemma~\ref{lemma:guessingwonthelp} as if the renormalization would be far from uniform, one could guess $\widehat{W}(d)$ with non-negligible advantage. This allows us to rewrite \eqref{eq:notnormalized} as, 
\begin{align}\label{eq:normalizedidentity}
    \left|\sum_b \trace(W^0 Z^b \tilde{\phi}_W Z^b) - \sum_b \trace(W^1 Z^b \tilde{\phi}_W Z^b)\right| = O(\sqrt{\epsilon}).
\end{align}

As $\sum_{b,v} \trace(W^v Z^b \tilde{\phi}_W Z^b) = 1 $ \eqref{eq:normalizedidentity} gives us that for all $v \in \{0,1\}$ we have, 
\begin{align}\label{eq:renormalizedidentity}
    \mu_{W,v} = \left|\frac12 - \sum_{b} \trace(W^v Z^b \tilde{\phi}_{W} Z^b)\right| = O(\sqrt{\epsilon})
\end{align}
Now invoking \cite[Lemma~7.2]{BCM}, with $\phi = \tilde{\phi}_W$, $M = W^{v}$ and $\Pi = Z^0$ and $\omega = 1/2 + \Omega(\epsilon^{1/4})$, and using equations~\ref{eq:renormalizedidentity} and \ref{eq:winningprobXY} we get that projection $K^v$ onto the direct sum of eigenspaces of 
$$\frac12 (ZW^vZ + W^v) = Z^0 W^v Z^0 + Z^1 W^v Z^1$$
has eigenvalues at most $\Omega(\epsilon^{1/4})$ away from $\frac12$ and satisfies 
$$\trace((\id-K^v)\tilde{\phi}_W) = O(\sqrt{\epsilon}).$$ 
Direct derivation gives us the identity
\begin{equation}\label{eq:differentexp}
\frac14 \{Z,W\}^2 = (Z^0WZ^0 + Z^1WZ^1)^2.
\end{equation}
Now to bound $\trace((Z^0WZ^0 + Z^1WZ^1)^2 \phi_{W})$, we use the following inequality. 
\begin{align}
&\trace \left( \left(Z^0W^vZ^0 + Z^1W^vZ^1 - \frac{1}{2} \id \right)^2 \phi_W \right) \nonumber \\
&= \trace \left( \left(Z^0W^vZ^0 + Z^1W^vZ^1 - \frac{1}{2} \id\right)^2 K^v\phi_W \right) \nonumber \\ \label{eq:splitink}
&+ 
\trace \left( \left(Z^0W^vZ^0 + Z^1W^vZ^1 - \frac{1}{2} \id \right)^2\left(\id-K^{v}\right) \phi_W \right) \\ 
&\leq O(\sqrt{\epsilon}) + O \left(|\trace(\id-K^v)\phi_W|^{1/2} \right) \leq O(\epsilon^{1/4}). \label{eq:boundwithk}
\end{align}
the transition from \eqref{eq:splitink} to \eqref{eq:boundwithk} is derived by bounding the first sum term using the fact that $Z^0WZ^0 + Z^1WZ^1 - \frac12 \id$ is a bounded operator and using the definition of $K^v$, and bounding the second sum term via Cauchy-Schwarz inequality and the bound on $\trace((\id-K^v)\phi)$. 

We arrive at  
$$\trace((Z^0WZ^0 + Z^1WZ^1)^2\phi_W) =O(\epsilon^{1/4})$$
which by \eqref{eq:differentexp} concludes the proof.
\end{proof}

The following is derived by a direct computation.

\begin{fact}\label{fact:anticommutator}
Let $X,Y$ be any two binary observables on $\mathcal{H}$. Then for every state $\rho \in L(\mathcal{H})$ we have
$$
\frac14 \trace( \{X,Y\}^2 \rho) = \trace((Y X^0 Y X^0 + X^1 Y X^1 Y) \rho),
$$
where $X = X^0 - X^1$ is the decomposition of $X$ along its eigenvectors. 
\end{fact}

\begin{lemma}\label{lem:Z*anticommute}
Let $D$ be as in Lemma~\ref{lemma:zanticommuteswithxy}, then on average over $pk$ and $y$, for all $W\in\{X,Y\}$ we have, 
$$
    \sum_{U \in \{X,Y\}}\trace(\{Z,W\}^2 \phi_U) \leq O(\e^{1/4}).
$$
\end{lemma}
\begin{proof}
Let $W \in \{X,Y\}$. Assume towards contradiction that $\trace(\{Z,W\}^2 \phi_{U'})$ is noticably larger than $\trace(\{Z,W\}^2 \phi_{U})$, where $\{U',U\} = \{X,Y\}$ is a renaming of $\{X,Y\}$. Now let $\tilde{\phi}_U$ ($\tilde{\phi}_{U'}$ respectively) be the result of first measuring $\{\Pi^{(b,x_b)}\}$ on $\phi$ and then measuring $\{M^d\}$. Due to the collapsing property $\sum_U \trace(\{Z,W\}^2 \tilde{\phi}_U)$ is in negligible distance of $\sum_U \trace(\{Z,W\}^2 {\phi}_U)$. We argue that $\trace(\{Z,W\}^2 \phi_U)$ can be efficiently estimated. This is because given any state $\ket{\psi}$, one can implement the following unitary, \begin{align}
        \ket{\psi}\ket{0} \rightarrow \frac{1}{\sqrt{2}} (\ket{\psi}\ket{0} + \ket{\psi}\ket{1}) \rightarrow \frac{1}{\sqrt{2}} (ZW \ket{\psi}\ket{0} + WZ\ket{\psi}\ket{1}) \label{eq:distproccess}
    \end{align}
    
    Now if one measures the second register in the Hadamard basis, the probability of obtaining a $\ket{+}$ is $\frac{1}{2} \bra{\psi}\{Z,W\}^2\ket{\psi}$.
    
We describe an adversary that wins in the adaptive hardcore bit game with high probability. 
Consider an adversary $\adv$ starting with a state $\phi$. $\adv$ measures $\{\Pi^{(b,x_b)}\}$ on $\phi$ and obtains $(b,x_b)$, and proceeds by measuring $M$ and obtains $d$. 
Then he implements the unitary described in eq.~\ref{eq:distproccess} on the state $\tilde{\phi}$. 
He measures the state in the Hadamard basis and if the outcome is $\ket{+}$ returns $U^* = U'$, and otherwise returns $U^*\xleftarrow{\$} \{X,Y\}$. 
He then outputs $(b,x_b,d,U^*)$. If $\trace(\{Z,W\}^2 \phi_{U'})$ is noticeably greater than $\trace(\{Z,W\}^2 \phi_{U})$, $\adv$ wins with non-negligible probability. 
Hence $|\trace(\{Z,W\}^2 \phi_{U}) - \trace(\{Z,W\}^2 \phi_{U'})|\leq \text{negl}(\ell)$. 
Combining this with Lemma~\ref{lemma:zanticommuteswithxy} we get the desired inequality. 
\end{proof}

In the previous lemma, we showed that $Z$ almost anticommutes with $X,Y$. The next step is to show that $X$ and $Y$ also anticommute among themselves. The first fact we use is that as $Z$ anticommutes with both $X$ and $Y$, up to an isometry we can consider these operators to act as \say{single qubit operators} built from Pauli matrices. 
The following is a well-known fact.

\begin{lemma}\label{lemma:isometry}
Let $\ket{\psi}\in \mathcal{H}_A \otimes \mathcal{H}_B$ and $Z,X,Y$ observables on $\mathcal{H}_A$ such that $\{Z,X\}^2\ket{\psi} \approx_\delta 0$ and $\{Z,Y\}^2\ket{\psi}\approx_\delta 0$. Then there exists $\delta' = O(\sqrt{\delta})$, an efficiently computable isometry $V:\mathcal{H}_A \to \Com^2 \otimes \mathcal{H}_{A'}$, and commuting Hermitians $A_X , A_Y$ on $\mathcal{H}_{A'}$ such that $A_X ^2 + A_Y ^2  = \id$ and, 
\begin{align*}
   Z \simeq_{\delta'} \sigma_Z \otimes \id, \\ 
   X \simeq_{\delta'} \sigma_X \otimes \id, \\
   Y \simeq_{\delta'} \sigma_X \otimes A_X + \sigma_Y \otimes A_Y.
\end{align*}
  Moreover, if $\{X,Y\}^2 \ket{\psi} \simeq_{\delta} 0$, we have that $Y \simeq_{\delta'} \sigma_Y \otimes \sigma_Z \otimes \id$. The fact that $V$ is efficient means that there exists a classical polynomial-time algorithm that given explicit descriptions of circuits implementing $Z,X,Y$ as input returns an explicit description of a circuit that implements $V$.
\end{lemma}

The following is a lemma about the rigidity of the QRAC test proven in \cite{vidick}.

\begin{lemma}\label{lem:qrac}
Let $\{\phi_u\}_{u \in \{1,3,5,7\}}$ and $X, Y$ be a QRAC whose success
probability is at least $(1 - \delta)OPT_Q$, 
for some $0 \leq \delta < 1$. 
Then
$$
\frac14 \sum_{u \in \{1,3,5,7\}} \trace(\{X,Y\}^2 \phi_u) = O(\delta).
$$
\end{lemma}

\begin{lemma}\label{lem:XYanticommute}
Let $D$ be as in Lemma~\ref{lemma:zanticommuteswithxy}, then on average over $pk$ and $y$, we have, 
$$
    \sum_{U \in \{X,Y\}}\trace(\{X,Y\}^2 \phi_U) \leq O(\e^{1/2}).
$$
\end{lemma}

\begin{proof}
We perform a reduction to Lemma~\ref{lem:qrac}.     
By the Jordan lemma, there exists a basis in which $\XYm$ and $\XYp$ are block diagonal.
For a $j$-th block and $U \in \{X,Y\},v \in \{0,1\}$ let $\rho_j$ be the projection of $\phi_X +\phi_Y$ onto that block and $\rho^j_{U,v}$ be the projection of $\phi_{U,v}$ onto that block, and finally let $p^j :=  \trace(\rho^j)$.
The success probability of the QRAC test from R3aB is equal to the average over $p^j$ of success probabilities of 1-qubit QRAC tests on $\rho^j_{X,0},\rho^j_{X,1},\rho^j_{Y,0},\rho^j_{Y,1}$ and $\XYm$ and $\XYp$ projected on the $j$-th block.
By the assumption on the success probability in the R3aB test, Lemma~\ref{lem:qrac} and the success probability in the R3aA test we get the statement of the lemma.
\end{proof}

\subsection{Proof of Theorem~\ref{thm:mainRSP}}

Now we are ready to prove the soundness formally, i.e. we show Theorem~\ref{thm:mainRSP}.

\begin{proof}
Let $D$ be a device that succeeds in the protocol with probability $1 - \e$. By Lemma~\ref{lem:simplification} we know that there exists a device $D'$ whose state is within $O(\e^{1/2})$ of $D$ such that $D'$ wins the preimage test with probability negligibly close to $1$. It is enough to prove the result for $D'$ by paying $O(\e^{1/2})$ in the distance between states. We know that $D'$ succeeds in and each of $X,Y,Z$ tests with probability $1-3\e$. By Lemma~\ref{lem:Z*anticommute} and \ref{lem:XYanticommute} we know that $X,Y,Z$ pairwise almost anti-commute. This, combined with Lemma~\ref{lemma:isometry} gives us that there exist $\delta = O(\e^{1/2})$ and an efficient isometry $V : \mathcal{H}_B \rightarrow \Com^2 \otimes \Com^2 \otimes \mathcal{H}_{B'}$ such that
$$
Z \approx_\delta \sigma_Z \otimes \id, \ X \approx_\delta \sigma_X \otimes \id, \ Y \approx_\delta \sigma_Y \otimes \sigma_Z \otimes \id.   
$$
The fact that $D$ succeeds in the $X$ tests implies that for every $v \in \{0,1\}$ 
$$\text{TD} \left(V \phi_{X,v} V^\dagger,\ket{+_{0\cdot \frac{\pi}{2} + v\pi}} \bra{+_{0\cdot \frac{\pi}{2} + v\pi}} \otimes \ket{\text{AUX}_{X,v}}\bra{\text{AUX}_{X,v}} \right) \leq O(\delta),$$ 
for some states $\ket{\text{AUX}_{X,v}}$.
The fact that $D$ succeeds in the $Y$ tests implies that for every $v \in \{0,1\}$ 
$$\text{TD} \left(V \phi_{Y,v} V^\dagger,\ket{+_{1 \cdot \frac{\pi}{2} + v\pi}} \bra{+_{1 \cdot \frac{\pi}{2} + v\pi}} \otimes \sigma_Z \ket{\text{AUX}_{Y,v}}\bra{\text{AUX}_{Y,v}} \sigma_Z \right) \leq O(\delta),$$ 
for some states $\ket{\text{AUX}_{Y,v}}$. 

Simple algebra gives us that
\begin{align*}
V \phi_X V^\dagger &\simeq_{O(\delta)} \frac12 \left(\ket{0}\bra{0} + \ket{1}\bra{1}\right) \otimes (\ket{\text{AUX}_{X,0}}\bra{\text{AUX}_{X,0}} + \ket{\text{AUX}_{X,1}}\bra{\text{AUX}_{X,1}}) \\
& + \frac12 \left(\ket{0}\bra{1} +\ket{1}\bra{0} \right) \otimes (\ket{\text{AUX}_{X,0}}\bra{\text{AUX}_{X,0}} - \ket{\text{AUX}_{X,1}}\bra{\text{AUX}_{X,1}})
\end{align*}
and
\begin{align*}
V \phi_Y V^\dagger &\simeq_{O(\delta)} \frac12 \left(\ket{0}\bra{0} + \ket{1}\bra{1}\right) \otimes \sigma_Z(\ket{\text{AUX}_{Y,0}}\bra{\text{AUX}_{Y,0}} + \ket{\text{AUX}_{Y,1}}\bra{\text{AUX}_{Y,1}})\sigma_Z \\
& + \frac12 \left(e^{- \frac{\pi}{2}}\ket{0}\bra{1} + e^{-i\frac{\pi}{2}}\ket{1}\bra{0} \right) \otimes \sigma_Z(\ket{\text{AUX}_{Y,0}}\bra{\text{AUX}_{Y,0}} - \ket{\text{AUX}_{Y,1}}\bra{\text{AUX}_{Y,1}})\sigma_Z
\end{align*}
Recall (Lemma~\ref{lemma:guessingwonthelp}) that $\phi_X$ and $\phi_Y$ are computationally indistinguishable. Note that the trace distance between $\left(\ket{0}\bra{1} +\ket{1}\bra{0} \right)$, $\left(e^{- \frac{\pi}{2}}\ket{0}\bra{1} + e^{-i\frac{\pi}{2}}\ket{1}\bra{0} \right)$ and $\left(\ket{0}\bra{0} + \ket{1}\bra{1}\right)$ is constant. This means that states $\ket{\text{AUX}_{X,0}}\bra{\text{AUX}_{X,0}} - \ket{\text{AUX}_{X,1}}\bra{\text{AUX}_{X,1}}$ and $\sigma_Z(\ket{\text{AUX}_{Y,0}}\bra{\text{AUX}_{Y,0}} - \ket{\text{AUX}_{Y,1}}\bra{\text{AUX}_{Y,1}})\sigma_Z$ are at most $O(\delta)$ computationally distinguishable as otherwise we could distinguish $\phi_X$ and $\phi_Y$. This implies that $\ket{\text{AUX}_{X,0}}\bra{\text{AUX}_{X,0}} \approx_{O(\delta)} \ket{\text{AUX}_{X,1}}\bra{\text{AUX}_{X,1}}$ and $\ket{\text{AUX}_{Y,0}}\bra{\text{AUX}_{Y,0}} \approx_{O(\delta)} \ket{\text{AUX}_{Y,1}}\bra{\text{AUX}_{Y,1}}$. Knowing that, in turn, implies that $\ket{\text{AUX}_{X,v}}\bra{\text{AUX}_{X,v}}$ and $\sigma_Z \ket{\text{AUX}_{Y,v}}\bra{\text{AUX}_{Y,v}} \sigma_Z$ are at most $O(\delta)$ computationally distinguishable. Thus all the auxiliary states can be replaced by a single one by incurring a $O(\delta)$ distinguishing difference. This finishes the proof for $G = 0$.

For the case $G=1$ note that the state of the prover is $O(\delta)$ close to $\sum_{b} \ket{b}\bra{b} \otimes \ket{\text{AUX}_b}\bra{\text{AUX}_b}$. The collapsing property (Definition~\ref{def:collapsing}) gives us that $\ket{\text{AUX}_b}\bra{\text{AUX}_b}$ are indistinguishable from each other and any $\ket{\text{AUX}_{X,v}}\bra{\text{AUX}_{X,v}}$ and $\sigma_Z \ket{\text{AUX}_{Y,v}}\bra{\text{AUX}_{Y,v}} \sigma_Z$.
\end{proof}

\subsection{Proof of Theorem~\ref{thm:main_formal}}

We state a standard Chernoff bound for completeness.

\begin{fact}[Chernoff Bound]
Let $X_i$ be i.i.d. Bernoulli variables with parameter $p$. Then for every $n \in \N$, every $\delta \in (0,1)$
$$
\Prob \left[ \left|\frac1n \sum_{i=1}^n X_i - p \right| \geq \delta p \right] \leq 2 e^{-\frac{\delta^2 np}{3}}.
$$
\end{fact}

\begin{proof}
First, we address completeness and then move on to soundness.

\paragraph{Completeness.}  
If $\alice$ and $\bob$ follow the strategy of $\prov$ guaranteed to exist by Theorem~\ref{thm:completeness_for_qubit_prep} then they pass all the checks apart from the R3a B check with probability $1$.
Because of \eqref{eq:state_after_fourier} when the R3a B check is performed $\prov$ replies with a bit that leads to the right distribution. 
Observe that the probability that the R3a B check is performed is negligibly close to $\frac4{15}$. This means that with probability $1- \delta/3$, $\Omega(N)$ samples for estimating the probabilities are collected. By the Chernoff bound it is enough to have a good estimate with high probability.

Moreover the following holds. 
For every state $\tau \in \{\zero,\one,\plus,\minus,\ii,\minusi \}$ the probability that $\tau$ was prepared as the post-measurement state after returning equation $d$ (see Figure~\ref{fig:final_protocol}) is at least $\frac12 \cdot \frac14 - \text{negl}(n) \geq \frac1{10}$, as the probability of $G = 0$ is $\frac12$ and the distribution over states is negligibly close to uniform by Theorem~\ref{thm:completeness_for_qubit_prep}.

Next, observe that the probability that $\alice$ ($\bob$) reaches round R3b is $\Omega(1)$. This means that, with probability $ 1- \frac{\delta}3$, over the run of the protocol for every $s,t \in \{0,\dots,5\}$, $\Omega(N)$ pairs $(\tau_s,a),(\omega_t,b)$ were collected. Note that for convenience we identified $\{0,\dots,5\}$ with $\{\zero,\one,\plus,\minus,\ii,\minusi \}$. Now we use the Chernoff bound to compute the probability of estimating $\sum_{s,t} \beta_{s,t} \ \Prob_{\alice_\bob}^\rho[a = 1, b= 1 \ | \ \tau_s , \omega_t]$ correctly. 

Fix $s,t \in \{0,\dots, 5\}$, by the Chernoff bound we have that
$$
\Prob\left[\left|\Prob_{\alice,\bob}^\rho[a = 1, b= 1 \ | \ \tau_s , \omega_t]  - \widehat{\Prob}_{\alice,\bob}^\rho[a = 1, b= 1 \ | \ \tau_s , \omega_t] \right| \leq \frac{\eta}{8\max_{s,t}|\beta_{s,t}|}\right] \geq 1 - \frac{\delta}{144}. 
$$
By the union bound this means that 
\begin{equation}\label{eq:Iapproxforcompleteness}
\Prob \left[ \left|\hat{I}_\rho - I_\rho \right| \leq \frac{\eta}8 \right] \geq 1 - \frac{\delta}{2} \geq 1 - \delta.
\end{equation}
Applying Lemma~\ref{lem:entanglementwitness} establishes completeness.

\paragraph{Soundness.} First, we argue that $\ver$ can estimate the probability of success of the qubit preparation for $\alice$ and $\bob$. Note that with probability $1 - \frac\delta4$ there are at least $\Omega(N)$ rounds in the protocol, that are equivalent to running the qubit preparation protocol for $\alice$ and $\bob$. The interaction is accepted only if all the checks are passed. By the Chernoff bound we have then that with probability $1 -\frac{\delta}2$ the probability of $\alice$ and $\bob$ succeeding in the protocol from Figure~\ref{fig:qubit_preparation_protocol} is at least $1 - \delta$. This puts us in a position to apply Theorem~\ref{thm:mainRSP}.

We want to claim that if $\alice$ and $\bob$ shared a separable state only, then they are accepted in the protocol with probability at most $\delta + \text{negl}(n)$. To this end we bound $\Prob[\hat{I}] < 0$.

For every $s \in \{0,\dots,5\}$, that corresponds to $\tau_s \in \{\ket{0},\ket{1},\ket{+},\ket{-},\ii,\minusi\}$, let $\widehat{W}_s,\hat{v}_s \in \{0,1\}$ be such that $\tau_s = \ket{+_{\widehat{W}_s \cdot \frac{\pi}2 + \hat{v}_s \cdot \pi}}$, i.e. $\widehat{W}_s,\hat{v}_s$ represent the state $\tau_s$. For $s \in \{0,\dots, 5\}$ let $\mathcal{E}_{s}$ be the event that on the $\alice$ side $\widehat{W} = \widehat{W}_s$ and $\hat{v} = \hat{v}_s$ and the interaction reaches R3b. An analogous event $\mathcal{F}_s$ is defined for $\bob$. Note that $\hat{I}_\lambda$ is an estimation of 
\begin{equation}
\sum_{s,t} \beta_{s,t} \ \Prob[\ver \leftrightarrow \alice \text{ in Fig.~\ref{fig:final_protocol} returns } (\tau_s, 1), \ver \leftrightarrow \bob \text{ in Fig.~\ref{fig:final_protocol} returns } (\omega_t, 1) \ | \ \mathcal{E}_s, \mathcal{F}_t],
\end{equation}
over $N$ samples. \textbf{By Theorem~\ref{thm:mainRSP} we can state a crucial bound}. We have that for every $s,t \in \{0,\dots,5\}$
\begin{align}
&\big|\Prob[\ver \leftrightarrow \alice \text{ in Fig.~\ref{fig:final_protocol} returns } (\tau_s, 1), \ver \leftrightarrow \bob \text{ in Fig.~\ref{fig:final_protocol} returns } (\omega_t, 1) \ | \ \mathcal{E}_s, \mathcal{F}_t]
+ \nonumber\\
&- \Prob_{\alice,\bob}^\lambda[a = 1, b = 1 \ | \ \tau_t, \omega_t] \big| \leq O(\delta^c). \label{eq:crucialtransition}
\end{align}
This is what allows us to relate the actions of computationally bounded $\alice$ and $\bob$ to the action of $\alice$ and $\bob$ that act in the semi-quantum game setup. The quantity $\Prob_{\alice,\bob}^\lambda[a = 1, b = 1 \ | \ \tau_t, \omega_t]$ is understood as any efficient ($\in \text{QPT}(n)$) POVM applied on $\tau_s$ for $\alice$ and $\omega_t$ for $\bob$ (and on shared randomness $\lambda$). But in particular this probability is of the form of the one in Lemma~\ref{lem:entanglementwitness}. Thus 
\begin{align}
& \sum_{s,t} \beta_{s,t} \ \Prob \Big[ \ver \leftrightarrow \alice \text{ in Fig.~\ref{fig:final_protocol} returns } (\tau_s, 1), \nonumber \\
&\ver \leftrightarrow \bob \text{ in Fig.~\ref{fig:final_protocol} returns } (\omega_t, 1) \ \Big| \ \mathcal{E}_s, \mathcal{F}_t \Big]
\nonumber \\
&\geq \sum_{s,t} \beta_{s,t} \ \Prob_{\alice,\bob}^\lambda[a =1 , b = 1 \ | \tau_s, \omega_t] - \max_{s,t}|\beta_s,t| \cdot O (\delta^c) \nonumber \\
&\geq \eta/4 - \max_{s,t}|\beta_s,t| \cdot O (\delta^c) && \text{By Lemma~\ref{lem:entanglementwitness} } \nonumber \\
&\geq \eta/5. \label{eq:almostend}
\end{align}
The last step is to argue that the estimated values $$\widehat{\Prob}[\ver \leftrightarrow \alice \text{ in Fig.~\ref{fig:final_protocol} returns } (\tau_s, 1), \ver \leftrightarrow \bob \text{ in Fig.~\ref{fig:final_protocol} returns } (\omega_t, 1) \ | \ \mathcal{E}_s, \mathcal{F}_t]$$ are close to the real probabilities. This is the same computation that we already performed while deriving \eqref{eq:Iapproxforcompleteness}. By this Chernoff-bound argument and \eqref{eq:almostend}, using $\eta/5 - \eta/8 \geq \eta/20$, we get that 
\begin{equation}\label{eq:finalbound}
\Prob \left[\hat{I}_\lambda \geq \eta/20 \right] \geq 1 - \delta/2.
\end{equation}
We conclude by taking the union bound over the probability of estimating the success probability correctly and \eqref{eq:finalbound}.

\end{proof}

\section{Trapdoor Claw-free Functions}\label{apx:clawfree}


\paragraph{Overview.} The RSP protocol is formally defined in Figure~\ref{fig:qubit_preparation_protocol}. 
In order to force $\prov$ to prepare one of the states $\{\zero,\one,\plus,\minus,\ii,\minusi \}$ we use claw-free functions defined in Section~\ref{sec:clawfreeoverview}. 
Intuitively these are 2-to-1 functions that are easy to compute but hard to invert. The key is to engineer the preparation of the following state on $\prov$'s side
$$
\frac{1}{\sqrt{2}}(\ket{0}\ket{x_0} + \ket{1}\ket{x_1}),
$$
where $x_0, x_1 \in \nbits$ are the two preimages under some 2-to-1 claw free function $f : \nbits \rightarrow \nbits$. 
Then we ask $\prov$ to perform the Fourier transform of the second register over $\Z_4$. 
After the transformation the first register is, depending on the $x_0+x_1 \text{ mod } 4$, one of the states $\{\plus, \minus, \ii, \minusi\}$ (up to a global phase). What is crucial is that $\prov$ doesn't know which state was prepared. 
This will follow from the properties of claw-free functions. 
To prepare the states $\{\zero, \one\}$ there is a second mode of the protocol that engineers the preparation of the following state on $\prov$'s side
$$
\ket{b}\ket{x_b},
$$
for some $b \in \{0,1\}$ and $x_b \in \nbits$ that form a preimage under some 1-to-1 function $g : \nbits \rightarrow \nbits$. 
The first register of this state can be identified with one of $\{\zero, \one\}$. This time it is crucial that $\prov$ doesn't know what $b$ is. 
To achieve this we need $g$ to be hard to invert. 
For overall security, we also need the property that it is hard to distinguish $g$ and $f$. 
Families of functions satisfying all these properties exist as we explain in more detail in Section~\ref{sec:clawfreeoverview}.

\subsection{Trapdoor claw-free functions.}\label{sec:clawfreeoverview}

As we described in the main body of the paper we need a family of functions $(\mathcal{F},\mathcal{G})$ that satisfy a list of properties. We explain these in more detail now.

An extended trapdoor claw-free family of functions $(\mathcal{F},\mathcal{G})$ is a tuple of function families, such that for every index $k \in \{0,1\}^m$, $f_k \in \mathcal{F}$ and $g_k \in \mathcal{G}$ are functions having the same domain and support, i.e. $f_k,g_k : \{0,1\}^n \rightarrow \{0,1\}^n$. 
Moreover, for every $k$, $g_k$ (and also $f_k$) is a 1-to-1 trapdoor function. 
Informally speaking, a 1-to-1 function with a trapdoor $t$ is a function such that, given $k \in \{0,1\}^m, x \in \nbits$, it is easy to compute $g_k(x)$.
Moreover, knowing $t$, for every $y \in \nbits$ it is easy to compute $g_k^{-1}(y)$ but it is hard to compute $g_k^{-1}(y)$ without the knowledge of $t$.

Properties of $\mathcal{F}$ are more delicate. For all indexes, $k$, $f_k$ is a 2-to-1 function (a preimage of every element in the image of $f_k$ consists of two elements), such that for any quantum polynomial time (QPT) adversary, given $k$, the adversary can do one of the following but not both simultaneously. 
\begin{enumerate}
    \item Return $y$ in the image of $f_k$ and $x$ in domain of $f_k$, such that $f_k(x) = y$. \label{getpreimage}
    \item Return $y$ in the image of $f_k$ and $d$ a parity such that $d \cdot (x_0+x_1) = 0$ where $x_0,x_1$ are preimages of $y$. \label{getequation}
\end{enumerate}
The final requirement for this family is that given a function $h\in \cF\cup\cG$, a QPT adversary can not distinguish between $h \in \cF$ and $h \in \cG$. 

The existence of such families was shown assuming the hardness of LWE, i.e. LWE $\not\in \BQP$ in \cite{vidick}. Finally, we want to emphasize that the construction is not complicated. For example, function evaluation consists of a few matrix-vector multiplications.

In this section, we formally define the family of claw-free functions and introduce the two main computational assumptions.

\subsection{Adaptive hardcore bit property}
We start by introducing the adaptive hardcore bit property. 
Intuitively this property guarantees that for any function $f$, no computationally bounded adversary can compute a tuple $(y,x,d)$ such that $x$ is a preimage of $y$ under $f$, and $d$ is a valid parity of the two preimages of $y$ under $f$. 

\begin{definition}\label{def:adaptivebit}
For a security parameter $\ell \in \N$, let $\mathcal{X},\mathcal{Y}$ be finite sets and $\mathcal{K}_{\mathcal{F}}$ be a finite key set. An NTCF family $\mathcal{F}$ is said to have adaptive $\Z_4$ hardcore bit property, if the following property holds (for some $\omega$ polynomially bounded in $\ell$). 

\begin{enumerate}
    \item For any $b\in \{0,1\}$, and $x\in \mathcal{X}$, there exists a set $G_{k,b,x}\in \Z_4^{\omega}$ such that $\Prob_{d\gets \Z_4^\omega}[d \not\in G_{k,b,x}]$ is negligible in $\ell$, and checking membership of $d$ in $G_{k,b,x}$, given $(k,b,x)$ is efficient given a trapdoor $t_k$.
    \item There exists an efficiently computable injective function $J:\mathcal{X} \to \Z_4^{\omega}$, efficiently invertible on its support, such that the following holds. Let $\widehat{W}:\Z_4^\omega \to \{0,1\}$ and $\hat{v}:\Z_4^\omega \to \{0,1\}$ be the unique values such that $d \cdot (J(x_0) + J(x_1)) \mod 4= \widehat{W}(d) + 2 \hat{v}(d)$, where $x_b = \mathcal{F}^{-1}(t_k,b,y)$, when $d \in G_{k,0,x_0}\cap G_{k,1,x_1}$, and $\bot$ otherwise.  
    \begin{align*}
        H_k = \{(b,x_b,d,W , v) \ | \ b\in\{0,1\}, (x_0,x_1)\in \mathcal{R}_k , (W,v) = (\widehat{W}(d),\hat{v}(d))\} \\
        \overline{H}_k = \{(b,x,d,W, v)| (b,x,d,W,1-v) \in H_k\}.
    \end{align*}
    Then for any QPT bounded adversary $\adv$,
    \begin{equation}\label{eq:hardcorebit1}
        |\Prob[\adv(k) \in H_k] - \Prob[\adv(k)\in \overline{H}_k]| \leq \negl(\ell).
    \end{equation}
    Moreover, if for $w \in \{0,1\}$ we define $H_k^{w} = \{(b,x_b,d,W) \ | \ x_b \in \mathcal{R}_k, W = \widehat{W}(d)+w \text{ mod } 2\}$, then for any QPT bounded adversary $\adv'$, 
    \begin{align}\label{eq:hardcorebit2}
        |\Prob[\adv'(k) \in H_k^{0}] - \Prob[\adv'(k) \in H_k^{1}]| \leq \negl(\ell).
    \end{align}
\end{enumerate}

\end{definition}

The first condition simply implies that there exists a small subset of possible values for $d$ such that having the trapdoor one could check the membership of a given $d$ in this set. The second condition implies that no adversary can generate a preimage and a parity. The two conditions simply imply he does not know which eigenspace he would be in and even given the eigenspace he can not figure out the eigenvalue of the state he has prepared.

The existence of such a family comes directly from \cite{vidick}, which in turn was inspired by \cite{mahadevFHE, brakerskicertifiedrandomness}. The only difference is that we require the conditions to work in $\Z_4$ instead of $\Z_8$. 

\subsection{Collapsing property}
The other requirement for the family of functions we use in this work is called the \textit{collapsing property}. This property ensures that no computationally bounded adversary can distinguish whether the function sent to them by the verifier is 2-to-1 or injective. 

\begin{definition}\label{def:collapsing}
Let $(\mathcal{F},\mathcal{G})$ be an extended NTCF family, and let $\phi = \sum_{y\in \mathcal{Y}} \ket{y}\bra{y} \otimes \phi_y$ be an efficiently preparable state. Let $\{\Pi^{b,x_b}\}$ be an efficient POVM such that $\trace(\Pi^{b,x_b}\phi_y)$ is 0 if $f_{k,b} (x_b)\neq y$ (2-to-1 function) and $g_{k,b}(x_b) \neq y$ (injective function). $(\mathcal{F},\mathcal{G})$ is called collapsing if for any QPT bounded adversary $\adv$, $\adv$ has a negligible advantage in $\ell$ to distinguish $\phi_y$ and $\phi'_y = \Pi^{(0,x_0)}\phi_y \Pi^{(0,x_0)} + \Pi^{(1,x_1)}\phi_y \Pi^{(1,x_1)}$ (the state after measuring the input register), when $k\gets Gen_\mathcal{F}(1^\ell)$, and $y$ is distributed according to $\trace(\phi_y)$.
\end{definition} 

\end{document}